\documentclass[10pt,a4paper]{article}

\usepackage[T1]{fontenc}
\usepackage[utf8]{inputenc}
\usepackage{amsmath,amssymb,amsthm,amsfonts}
\usepackage{a4wide}

\usepackage{enumitem}

\theoremstyle{plain}
\newtheorem{theorem}{Theorem}[section]
\newtheorem{proposition}[theorem]{Proposition}
\newtheorem{lemma}[theorem]{Lemma}
\newtheorem{corollary}[theorem]{Corollary}
\theoremstyle{definition}
\newtheorem{definition}[theorem]{Definition}
\newtheorem{remark}[theorem]{Remark}
\newtheorem{example}[theorem]{Example}
\newtheorem{assumption}[theorem]{Assumption}
\theoremstyle{remark}
\numberwithin{equation}{section}

\usepackage[pdfborder={0 0 0},bookmarksopen]{hyperref}
\hypersetup{%
  pdftitle = {Model Uncertainty, Recalibration, and the Emergence of Delta-Vega Hedging},
  pdfauthor = {Sebastian Herrmann, Johannes Muhle-Karbe},
  pdfkeywords = {},
  pdfpagemode = UseNone,
}

\DeclareMathOperator{\D}{D}
\DeclareMathOperator{\diag}{diag}
\DeclareMathOperator{\Trace}{Trace}

\newcommand{\FF}{\mathbb{F}}
\newcommand{\NN}{\mathbb{N}}
\newcommand{\RR}{\mathbb{R}}

\newcommand{\bfNull}{\mathbf{0}}
\newcommand{\bfa}{\mathbf{a}}
\newcommand{\bfc}{\mathbf{c}}
\newcommand{\bfe}{\mathbf{e}}
\newcommand{\bfv}{\mathbf{v}}
\newcommand{\bfw}{\mathbf{w}}
\newcommand{\bfx}{\mathbf{x}}
\newcommand{\bfz}{\mathbf{z}}

\newcommand{\bfD}{\mathbf{D}}
\newcommand{\bfG}{\mathbf{G}}
\newcommand{\bfX}{\mathbf{X}}
\newcommand{\bfZ}{\mathbf{Z}}

\newcommand{\bfupsilon}{{\boldsymbol{\upsilon}}}
\newcommand{\bfzeta}{{\boldsymbol{\zeta}}}

\newcommand{\cC}{\mathcal{C}}
\newcommand{\cF}{\mathcal{F}}

\newcommand{\cV}{\mathcal{V}}

\newcommand{\fP}{\mathfrak{P}}
\newcommand{\fY}{\mathfrak{Y}}

\newcommand{\sfC}{\mathsf{C}}
\newcommand{\sfV}{\mathsf{V}}

\newcommand{\diff}{\mathrm{d}}
\newcommand{\dd}{\,\mathrm{d}}
\newcommand{\lin}{\mathrm{lin}}

\newcommand{\1}{\mathbf{1}}

\newcommand*{\EX}[2][]{E^{#1}\left [ #2 \right ]}
\newcommand*{\wt}[1]{\widetilde{#1}}
\newcommand*{\wh}[1]{\widehat{#1}}
\newcommand*{\ol}[1]{\overline{#1}}
\newcommand*{\ul}[1]{\underline{#1}}

\newcommand{\tr}{\top}
\newcommand{\Vanna}{\frac{\partial \Delta}{\partial \Sigma}}
\newcommand{\norm}[1]{\left\lVert#1\right\rVert}
\newcommand{\enorm}[1]{\left\vert#1\right\vert}
\newcommand{\contlocmartpart}{c}

\begin{document}
\title{%
Model Uncertainty, Recalibration, and the\\
Emergence of Delta-Vega Hedging%
\footnote{The authors thank Martin Herdegen, David Hobson, Jan Kallsen, and Frank Seifried for fruitful discussions and, in particular, Martin Schweizer for pertinent remarks on the first draft. Detailed and helpful comments from two anonymous referees are also gratefully acknowledged.}
}
\date{}
\author{%
  Sebastian Herrmann%
  \thanks{
  University of Michigan, Department of Mathematics, 530 Church Street, Ann Arbor, MI 48109, USA, email
  \href{mailto:sherrma@umich.edu}{\nolinkurl{sherrma@umich.edu}}.
  Financial support by the Swiss Finance Institute is gratefully acknowledged.
  }
  \and
  Johannes Muhle-Karbe%
  \thanks{
  University of Michigan, Department of Mathematics, 530 Church Street, Ann Arbor, MI 48109, USA, email
  \href{mailto:johanmk@umich.edu}{\nolinkurl{johanmk@umich.edu}}.
  }
}
\maketitle

\begin{abstract}
We study option pricing and hedging with uncertainty about a Black--Scholes reference model which is dynamically recalibrated to the market price of a liquidly traded vanilla option. For dynamic trading in the underlying asset and this vanilla option, delta-vega hedging is asymptotically optimal in the limit for small uncertainty aversion. The corresponding indifference price corrections are determined by the disparity between the vegas, gammas, vannas, and volgas of the non-traded and the liquidly traded options.
\end{abstract}

\vspace{0.5em}

{\small

\vspace{0.25em}

\vspace{0.25em}
}

\section{Introduction}
\label{sec:introduction}

In the context of hedging exotic derivatives, the way mathematical models for financial markets are applied in practice is often inconsistent with the assumptions these models are based on and the way they are analysed in academic research. Classical models prescribe the \emph{stochastic} behaviour of certain financial \emph{variables}, e.g., asset prices or interest rates, in terms of \emph{deterministic} input quantities, the model's \emph{parameters}. In practice, however, these deterministic parameters are often not at all treated as deterministic: exotic derivatives traders \emph{recalibrate} the parameters frequently to the observed market prices of liquidly traded vanilla options and use these options to \emph{neutralise the sensitivities of their positions against changes in these parameters} (appropriately called \emph{out-of-model hedging} by Rebonato \cite{Rebonato2004}).

The benchmark Black--Scholes model is typically recalibrated by dynamic updating of the volatility parameter (that the model assumes constant) to the market price of a liquidly traded plain-vanilla option. \emph{Vega hedging}\footnote{\emph{Vega} is the sensitivity of the Black--Scholes price with respect to changes in the volatility parameter.} then corresponds to neutralising the sensitivity of the trader's total position with respect to changes in the volatility parameter. The logical inconsistency of this practice is succinctly summarised by Rebonato \cite[Section 1.3.2]{Rebonato2004}, for example:\footnote{Davis~\cite[Section~2.~(b)]{Davis2004}, Musiela and Rutkowski~\cite[Section 7.1.8]{MusielaRutkowski2005}, and Wilmott~\cite[Section 7.10.5]{Wilmott2006} raise the same concern.}
\begin{quote}
``Needless to say, out-of-model hedging is on conceptually rather shaky ground: if the volatility is deterministic and perfectly known, as many models used to arrive at the price assume it to be, there would be no need to undertake vega hedging. Furthermore, calculating the vega statistics means estimating the dependence on changes in volatility of a price that has been arrived at assuming the self-same volatility to be both deterministic and perfectly known. Despite these logical problems, the adoption of out-of-model hedging in general, and of vega hedging in particular, is universal in the complex-derivatives trading community.''
\end{quote}

The present paper provides a consistent justification for the use of Black--Scholes vega by acknowledging from the start that the true dynamics of the underlying are not known with certainty. We suppose that models are deemed more or less plausible depending on their ``distance'' from a reference Black--Scholes model for the underlying. A new feature is that the volatility parameter of the reference model is dynamically recalibrated to the observed prices of a liquidly traded vanilla option. In the limit for small aversion against this model uncertainty, delta-vega hedging then emerges naturally. 

\paragraph{Hedging problem.}
Consider an agent who has sold a non-traded option on a stock $S$ with payoff $\sfV(S_T)$\footnote{For simplicity, we restrict ourselves to vanilla options in this introduction. Our main result, Theorem~\ref{thm:main result}, is also applicable to a wide range of exotic options like barrier options, lookback options, Asian options, forward-start options, and options on the realised variance of the stock.} and has access to three liquidly traded securities to hedge her exposure: the stock $S$, a vanilla option $C$ on the stock (hereafter named ``call''), and a bank account with zero interest rate. In practice, the market price of the call is typically quoted in terms of its \emph{(Black--Scholes) implied volatility}. That is, instead of the market price $C_t$, traders quote the unique $\Sigma_t > 0$ such that 
\begin{align}
\label{eqn:intro:C}
C_t
&= \cC(t,S_t,\Sigma_t),
\end{align}
where $\cC(t,S,\Sigma)$ is the Black--Scholes price of the call corresponding to the volatility parameter $\Sigma$. Whence, instead of modelling the dynamics of the stock and call prices, one can equivalently describe the dynamics of the stock price and the implied volatility of the call, and define the call price via \eqref{eqn:intro:C}.

If the stock and the call are traded using a self-financing strategy $\bfupsilon = (\theta,\phi)$, the corresponding \emph{Profit\&Loss (P\&L) process} $Y^{\bfupsilon}$ has the following dynamics:
\begin{align*}
\diff Y^{\bfupsilon}_t
&= \theta_t \dd S_t + \phi_t \dd C_t - \diff\cV(t,S_t,\Sigma_t).
\end{align*}
Here, $\cV(t,S_t,\Sigma_t)$ is the Black--Scholes price of the non-traded option $\sfV$ evaluated at the implied volatility $\Sigma_t$ backed out from the price of the call at time $t$. That is, in line with industry practice, the non-traded option is ``marked to model'', whereas the liquidly traded stock and call are ``marked to market''. However, at maturity $T$ of the non-traded option, $\cV(T,S_T,\Sigma_T) = \sfV(S_T)$ is the option payoff so that $Y^{\bfupsilon}_T$ coincides with the agent's actual terminal P\&L.

We assume that the agent is uncertain about the dynamics of the stock and the call. To wit, she considers \emph{all} probability measures $P$ under which the dynamics of $(S,\Sigma)$ are governed by\footnote{As is customary in asymptotic analysis, the powers of the processes $\sigma^P$, $\nu^P$, $\eta^P$, and $\xi^P$ in the dynamics of $(S,\Sigma)$ are chosen so that all of them have a nontrivial effect on the leading-order term in the asymptotic expansions below. In contrast, using the uncorrelated volatility of implied volatility $\sqrt{\xi^P}$ instead of the uncorrelated \emph{squared} volatility $\xi^P$ would only generate a higher-order effect. This is an artefact of the Black--Scholes model: for any reference model with a \emph{nonzero} uncorrelated volatility of implied volatility, $\sqrt{\xi^P}$ would be the natural parametrisation; cf.~Remark~\ref{rem:leading order controls} for more details.}
\begin{align}
\label{eqn:intro:dynamics}
\begin{split}
\diff S_t
&= S_t \sigma^P_t \dd W^0_t,\\
\diff \Sigma_t
&= \nu^P_t \dd t + \eta^P_t \dd W^0_t + \sqrt{\xi^P_t} \dd W^1_t,
\end{split}
\end{align}
for a Brownian motion $(W^0,W^1)$ in $\RR^2$ and a process $\bfzeta^P = (\nu^P,\sigma^P,\eta^P,\xi^P)$ satisfying\footnote{Here, the partial derivatives $\cC_\Sigma$, $\cC_{SS}$, $\cC_{S\Sigma}$, and $\cC_{\Sigma\Sigma}$ of $\cC$ are evaluated in $(t,S_t,\Sigma_t)$.}
\begin{align}
\label{eqn:intro:drift condition}
\nu^P_t\cC_{\Sigma} + \frac{1}{2}S_t^2 \cC_{SS}((\sigma^P_t)^2 - \Sigma_t^2) + \sigma^P_t \eta^P_t  S_t\cC_{S\Sigma} + \frac{1}{2}((\eta^P_t)^2 + \xi^P_t) \cC_{\Sigma\Sigma}
&= 0.
\end{align}
The \emph{drift condition} \eqref{eqn:intro:drift condition} ensures that the call price process $C$ is a local $P$-martingale.\footnote{\label{ftn:local martingales}The local martingale property of the liquidly traded assets is sufficient to exclude arbitrage opportunities. It also ensures that the agent has no incentive to \emph{invest} in the market but only uses it as a hedging instrument for the non-traded option; cf.~Remark~\ref{rem:drift}.}
Note that the Black--Scholes model corresponds to $P^0$ with $\bfzeta^{P^0} = \bfzeta^0(\Sigma) := (0,\Sigma,0,0)$, i.e., the implied volatility is constant and coincides with the spot volatility.

We assume that the agent has moderate risk and uncertainty aversion.\footnote{In contrast, most of the literature on hedging under model uncertainty studies variants of the \emph{uncertain volatility model} introduced by Avellaneda, Levy, and Par\'as \cite{AvellanedaLevyParas1995} and Lyons \cite{Lyons1995}. These and many more recent studies (e.g., \cite{Frey2000, DenisMartini2006, NeufeldNutz2013, PossamaiRoyerTouzi2013, BiaginiBouchardKardarasNutz2015, Nutz2014}) look for hedging strategies that \emph{dominate} the payoff of the non-traded option \emph{almost surely} for \emph{every} model of a prespecified class. This \emph{worst-case approach} corresponds to preferences with infinite risk and uncertainty aversion.}
Concerning risk aversion, we assume that in any given model, the agent seeks to maximise the expected utility from her terminal P\&L. Concerning uncertainty aversion, we suppose that she takes models less seriously the more they deviate from the reference Black--Scholes model. In the spirit of the \emph{variational preferences} of Maccheroni, Marinacci, and Rustichini \cite{MaccheroniMarinacciRustichini2006} and the \emph{multiplier preferences} of Hansen and Sargent \cite{HansenSargent2001},\footnote{We refer to \cite[Section~1]{HerrmannMuhleKarbeSeifried2017} for more details on these preferences and their relation to the standard expected utility framework as well as the worst-case approach.} this leads to the following stochastic differential game (SDG):\footnote{Our analysis also applies to somewhat more general penalty terms; cf.~\eqref{eqn:penalty function}--\eqref{eqn:objective}. The inclusion of the term $U'(Y^\bfupsilon_t)$ is not crucial but has some appealing properties. For instance, it renders the preferences invariant under affine transformations of the utility function; cf.~Remark~\ref{rem:marginal utility} for more details.}
\begin{align}
\label{eqn:intro:hedging problem}
v(\psi)
&= \sup_{\bfupsilon = (\theta,\phi)} \inf_P \EX[P]{U(Y^{\bfupsilon}_T) +  \frac{1}{2\psi} \int_0^T U'(Y^{\bfupsilon}_t) \enorm{\bfzeta^P_t - \bfzeta^{P^0}_t}^2 \dd t}.
\end{align}
Here, $\psi > 0$, $U$ is a utility function, the supremum runs over a suitable class of trading strategies, and the infimum is taken with respect to a suitable class of probability measures satisfying \eqref{eqn:intro:dynamics}--\eqref{eqn:intro:drift condition}. One interpretation is that the agent plays a game against a fictitious adversary (a ``malevolent nature'') who controls the true dynamics of the liquidly traded assets. However, ``extreme'' choices of this adversary are penalised by the positive second term in \eqref{eqn:intro:hedging problem}: the more the chosen model $P$ deviates from the reference Black--Scholes model $P^0$, the higher the penalty for the adversary. The scaling factor $\psi > 0$ measures the magnitude of the agent's uncertainty aversion: small values of $\psi$ lead to high penalties even for small deviations from the Black--Scholes reference model, which means that alternative models are taken less seriously. Note that as $\bfzeta^{P^0}_t = (0,\Sigma_t,0,0)$, the reference Black--Scholes model reflects the belief that ``the future implied volatility stays at the \emph{currently} observed level.'' Put differently, the reference Black--Scholes model is \emph{dynamically recalibrated} to the quoted option prices.

A related hedging problem without a liquidly traded call is studied in \cite{HerrmannMuhleKarbeSeifried2017} for a local volatility reference model. There, the fictitious adversary chooses the true spot volatility of the stock, but is penalised according to its distance from the reference local volatility.

\paragraph{Asymptotics.}
To obtain explicit formulas, we pass to the limit where uncertainty aversion $\psi$ tends to zero.\footnote{Asymptotic analyses of the uncertain volatility model have been carried out by \cite{Lyons1995, AhnMuniSwindle1997, AhnMuniSwindle1999, FouqueRen2014}.} That is, we consider the hedging problem \eqref{eqn:intro:hedging problem} as a small perturbation of the classical hedging problem in the Black--Scholes model and look for hedging strategies and price corrections that take into account the impact of model uncertainty in an asymptotically optimal manner. Our main result, Theorem~\ref{thm:main result}, describes a hedging strategy $\bfupsilon^\star = (\theta^\star, \phi^\star)$, a family of models $(P^\psi)_\psi$, and $\wt w_0 \geq 0$ such that, as $\psi \downarrow 0$:
\begin{align}
\label{eqn:intro:value expansion}
\begin{split}
v(\psi)
&= U(Y_0) - U'(Y_0) \wt w_0 \psi + o(\psi)\\
&= \EX[P^\psi]{U(Y^{\bfupsilon^\star}_T) +  \frac{1}{2\psi} \int_0^T U'(Y^{\bfupsilon^\star}_t) \enorm{\bfzeta^{P^\psi}_t - \bfzeta^{P^0}_t}^2 \dd t} + o(\psi).
\end{split}
\end{align}
The first line in \eqref{eqn:intro:value expansion} is a first-order expansion of the optimal value of the hedging problem for small values of the uncertainty aversion parameter $\psi$. The second line shows that the family $(\bfupsilon^\star, P^\psi)_\psi$ attains this optimal value at the leading order $O(\psi)$.\footnote{A \emph{second-order} expansion and a \emph{next-to-leading order} optimal strategy are obtained in \cite[Theorem~3.4]{HerrmannMuhleKarbeSeifried2017}, where only the stock but no additional vanilla option is used for dynamic hedging.} More precisely, Theorem~\ref{thm:main result} shows that $(\bfupsilon^\star,P^\psi)_\psi$ is in fact an asymptotic saddle point for the family of SDGs \eqref{eqn:intro:hedging problem}, i.e., the performance of the strategy $\bfupsilon^\star$ is optimal at the leading order $O(\psi)$ and $(P^\psi)_\psi$ is a family of leading-order optimal choices for the fictitious adversary. The ask price at which the agent is indifferent between keeping a flat position and selling the option $\sfV$ has the expansion
\begin{align*}
p_a(\psi)
&= V_0 + \wt w_0 \psi + o(\psi),
\end{align*}
where $V_0$ is the Black--Scholes price of the option $\sfV$ at time $0$, evaluated with volatility $\Sigma_0$. Thus, $\wt w_0 \psi$ is the leading-order premium that the agent demands as a compensation for exposing herself to model uncertainty. Accordingly, $\wt w_0$ measures the option's susceptibility to model misspecification and we call it the \emph{cash equivalent (of small uncertainty aversion)}. We next display and discuss explicit formulas for the hedging strategy $\bfupsilon^\star$, the family of models $(P^\psi)_\psi$, and the cash equivalent~$\wt w_0$.

The hedging strategy $\bfupsilon^\star = (\theta^\star,\phi^\star)$ is the \emph{delta-vega hedge} for the option $\sfV$:
\begin{align*}
\theta^\star_t
&= \cV_S(t,S_t,\Sigma_t) - \phi^\star_t \cC_S(t,S_t,\Sigma_t),\quad
\phi^\star_t
= \frac{\cV_\Sigma}{\cC_\Sigma}(t,S_t,\Sigma_t).
\end{align*}
To wit, the number of calls $\phi^\star$ is chosen so that the \emph{net vega} of the agent's position, $\phi^\star \cC_\Sigma - \cV_\Sigma$, vanishes. This leaves the agent with a \emph{net delta}\footnote{\emph{Delta} is the sensitivity of a Black--Scholes option value with respect to changes in the price of the underlying.} of $-\cV_S + \phi^\star \cC_S$ which is in turn neutralised by holding $\theta^\star$ shares of the underlying, so that the total portfolio is both \emph{delta-} and \emph{vega-neutral}.\footnote{The vega of the underlying is obviously zero.} We emphasise that the leading-order optimality of the delta-vega hedge is independent of both the agent's utility function and her uncertainty aversion parameter $\psi$. While it is important that the agent \emph{is} risk-averse (otherwise, there would be no need to hedge at all in any given model) and \emph{is} moderately uncertainty-averse in our sense (vega hedging is redundant without uncertainty aversion), the precise configuration of the agent's preferences is by and large irrelevant. Moreover, note that the delta-vega hedge is computed with the currently observed implied volatility $\Sigma_t$ of the liquidly traded call, i.e., the Black--Scholes model used to compute the hedge is dynamically recalibrated.

We next address the asymptotically optimal models $(P^\psi)_\psi$. The process $\bfzeta^{P^\psi}$
describing the model $P^\psi$ satisfies
\begin{align*}
\bfzeta^{P^\psi}_t
&= \bfzeta^{P^0}_t + \wt\bfzeta(t,S_t,\Sigma_t)\psi + o(\psi)
\end{align*}
for some $\wt\bfzeta = \wt\bfzeta(t,S,\Sigma)$ arising from a linearly constrained quadratic programming problem derived from the Hamilton--Jacobi--Bellman--Isaacs (HJBI) equation associated to the SDG \eqref{eqn:intro:hedging problem} (the constraints originate from the drift condition \eqref{eqn:intro:drift condition} and the restriction that the uncorrelated \emph{squared} volatility of implied volatility is nonnegative). The model $P^\psi$ is a perturbation of the Black--Scholes model $P^0$, parametrised by the four processes $\nu^P$, $\sigma^P$, $\eta^P$, and $\xi^P$ in \eqref{eqn:intro:dynamics}. The explicit formula for $\wt\bfzeta$ (cf.~\eqref{eqn:candidate control:first order}) shows that the asymptotically optimal perturbation exploits the disparity between the \emph{vegas}, \emph{gammas}, \emph{vannas}, and \emph{volgas}\footnote{\emph{Gamma}, \emph{vanna}, and \emph{volga} are the second-order partial derivatives $\partial^2/\partial S^2$, $\partial^2/(\partial S\partial\Sigma)$, and $\partial^2/\partial \Sigma^2$ of the Black--Scholes value of an option.} of the non-traded option $\sfV$ and the liquidly traded call while preserving the drift condition \eqref{eqn:intro:drift condition} and the restriction $\xi^P \geq 0$. In fact, if each of these greeks has the same value for both the non-traded option and the liquidly traded call (e.g., if $\sfV$ is a put with the same maturity and strike as the call), then the leading-order optimal perturbation $\wt\bfzeta$ is zero.

Finally, we discuss the structure of the expansion \eqref{eqn:intro:value expansion} and the cash equivalent $\wt w_0$. As the Black--Scholes model is complete and the traded assets are local martingales, the zeroth-order term in the expansion \eqref{eqn:intro:value expansion} of $v(\psi)$ simply is the utility $U(Y_0)$ generated by the initial P\&L. The first-order correction term $-U'(Y_0)\wt w_0\psi$ is nonpositive and describes the impact of model uncertainty for small uncertainty aversion. The cash equivalent $\wt w_0$ is determined by a linear second-order parabolic partial differential equation (PDE) with a source term. It has the following probabilistic representation:
\begin{align}
\wt w_0
&= \frac{1}{2}\EX[P^0]{\int_0^T \wt g(t,S_t,\Sigma_0) \dd t},\quad\text{with}\notag\\
\label{eqn:intro:source term}
\wt g(t,S,\Sigma)
&= -\Sigma (\underbrace{\phi^\star S^2\cC_{SS} - S^2 \cV_{SS}}_{\text{net cash gamma}})\wt\sigma
-\Sigma (\underbrace{\phi^\star S \cC_{S\Sigma} - S\cV_{S\Sigma}}_{\text{net cash vanna}})\wt\eta
-\frac{1}{2}(\underbrace{\phi^\star \cC_{\Sigma\Sigma} - \cV_{\Sigma\Sigma}}_{\text{net volga}})\wt\xi\geq 0.
\end{align}
Here, $\wt\bfzeta(t,S,\Sigma) = (\wt\nu,\wt\sigma,\wt\eta,\wt\xi)(t,S,\Sigma)$ and all functions on the right-hand side of \eqref{eqn:intro:source term} are evaluated in $(t,S,\Sigma)$. The cash equivalent $\wt w_0$ is thus determined by the expected \emph{net cash gamma}, \emph{net cash vanna}, and \emph{net volga} of the delta-vega hedged position that is accumulated over the lifetime of the option $\sfV$.\footnote{In contrast, if there is no liquidly traded call available as a hedging instrument, then the option's cash gamma is the only greek that appears in the probabilistic representation of the cash equivalent \cite{HerrmannMuhleKarbeSeifried2017}.} These three net (cash) greeks are weighted by the leading-order optimal perturbation of the spot volatility, the correlated volatility of implied volatility, and the uncorrelated squared volatility of implied volatility, respectively. The larger $\wt g$, the larger the cash equivalent $\wt w_0$. In particular, a short net gamma position (after vega hedging) is exposed to high spot volatility (positive $\wt\sigma$), a short net vanna position is exposed to volatility of implied volatility that is positively correlated with the underlying (positive $\wt\eta$), and a short net volga position is exposed to volatility of implied volatility (positive $\wt\xi$).\footnote{According to formula \eqref{eqn:intro:source term}, a short net volga position is only exposed to the part of the volatility of implied volatility that is \emph{uncorrelated} with the underlying. However, it can be seen from the proof that the correlated volatility of implied volatility has the same effect, albeit only at the order $O(\psi^2)$.} Conversely, long positions in net gamma or net vanna have the reverse exposures, but a long net volga position is \emph{not} exposed to volatility of implied volatility because $\wt\xi$ cannot be negative.

\paragraph{Techniques.}
The HJBI equation associated to the SDG~\eqref{eqn:intro:hedging problem} involves a pointwise min-max problem for the hedging strategy of the agent and for the control variables of the fictitious adversary. This min-max problem has a nonlinear equality constraint and an inequality constraint that originate from the drift condition \eqref{eqn:intro:drift condition} and the restriction $\xi^P \geq 0$, respectively.

Formally passing to the limit as $\psi \downarrow 0$, this problem can be approximated by a \emph{linearly} constrained \emph{quadratic} minimisation problem and an unconstrained quadratic maximisation problem. Both of these problems can be solved explicitly and give rise to the delta-vega hedge and candidate controls $(\bfzeta^\psi)_\psi$ corresponding approximately to the family of models $(P^\psi)_\psi$. Plugging these candidates back into the HJBI equation yields a PDE for the first-order term in the expansion of the value function of the SDG.

The rigorous verification of the (asymptotic) optimality of these candidates combines an asymptotic analysis of the HJBI equation with classical verification arguments for SDGs. It is divided into a purely analytic and a probabilistic part. Due to the constraints in the min-max problem, both parts of the proof require substantially different approaches compared to those used in \cite{HerrmannMuhleKarbeSeifried2017}. The analytic part uses careful direct estimates and Lagrange duality theory for constrained optimisation problems to show that the candidate value function is asymptotically (in a suitable sense) a solution to the HJBI equation. The probabilistic part of the proof adapts classical verification arguments for SDGs to the asymptotic setting. New difficulties arise now from the fact that the candidate controls of the fictitious adversary do not satisfy the drift condition~\eqref{eqn:intro:drift condition} exactly (as the nonlinear constraint is only approximated by a linear one).

\paragraph{Related literature.}
Let us now put our results in context by discussing some of the extant literature on the hedging of exotics using vanilla options. One strand of literature postulates that both the asset price and its spot volatility are stochastic and follow given dynamics driven by two Brownian motions. Stochastic volatility models of this type can typically be completed by using a single plain-vanilla option as an extra hedging instrument in addition to the underlying stock.\footnote{See, e.g., \cite{RomanoTouzi1997,Davis2004,DavisObloj2008} for precise conditions.} In Markovian settings, replicating strategies can then be determined in close analogy to the classical Black--Scholes argument. This leads to the so-called ``delta-sigma hedge'' \cite{HullWhite1987,Scott1991}, which neutralises the portfolio's sensitivity with respect to changes in both the underlying stock price and the spot volatility. This strategy is related to the delta-vega hedge in that it also makes use of the derivative of the option price with respect to ``volatility''. Here, however, ``volatility'' refers to the spot volatility that can (at least in theory) be backed out from the realised variance of the stock. Instead, delta-vega hedging neutralises a portfolio's sensitivity with respect to changes in the (Black--Scholes) implied volatility that is deduced from the market price of a liquidly traded option. While the spot volatility gives the \emph{instantaneous} volatility of the stock price, the implied volatility is rather an estimate for the future volatility realised over the whole time \emph{interval} ranging from today to the maturity of the liquidly traded option. Moreover, in practice, also stochastic volatility models have to be recalibrated once the model and market prices of liquidly traded options diverge.

Another strand of literature studies the \emph{robust} hedging of exotic derivatives. To wit, these studies look for hedging strategies that work in some large class of models (e.g., any continuous martingale model). The hedging strategies are typically of \emph{semi-static} form: they allow a static position in a portfolio of calls and puts (often for one maturity and all strikes) and dynamic trading in the underlying.\footnote{Semi-static hedging problems have also been analysed numerically in the context of the \emph{Lagrangian uncertain volatility model} \cite{AvellanedaParas1996, AvellanedaBuff1999}.}
For variance swaps, this leads to a robust replicating strategy \cite{Neuberger1994}, whereas robust sub- and superhedging strategies have been determined for various other exotic options (cf., e.g., \cite{Hobson1998.lookback,BrownHobsonRogers2001,CoxObloj2011.notouch,CoxObloj2011.touch,CarrLee2010, HobsonNeuberger2012,HobsonKlimmek2012,HobsonKlimmek2015}). In these studies, the goal is to find portfolios that sub- or superreplicate the exotic option in each possible scenario.\footnote{General superhedging duality results in the semi-static context have been obtained, among others, by \cite{AcciaioBeiglbockPenknerSchachermayer2016, BeiglbockHenryLaborderePenkner2013, DolinskySoner2014,  GalichonHenryLabordereTouzi2014, BouchardNutz2015}; see also the references therein.}
The underlying preferences therefore correspond to infinite aversion both against risk in a given model and uncertainty about the model itself. In contrast, as in \cite{HerrmannMuhleKarbeSeifried2017}, we consider a more moderate attitude towards risk and uncertainty that interpolates smoothly between the worst-case approach and the classical setting with one fixed model. The other major difference is that we allow \emph{dynamic} trading in a \emph{single} vanilla option instead of static positions in puts and calls of many strikes.

In practice, even the most liquid at-the-money options have substantially larger bid-ask spreads than the underlying stocks. As a result, a direct implementation of the delta-vega hedge with, e.g., daily rebalancing leads to substantial transaction costs and is found to be inferior to semi-static alternatives in several case studies~\cite{CoxObloj2011.touch,OblojUlmer2012}. As a remedy, the delta-vega hedge needs to be implemented with a suitable ``buffer''. That is, rebalancing trades should only take place once the hedge portfolio deviates sufficiently from its frictionless target. The corresponding trading boundaries for Black--Scholes delta-hedging strategies have been determined explicitly in the small-cost limit by \cite{WhalleyWilmott1997}; cf.~also \cite{KallsenMuhleKarbe2015} and the references therein for extensions to more general settings. Extending these tracking results to more general target strategies involving liquid vanilla options is a major challenge for future research. To date, the only result of this kind concerns the dynamic trading of options to reduce transaction costs~\cite{GoodmanOstrov2011}, which leads to a buffered version of the delta-gamma hedge.

\paragraph{Organisation of the paper.} The remainder of the article is organised as follows. The mathematical framework for the hedging problem under model uncertainty is introduced in Section~\ref{sec:problem formulation}. Section~\ref{sec:heuristics} outlines the heuristic derivation of the asymptotically optimal solution. Our main results are stated and discussed in Section~\ref{sec:main results}. Finally, all proofs are relegated to Section~\ref{sec:proofs}.

\paragraph{Notation.}
Vectors $\bfa \in \RR^n$ and vector-valued functions are printed in boldface type. The transpose of a vector $\bfa$ is denoted by $\bfa^\tr$ and its Euclidean norm by $\enorm{\bfa}$. For the sake of readability, we mostly suppress the arguments of functions in the notation. In calculations and estimates, we typically display the arguments only on the left-most side of (in-)equalities; the omitted arguments should then be clear from the context. Partial derivatives of functions with respect to scalar variables are denoted by subscripts as in \eqref{eqn:intro:drift condition} and $\D_\bfzeta H$ denotes the gradient of a function $H(\ldots;\bfzeta)$ with respect to the vector variable $\bfzeta$.

\section{Problem formulation}
\label{sec:problem formulation}

To allow for dynamic trading in both the stock and an option on the stock, we consider \emph{market models} for the joint evolution of both assets. Instead of prescribing the dynamics of the option, we follow Sch\"onbucher's approach \cite{Schonbucher1999} and model its Black--Scholes implied volatility.\footnote{Other early articles on risk-neutral dynamics for stochastic implied volatility models include \cite{Lyons1997, BraceGoldysKlebanerWomersley2001, LedoitSantaClaraYan2002}. For more recent developments on arbitrage-free market models for (parts of or the whole) option price surface, we refer the reader to \cite{SchweizerWissel2008.existence, SchweizerWissel2008.multistrike, CarmonaNadtochiy2009, JacodProtter2010, CarmonaNadtochiy2011, CarmonaNadtochiy2012,KallsenKruhner2015} and the references therein.}
This approach is outlined in Section~\ref{sec:market models} and motivates the precise setup introduced in Section~\ref{sec:setup}. The hedging problem is in turn formulated in Section~\ref{sec:hedging problem}.

\subsection{Market models for the underlying and its implied volatility}
\label{sec:market models}

We consider a financial market with three liquidly traded securities: a stock $S$, an option written on the stock, and a bank account with zero interest rate. The liquidly traded option has a payoff of the form $\sfC(S_{T_\sfC})$ at maturity ${T_\sfC}$. To avoid  confusion with the non-traded option introduced later, this liquidly traded option will be named ``call'' hereafter. It is market practice to quote option prices in terms of their \emph{(Black--Scholes) implied volatilities}. That is, traders do not quote the market price $p$ of the call $C$, but instead the unique $\Sigma > 0$ such that $p = \cC(t,S_t,\Sigma)$, where $\cC(\cdot,\cdot,\Sigma)$ is the solution of the Black--Scholes PDE
\begin{align}
\label{eqn:C:PDE}
\begin{split}
\cC_t(t,S,\Sigma) + \frac{1}{2}\Sigma^2 S^2 \cC_{SS}(t,S,\Sigma)
&= 0, \quad (t,S) \in (0,{T_\sfC})\times\RR_+,\\
\cC({T_\sfC},S,\Sigma)
&= \sfC(S),\quad S \in \RR_+,
\end{split}
\end{align}
corresponding to volatility $\Sigma$, maturity ${T_\sfC}$, and the terminal payoff $\sfC(S_{T_\sfC})$ of the call. Following this practice and Sch\"onbucher's approach \cite{Schonbucher1999}, we model the implied volatility rather than the price process of the call. To wit, we assume that the joint dynamics of the stock $S$ and the call's implied volatility $\Sigma$ are governed by
\begin{align}
\label{eqn:S:dynamics}
\diff S_t
&= S_t \sigma_t \dd W^0_t,\\
\label{eqn:IV:dynamics}
\diff \Sigma_t
&= \nu_t \dd t + \eta_t \dd W^0_t + \sqrt{\xi_t} \dd W^1_t,
\end{align}
for a bivariate standard Brownian motion $(W^0,W^1)$ and processes $\sigma,\nu,\eta,\xi$. Here, $\sigma$ is the \emph{spot volatility}, and $\nu$, $\eta$, and $\xi$ correspond to the \emph{drift of implied volatility}, the \emph{correlated volatility of implied volatility}, and the \emph{uncorrelated squared\footnote{The parametrisation in terms of the \emph{squared} volatility of implied volatility is explained in Remark~\ref{rem:leading order controls}.} volatility of implied volatility}, respectively. The price process $C$ of the call in turn is
\begin{align}
\label{eqn:C:definition}
C_t
&= \cC(t,S_t,\Sigma_t).
\end{align}
By It\^o's formula, its dynamics are given by
\begin{align*}
\diff C_t
&= \diff\cC(t,S_t,\Sigma_t)
= \cC_t \dd t + \cC_S \dd S_t + \cC_{\Sigma} \dd\Sigma_t + \frac{1}{2}\cC_{SS} \dd \langle S \rangle_t + \cC_{S\Sigma} \dd\langle S, \Sigma \rangle_t + \frac{1}{2}\cC_{\Sigma\Sigma} \dd \langle \Sigma \rangle_t\notag\\
&= \cC_S \dd S_t + \eta_t \cC_{\Sigma} \dd W^0_t + \sqrt{\xi_t} \cC_{\Sigma} \dd W^1_t\\
&\qquad +\left\lbrace \cC_t +  \nu_t\cC_{\Sigma} + \frac{1}{2}\sigma_t^2 S_t^2 \cC_{SS} + \sigma_t \eta_t S_t\cC_{S\Sigma} + \frac{1}{2}(\eta_t^2 + \xi_t) \cC_{\Sigma\Sigma}\right\rbrace \diff t.
\end{align*}
We suppose that all liquidly traded assets are local martingales (cf.~Footnote~\ref{ftn:local martingales} and Remark~\ref{rem:drift}). Thus, the drift of the liquidly traded call must vanish. Using the PDE \eqref{eqn:C:PDE} to substitute $\cC_t=\cC_t(t,S_t,\Sigma_t)$, the following \emph{drift condition} obtains (cf.~\cite[Equation (3.6)]{Schonbucher1999}):
\begin{align}
\label{eqn:motivation:drift condition}
\nu_t\cC_{\Sigma} + \frac{1}{2}S_t^2 \cC_{SS}(\sigma_t^2 - \Sigma_t^2) + \sigma_t \eta_t  S_t\cC_{S\Sigma} + \frac{1}{2}(\eta_t^2 + \xi_t) \cC_{\Sigma\Sigma}
&= 0.
\end{align}
In view of \eqref{eqn:motivation:drift condition}, at most three of the four processes $\nu,\sigma,\eta$, and $\xi$ can be chosen arbitrarily for the resulting model to satisfy the drift condition. Further natural restrictions are $\sigma > 0$, $\xi \geq 0$, and $\Sigma > 0$. Note that the standard Black--Scholes model corresponds to the choice $\nu = \eta = \xi = 0$ and $\sigma_t = \Sigma_t = \Sigma_0$. Then, the drift condition~\eqref{eqn:motivation:drift condition} is clearly satisfied and spot and implied volatilities are constant and identical.

\begin{remark}
\label{rem:drift}
Let us briefly discuss as in \cite[Remark 2.2]{HerrmannMuhleKarbeSeifried2017} why we assume that the traded assets $S$ and $C$ have zero drifts. With nonzero drifts, the agent would use the traded assets not only as hedging instruments, but also as investment vehicles. This would complicate the analysis considerably as the limiting P\&L process would no longer be constant but stochastic. But the real-world drift rates usually have little impact on the \emph{hedging component}, i.e., the difference between a utility-based hedging strategy and the corresponding utility-based optimal investment strategy.\footnote{For example, \cite{HubalekKallsenKrawczyk2006} find in a L\'evy model that the (drift-dependent) variance-optimal hedge is virtually identical to the (drift-independent) Black--Scholes delta hedge.} Assuming that the traded assets have zero drifts allows us to focus on hedging rather than optimal investment. Indeed, the agent then has no incentive to trade the stock and the call other than as hedging instruments for the non-traded option.
\end{remark}

In the following Section~\ref{sec:setup}, we introduce a setup to formulate our hedging problem with \emph{uncertainty} about the processes $\nu,\sigma,\eta,\xi$.

\subsection{Model uncertainty setup}
\label{sec:setup}

Fix a time horizon $T>0$ and constants $S_0 > 0$, $\Sigma_0 > 0$, and $A_0 \in \RR$. Let
\begin{align*}
\Omega
&= \lbrace \omega = (\omega^S_t,\omega^\Sigma_t,\omega^A_t)_{t\in[0,T]} \in C([0,T];\RR^3) : \omega_0 = (S_0,\Sigma_0, A_0) \rbrace
\end{align*}
be the canonical space of continuous paths in $\RR^3$ starting in $(S_0, \Sigma_0, A_0)$, endowed with the topology of uniform convergence. Moreover, let $\cF$ be the Borel $\sigma$-algebra on $\Omega$. We denote by $(S_t)_{t \in [0,T]}$, $(\Sigma_t)_{t\in[0,T]}$, and $(A_t)_{t\in[0,T]}$ the first, second, and third component of the canonical process, respectively, i.e., $S_t(\omega) = \omega^S_t$, $\Sigma_t(\omega) = \omega^\Sigma_t$, and $A_t(\omega) = \omega^A_t$. We write $\FF = (\cF_t)_{t\in[0,T]}$ for the (raw) filtration generated by $(S,\Sigma,A)$, and denote by $M_t := \sup_{u\in[0,t]} S_u$, $t \in [0,T]$, the running maximum of $S$. Unless otherwise stated, all probabilistic notions requiring a filtration, such as progressive measurability etc., pertain to $\FF$. Finally, we write $(\bfX_t)_{t\in[0,T]}$ for the vector-valued process $\bfX_t = (S_t,A_t,M_t,\Sigma_t)$.

\begin{remark}
The processes $S$, $M$, and $\Sigma$, model the stock price, its running maximum, and the implied volatility of the traded call, respectively. The process $A$ is an additional state variable that can be used to track exotic features of the non-traded option the agent has to hedge. For instance, an Asian call option with strike $K>0$ has the payoff $\left(\frac{1}{T}\int_0^T S_t \dd t - K\right)^+$. Setting $A_t = \int_0^t S_u \dd u$, the payoff can be recast as $(\frac{1}{T}A_T - K)^+$ and exploiting the Markovian structure of the process $(S,A)$, the Black--Scholes value of the Asian call can be written as a function $\cV(t,S_t,A_t)$ of time, the current stock price, and the current value of the additional state variable $A$.
\end{remark}

We now introduce a large class of probability measures on $(\Omega, \cF)$ that will serve as alternative models for the evolution of the traded assets.

\begin{definition}
\label{def:fP00}
$\fP^{00}$ is the set of probability measures on $(\Omega,\cF)$ for which there exists a quadruple $\bfzeta^P = (\nu^P_t, \sigma^P_t, \eta^P_t, \xi^P_t)_{t\in[0,T]}$ of real-valued progressively measurable processes such that:
\begin{enumerate}
\item $S$ and $\Sigma - \int_0^\cdot \nu^P_t \dd t$ are (continuous) local $P$-martingales with quadratic (co-)variations
\begin{align}
\label{eqn:covariation:S and IV}
\begin{split}
\diff \langle S \rangle_t 
&= S_t^2 (\sigma^P_t)^2 \dd t,\\
\diff \langle \Sigma \rangle_t
&= ((\eta^P_t)^2 + \xi^P_t) \dd t,\\
\diff \langle S, \Sigma \rangle_t
&= S_t \sigma^P_t \eta^P_t \dd t;
\end{split}
\end{align}
\item $S$ and $\Sigma$ are $P$-a.s.~positive;
\item $\xi^P \geq 0$ $P$-a.s.;
\item the \emph{drift condition}
\begin{align}
\label{eqn:drift condition}
\nu^P_t \cC_\Sigma + \frac{1}{2}S_t^2\cC_{SS}((\sigma^P_t)^2 - \Sigma_t^2) + \sigma^P_t \eta^P_t S_t \cC_{S\Sigma} + \frac{1}{2}((\eta^P_t)^2 + \xi^P_t) \cC_{\Sigma\Sigma}
&= 0
\end{align}
holds $\diff t\times P\text{-a.e.}$ Here, the partial derivatives of $\cC$ are evaluated in $(t,S_t,\Sigma_t)$.
\end{enumerate}
\end{definition}

A probability measure $P \in \fP^{00}$ is called a \emph{model} and the process $\bfzeta^P$ is referred to as the \emph{control} corresponding to the model $P$. Each $P$ represents a market model for the stock price $S$ and the implied volatility $\Sigma$ with dynamics of the form \eqref{eqn:S:dynamics}--\eqref{eqn:IV:dynamics} (with $\sigma$ replaced by $\sigma^P$ etc.) and \eqref{eqn:drift condition} guarantees that the call price process is a local $P$-martingale (cf.~\eqref{eqn:motivation:drift condition}).

\begin{definition}
\label{def:reference model}
The function $\bfzeta^0:\RR_+ \to \RR^4$ given by $\bfzeta^0(\Sigma) = (0,\Sigma,0,0)^\tr$ is called \emph{reference feedback control}.
A probability measure $P \in \fP^{00}$ such that $\bfzeta^P_t = \bfzeta^0(\Sigma_t)$ $\diff t\times P$-a.e.~is called \emph{reference model}.
\end{definition}
Note that a reference model corresponds to a Black--Scholes model with constant volatility $\sigma_t \equiv \Sigma_t \equiv \Sigma_0$ and trivially satisfies the drift condition \eqref{eqn:drift condition}.

Next, we consider a subclass $\fP^0 \subset \fP^{00}$ which (in contrast to $\fP^{00}$) also prescribes the dynamics for the additional state variable $A$ that tracks exotic features of the non-traded option. To this end, we fix Borel functions $\alpha, \beta, \gamma, \delta: [0,T]\times\RR^3\to\RR$.

\begin{definition}
\label{def:fP0}
$\fP^0 = \fP^0(\alpha,\beta,\gamma,\delta) \subset \fP^{00}$ is the subset of probability measures $P$ such that $A$ is a (continuous) $P$-semimartingale with canonical decomposition
\begin{align}
\label{eqn:A:dynamics}
\diff A_t
&= \left( \alpha + \frac{(\sigma^P_t)^2}{2}\beta  \right) \dd t + \gamma \dd S_t + \delta \dd M_t
\end{align}
under $P$ (the functions $\alpha,\beta,\gamma,\delta$ are evaluated in $(t,S_t,A_t,M_t)$).
\end{definition}

The form \eqref{eqn:A:dynamics} for the dynamics of $A$ is flexible enough to express Black--Scholes values of, e.g., Asian options, options on the realised variance, or forward-start options by PDE methods. We also note that given sufficiently regular functions $\alpha,\beta,\gamma,\delta$, there is a unique reference model in $\fP^0$.

\subsection{Hedging problem}
\label{sec:hedging problem}

\paragraph{Dynamic model recalibration.}
Consider an agent who has sold a non-traded option (possibly exotic) on $S$ with sufficiently regular\footnote{See Assumption~\ref{ass:main result} for the precise details.} payoff $\sfV(S_T,A_T,M_T)$\footnote{Recall that $M$ is the running maximum of $S$ and that $A$ is a general state variable with dynamics of the form \eqref{eqn:A:dynamics} which can track exotic features of the option like the average stock price or the stock price at an intermediate time; cf.~\cite[Section~4.2]{HerrmannMuhleKarbeSeifried2017} for examples.} at maturity $T$.  She can hedge her exposure by trading dynamically and frictionlessly in the stock, the call, and the bank account. 

Among all possible dynamics, the agent considers as most plausible the Black--Scholes model corresponding to the currently observed implied volatility, i.e., $\nu = \eta = \xi = 0$ and $\sigma = \Sigma$ (recall that the drift condition \eqref{eqn:motivation:drift condition} holds for this choice). This corresponds to the reference belief that ``The future implied volatility stays at the currently observed level.'' Note that this differs from the conviction that  ``The future implied volatility equals the implied volatility observed at time $0$.'': the former belief allows for dynamic updating of the observed implied volatility, the latter does not. In particular, at each time $t$, the agent (re-)calibrates her Black--Scholes model to the observed market price of the liquidly traded call option. This is in line with the market practice of frequent recalibration of pricing models to observed option prices. The corresponding Black--Scholes value of the non-traded option can readily be obtained by PDE methods. To this end, let $\bfG = \RR_+ \times \RR \times \RR_+$ be the state space of the process $(S, A, M)$ and for each $\Sigma > 0$, let $\cV(\cdot,\Sigma)$ be a classical solution to the PDE
\begin{align}
\label{eqn:V:PDE}
\begin{split}
\cV_t + (\alpha + \frac{1}{2}\beta\Sigma^2)\cV_A + \frac{1}{2}\Sigma^2 S^2 (\cV_{SS} + 2\gamma\cV_{SA} + \gamma^2\cV_{AA})
&= 0\quad \text{on } (0,T)\times\bfG,\\
\delta\cV_A + \cV_M
&= 0\quad \text{on } \lbrace (t,S,A,M) : S \geq M \rbrace,\\
\cV(T,\cdot,\Sigma)
&= \sfV\quad \text{on } \bfG.
\end{split}
\end{align}
Define the process $V = (V_t)_{t\in[0,T]}$ by
\begin{align}
\label{eqn:V:definition}
V_t
&= \cV(t,S_t,A_t,M_t,\Sigma_t).
\end{align}
Then, as is well known, $V_t$ is the Black--Scholes value at time $t$ of the non-traded option $\sfV$ \emph{given the current observation} of the stock price $S_t$, the state variables $A_t$ and $M_t$, \emph{and the implied volatility $\Sigma_t$}. In other words, the Black--Scholes model used to value the option $\sfV$ is dynamically recalibrated to the observed call prices.

\paragraph{Trading strategies and Profit\&Loss processes.}
A (self-financing) \emph{trading strategy} is represented by a pair $\bfupsilon = (\theta, \phi)$ of real-valued, locally bounded,\footnote{For locally bounded, progressively measurable integrands, the stochastic integrals in \eqref{eqn:Y:definition} are well defined under each measure in $\fP^0$. The delta-vega hedge considered in our main result, Theorem~\ref{thm:main result}, is even continuous.} progressively measurable processes $\theta = (\theta_t)_{t\in[0,T]}$ and $\phi = (\phi_t)_{t\in[0,T]}$, which describe the number of stocks and calls held by the agent, respectively. Fix a constant $Y_0\in\RR$ and for each $P\in\fP^0$ and any trading strategy $\bfupsilon$, define the \emph{Profit\&Loss (P\&L) process} $Y^{\bfupsilon,P}=(Y^{\bfupsilon,P}_t)_{t\in[0,T]}$ by
\begin{align}
\label{eqn:Y:definition}
Y^{\bfupsilon,P}_t
&= Y_0 + V_0 + \int_0^t \theta_u \dd S_u + \int_0^t \phi_u \dd C_u - V_t.
\end{align}
Here, the stochastic integrals are constructed under $P$. The process $Y^{\bfupsilon,P}_t$ describes the value of the agent's portfolio at time $t$ under the model $P$, i.e., her initial capital $Y_0 + V_0$ (recall the definition of $V_t$ in \eqref{eqn:V:definition}) plus gains from self-financing trading in the liquidly traded assets (computed under $P$) minus the (recalibrated) Black--Scholes value $V_t$ of the non-traded option at time $t$. Note that while the position in the liquidly traded assets are ``marked to market'' and constitute ``real values'' (because these assets could be liquidated instantly by assumption), the non-traded option has to be ``marked to model'' and thus only has a ``theoretical value''. However, at the maturity $T$ of the non-traded option, $V_T$ equals the option's payoff and the value of the option becomes ``real''. In particular, $Y^{\bfupsilon,P}_T$ is the agent's actual terminal wealth.

\paragraph{Uncertainty aversion.} Fix a model set $\fP \subset \fP^0$ and a set $\fY$ of trading strategies. Similarly to \cite{HerrmannMuhleKarbeSeifried2017},\footnote{In \cite{HerrmannMuhleKarbeSeifried2017}, only the underlying but no liquid call is available for dynamic hedging and the spot volatility is the only control variable of the fictitious adversary.} we assume that the agent ranks trading strategies in $\fY$ according to a numerical representation of her preferences of the form
\begin{align}
\label{eqn:numerical representation}
\inf_{P \in \fP} \EX[P]{U(Y^{\bfupsilon,P}_T) + \frac{1}{\psi}\int_0^T U'(Y^{\bfupsilon,P}_t) f(\Sigma_t,\bfzeta^P_t) \dd t },
\end{align}
where $f$ is a suitable function such that for each $\Sigma > 0$, $\RR^4 \ni \bfzeta \mapsto f(\Sigma,\bfzeta)$ is strictly convex with a unique minimum of $0$ at the reference point $\bfzeta^0(\Sigma)$. The utility function $U$ describes the agent's attitude towards risk in a given model. The infimum over models in $\fP$ together with the penalty term\footnote{\label{ftn:penalty}Note that the penalty is imposed on the fictitious adversary who chooses the model $P$ after the agent has chosen her trading strategy $\bfupsilon$. Alternatively, it can be interpreted as a fictitious bonus for the agent.} (the second summand inside the expectation in \eqref{eqn:numerical representation}) expresses her attitude towards model uncertainty. The parameter $\psi > 0$  quantifies the magnitude of her \emph{uncertainty aversion}. Indeed, in the limit $\psi\downarrow0$, the second summand in \eqref{eqn:numerical representation} converges to the indicator $+\infty\1_{\lbrace \bfzeta^P \neq \bfzeta^0(\Sigma) \rbrace}$ and the criterion \eqref{eqn:numerical representation} collapses to the standard expected utility under the reference model. In this case, the agent faces no uncertainty aversion at all as she only deems the reference model plausible. Conversely, in the limit $\psi\uparrow\infty$, the penalty term converges to $0$ for all $P \in \fP$ and the criterion \eqref{eqn:numerical representation} becomes the familiar worst-case expectation $\inf_{P \in \fP} \EX[P]{U(Y^{\bfupsilon,P}_T)}$. In this case, the agent is very uncertainty-averse in that she regards every model in $\fP$ as equally plausible. The criterion \eqref{eqn:numerical representation} interpolates smoothly between these two extreme cases. The reference model is not penalised, while alternative models are underweighted in the agent's decision making according to their ``distance'' from the reference model. The interpretation is that the reference model is considered most plausible. Alternative models are taken less seriously, but not ruled out a priori.

For tractability, we focus on the following quadratic specification for the penalty function $f$:\footnote{More general functions $f$ are considered in \cite{HerrmannMuhleKarbeSeifried2017}, where it becomes apparent that only the locally quadratic structure at the minimum matters for the leading-order asymptotics.}
\begin{align}
\label{eqn:penalty function}
\begin{split}
f(\Sigma,\bfzeta)
&= \frac{1}{2}(\bfzeta - \bfzeta^0(\Sigma))^\tr \Psi^{-1} (\bfzeta - \bfzeta^0(\Sigma))\\
&= \frac{1}{2}(\nu^2/\psi_\nu + (\sigma - \Sigma)^2/\psi_\sigma + \eta^2/\psi_\eta + \xi^2/\psi_\xi)
\end{split}
\end{align}
where
\begin{align}
\label{eqn:Psi}
\Psi
&= \diag(\psi_\nu,\psi_\sigma,\psi_\eta,\psi_\xi)
\quad\text{and}\quad
\psi_\nu,\psi_\sigma,\psi_\eta,\psi_\xi > 0.
\end{align}
The parameters $\psi_\nu,\psi_\sigma,\psi_\eta,\psi_\xi$ describe the agent's relative uncertainty about the true drift of implied volatility, spot volatility, correlated volatility of implied volatility, and uncorrelated squared volatility of implied volatility, respectively. The scaling parameter $\psi$ measures her overall level of aversion against uncertainty.

\begin{remark}
\label{rem:marginal utility}
Let us argue as in \cite[Remark~2.6]{HerrmannMuhleKarbeSeifried2017} why we include the term $U'(Y^{\bfupsilon,P}_t)$ in the penalty term of the numerical representation \eqref{eqn:numerical representation}.\footnote{Formally, this corresponds to directly imposing the penalty in monetary terms, i.e., inside the utility function in \eqref{eqn:numerical representation}.} First, in the standard expected utility framework, preferences are invariant under affine transformations of the utility function. The term $U'(Y^{\bfupsilon,P}_t)$ ensures that this property is preserved for uncertainty-averse decision makers whose preferences are described by \eqref{eqn:numerical representation}. Second, $U'(Y^{\bfupsilon,P}_t)$ (rather than, e.g., $U'(Y_0)$\footnote{Using $U'(Y_0)$ instead of $U'(Y^{\bfupsilon,P}_t)$ would yield the same expansion for $v(\psi)$ as in Theorem~\ref{thm:main result}. Formally, the delta-vega hedge and the candidate optimal controls for the fictitious adversary would still be leading-order optimal. This is because the P\&L process converges to a constant in the limit of small uncertainty aversion. Consequently, one could also remove $U'(Y^{\bfupsilon,P}_t)$ from the penalty term by replacing the matrix $\Psi$ by $\Psi/U'(Y_0)$. Then $U'(Y_0)$ would reappear in the candidate feedback control for the fictitious adversary and hence also in the cash equivalent $\wt w_0$. Keeping $U'(Y^{\bfupsilon,P}_t)$ in the penalty term avoids that the candidate optimal controls depend on the current P\&L of the agent. This avoids some mathematical subtleties in the formulation of the hedging problem; cf.~\cite{HerrmannMuhleKarbeSeifried2017}, where the P\&L process $Y$ lives on the canonical space so that (progressively measurable) controls may depend on $Y$.}) is the natural choice for a dynamic formulation of the hedging problem \eqref{eqn:value} in terms of a family of conditional problems parametrised by the initial time $t$, stock price $S_t = s$, and P\&L $Y^{\bfupsilon,P}_t = y$. Third, our results show that the preferences described by \eqref{eqn:numerical representation} have approximately ``constant uncertainty aversion'' in the sense that the cash equivalent $\wt w_0$ does not depend on the P\&L (cf.~Proposition~\ref{prop:Feynman-Kac}). This would not be the case if one omitted the term $U'(Y^{\bfupsilon,P}_t)$ in \eqref{eqn:numerical representation}.\footnote{In the context of robust portfolio choice, Maenhout \cite{Maenhout2004} also observes that some modification of the standard (non wealth-dependent) entropic penalty is reasonable to avoid that the agent's uncertainty aversion wears off as her wealth rises, and tackles this effect by directly modifying the HJBI equation.}

We also note that the penalty term depends on the chosen trading strategy $\bfupsilon$ of the agent only through her current P\&L (just as the indirect risk tolerance process of an investor depends on her trading strategy only through her current wealth).
\end{remark}

\paragraph{Hedging problem.}
Fix $\psi > 0$. For each trading strategy $\bfupsilon \in\fY$ and each model $P \in \fP$, we define the \emph{objective} of our hedging problem by
\begin{align}
\label{eqn:objective}
J^\psi(\bfupsilon,P)
&:= \EX[P]{U(Y^{\bfupsilon,P}_T) + \frac{1}{\psi} \int_0^T U'(Y^{\bfupsilon,P}_t)f(\Sigma_t,\bfzeta^P_t) \dd t}.
\end{align}
We note that Assumption~\ref{ass:main result}~\ref{ass:main result:strategy set} below guarantees that the negative part of the integrand in \eqref{eqn:objective} is bounded, so that the expectation is well defined.
The \emph{value} of our hedging problem is 
\begin{align}
\label{eqn:value}
v(\psi)
&= v(\psi;\fY,\fP)
:= \sup_{\bfupsilon\in\fY}\inf_{P\in\fP} J^\psi(\bfupsilon,P).
\end{align}
To wit, the agent wants to find a strategy in $\fY$ that maximises the numerical representation of her preferences \eqref{eqn:numerical representation}. The goal of this paper is to find an asymptotic expansion of the value $v(\psi)$ for small levels of uncertainty aversion $\psi$ and to find a trading strategy that achieves the leading-order optimal performance.

\section{Heuristics}
\label{sec:heuristics}

The asymptotic solution of the family of SDGs \eqref{eqn:value} is related to a linearly constrained quadratic programming problem. In this section, we derive this optimisation problem heuristically from the HJBI equation associated to \eqref{eqn:value}. This motivates the definitions of the functions introduced in the subsequent Section~\ref{sec:main results}.

\paragraph{Effective greeks.} 
Let us assume for the moment that the true dynamics of the stock price are given by the Black--Scholes model with some (constant) volatility $\Sigma_0$. Then It\^o's formula and the PDE \eqref{eqn:V:PDE} for $\cV$ show that the replicating strategy (trading only the stock and the bank account, not the call) of the option with payoff $\sfV(S_T,A_T,M_T)$ is given by $\theta_t = (\cV_S + \gamma \cV_A)(t,S_t,A_t,M_t,\Sigma_0)$. In particular, the \emph{delta} $\cV_S$ of the option only gives the replicating strategy if $\gamma \equiv 0$, i.e., if the additional state variable $A$ is of finite variation (e.g., for vanilla options like the liquidly traded call, or exotics like barrier, Asian, or lookback options). In general, however, the replicating strategy also has to take into account the indirect sensitivity of the option value with respect to changes in the stock price arising from the additional state variable $A$ (e.g., for a forward-start option as in Example~\ref{ex:forward start call}). Therefore, we call
\begin{align*}
\Delta
&= \cV_S + \gamma \cV_A
\end{align*}
the \emph{effective delta} of the option $\sfV$. Similarly, we call
\begin{align*}
\Gamma
&= \cV_{SS} + 2\gamma\cV_{SA} + \gamma^2\cV_{AA}
\quad\text{and}\quad
\Vanna
= \cV_{S\Sigma} + \gamma\cV_{A\Sigma}
\end{align*}
the \emph{effective gamma} and \emph{effective vanna} of the option $\sfV$, respectively.

\begin{example}[Forward-start call]
\label{ex:forward start call}
A \emph{forward-start call} with payoff $(S_T - S_{T_\mathrm{reset}})^+$ is a call option whose strike is set at some future \emph{reset date} $T_{\mathrm{reset}} \in (0,T)$ (cf., e.g.,~\cite[Section~6.2]{MusielaRutkowski2005}). This option payoff can be embedded into our framework by choosing $A_0 = S_0$ and $\gamma(t) = \1_{\lbrace t < T_\mathrm{reset} \rbrace}$. Indeed, then $A_t = S_{t \wedge T_\mathrm{reset}}$ and the option payoff can be written as $\sfV(S_T,A_T) = (S_T - A_T)^+$.
\end{example}

\paragraph{Dynamics of the P\&L process.}
In order to write down the HJBI equation associated to the hedging problem, we need the dynamics of the P\&L process $Y^{\bfupsilon,P}$ for generic strategies $\bfupsilon$ and models $P$. Applying It\^o's formula to $Y^{\bfupsilon,P}$ (defined in \eqref{eqn:Y:definition}) under $P$ (with associated control $\bfzeta^P$) yields (cf.~Lemma~\ref{lem:Y})
\begin{align*}
\diff Y^{\bfupsilon,P}_t
&= \big(\theta_t - (\Delta(t,\bfX_t) - \phi_t \cC_S(t,S_t,\Sigma_t)) \big) \dd S_t + \big(\phi_t \cC_\Sigma(t,S_t,\Sigma_t) -\cV_\Sigma(t,\bfX_t) \big) \dd \Sigma^{\contlocmartpart,P}_t\\
&\qquad -b^\cV(t,\bfX_t;\bfzeta^P_t) \dd t,
\end{align*}
where $\Sigma^{\contlocmartpart,P} = \Sigma - \int_0^\cdot \nu^P_u \dd u$ is the (continuous) local martingale part of $\Sigma$ under $P$, and (writing $\bfx = (S,A,M,\Sigma) \in \RR^4$ and $\bfzeta = (\nu,\sigma,\eta,\xi) \in \RR^4$),
\begin{align}
\label{eqn:heuristic:bV}
b^\cV(t,\bfx;\bfzeta)
&= \nu\cV_{\Sigma}
+ \frac{1}{2}(\beta \cV_A + S^2 \Gamma)(\sigma^2-\Sigma^2)
+ \sigma\eta  S \Vanna + \frac{1}{2}(\eta^2 + \xi) \cV_{\Sigma\Sigma}.
\end{align}
For small uncertainty aversion, models far from the reference model are heavily penalised. Whence, the fictitious adversary needs to choose among small perturbations $\bfzeta = \bfzeta^0(\Sigma) + \wt\bfzeta \psi$ of the reference feedback control $\bfzeta^0(\Sigma) = (0,\Sigma,0,0)$. Plugging this perturbation into \eqref{eqn:heuristic:bV}, we find
\begin{align}
\label{eqn:heuristic:bV:perturbation}
b^\cV(\bfzeta)
&= \bfv^\tr \wt\bfzeta \psi + o(\psi),
\end{align}
where $\bfv = (\cV_\Sigma,\Sigma(\beta \cV_A + S^2 \Gamma),\Sigma S\frac{\partial\Delta}{\partial\Sigma},\frac{1}{2}\cV_{\Sigma\Sigma})$. Note that by expanding the function $b^\cV(\bfzeta^0(\Sigma) + \wt\bfzeta \psi)$ around $\psi=0$, the vector-valued function $\bfv$ in \eqref{eqn:heuristic:bV:perturbation} can also be identified as the gradient $\D_\bfzeta b^\cV$ evaluated in $\bfzeta^0(\Sigma)$.

\begin{remark}
\label{rem:leading order controls}
We now explain the use of the uncorrelated \emph{squared} volatility of implied volatility $\xi$ as a control variable. Equation \eqref{eqn:heuristic:bV:perturbation} shows that an $O(\psi)$-perturbation of the squared volatility around zero (i.e., a positive fourth component of $\wt\bfzeta$) affects the drift $b^\cV$ at the order $O(\psi)$ (at least as long as we are in the generic case where $\cV_{\Sigma\Sigma}$ is nonzero). If we used instead the uncorrelated volatility $\xi' := \sqrt{\xi}$ as a basic control variable, then $\xi$ in \eqref{eqn:heuristic:bV} would be replaced by $(\xi')^2$. Following the arguments that lead to \eqref{eqn:heuristic:bV:perturbation}, we would then find for a perturbation of the form $\bfzeta' = \bfzeta^0(\Sigma) + \wt{\bfzeta'} \psi$ that
\begin{align}
\label{eqn:heuristic:bV:perturbation:other power}
b^\cV(\bfzeta')
&= (\bfv')^\tr  \wt{\bfzeta'} \psi + o(\psi),
\end{align}
where $\bfv'$ is given by $\bfv$ \emph{with the fourth component replaced by zero}. Thus, a perturbation of $\xi'$ around zero of order $O(\psi)$ would then have no impact on the $O(\psi)$ term in the expansion \eqref{eqn:heuristic:bV:perturbation:other power} of $b^\cV$. This is an artefact of the Black--Scholes reference model: for any reference model with a \emph{nonzero} uncorrelated volatility of implied volatility, $\xi'^0 \neq 0$, the fourth component of $\bfv'$ would generically \emph{not} vanish, and hence an $O(\psi)$-perturbation of $\xi'$ around $\xi'^0$ \emph{would} affect the drift $b^\cV$ at the order $O(\psi)$.
\end{remark}

\paragraph{HJBI equation.}
The drift condition \eqref{eqn:drift condition} can be rephrased as $b^\cC(t,\bfX_t;\bfzeta^P_t) = 0$ $\diff t\times P$-a.e., where
\begin{align*}
b^\cC(t,\bfx;\bfzeta)
&= \nu \cC_\Sigma
+\frac{1}{2}S^2\cC_{SS}(\sigma^2 - \Sigma^2)
+\sigma\eta S\cC_{S\Sigma}
+\frac{1}{2}(\eta^2 + \xi)\cC_{\Sigma\Sigma}.
\end{align*}
In addition, the uncorrelated \emph{squared} volatility of implied volatility $\xi^P$ (the fourth component of $\bfzeta^P$) must be nonnegative (cf.~Definition~\ref{def:fP00}). Hence, the HJBI equation associated to the hedging problem reads as
\begin{align}
\label{eqn:heuristic:HJBI}
w^\psi_t(t,\bfx,y)
+ \sup_{\bfupsilon \in \RR^2} \inf_{\substack{\bfzeta \in \RR^4:\\b^\cC(t,\bfx;\bfzeta)=0,\ \zeta_4 \geq 0}} H^\psi(t,\bfx,y;\bfupsilon,\bfzeta)
&= 0,\qquad
w^\psi(T,\bfx,y)
= U(y),
\end{align}
where the function $H^\psi(t,\bfx,y;\bfupsilon,\bfzeta)$ (spelled out explicitly in \eqref{eqn:Hamiltonian}) depends on first- and second-order partial derivatives of $w^\psi$ with respect to the space variables $\bfx$ and $y$, and on the drift and diffusion coefficients describing the dynamics of $S,A,M,\Sigma$, and $Y^{\bfupsilon,P}$ under a model $P$ such that $\bfzeta^P_t = \bfzeta$. We refer to \cite[Section~4.1]{HerrmannMuhleKarbeSeifried2017} for a derivation of the HJBI equation from the martingale optimality principle of stochastic optimal control. In essence, the left-hand side of the HJBI equation arises from the drift of the process $w^\psi(t,\bfX_t,Y^{\bfupsilon,P}_t) + \frac{1}{\psi}\int_0^t U'(Y^{\bfupsilon,P}_u) f(\Sigma_u,\bfzeta^P_u) \dd u$ under $P$, which can be computed via It\^o's formula.

\paragraph{Asymptotic ansatz.}
As the Black--Scholes model is complete and the drift of the liquidly traded assets is zero under each model by assumption, we expect that the zeroth-order term in the expansion of $v(\psi)$ is simply the utility $U(Y_0)$ generated by the initial P\&L. Similarly, the optimal control $\bfzeta$ with zero uncertainty aversion should simply be the reference feedback control $\bfzeta^0(\Sigma)$. This motivates the following ansatz for the asymptotic expansion of the value function and the almost optimal feedback control:\footnote{In view of \cite{HerrmannMuhleKarbeSeifried2017}, it is expected that $\psi$ (and not, e.g., $\psi^{1/2}$ or $\psi^2$) is the correct power for the expansion of the value function. Alternatively, one could write $\psi^\alpha$ instead of $\psi$ in \eqref{eqn:heuristic:ansatz:w} and then find $\alpha = 1$ by matching the powers of the penalty term and the drift term of the P\&L process in the expansion of the HJBI equation in such a way that the optimisation over $\wt\bfzeta$ becomes nontrivial.}
\begin{align}
\label{eqn:heuristic:ansatz:w}
w^\psi(t,\bfx,y)
&= U(y) - U'(y) \wt w(t,\bfx) \psi,\\
\label{eqn:heuristic:ansatz:zeta}
\bfzeta^\psi(t,\bfx)
&= \bfzeta^0(\Sigma) + \wt\bfzeta(t,\bfx) \psi,
\end{align}
for functions $\wt w$ and $\wt\bfzeta = (\wt\nu,\wt\sigma,\wt\eta,\wt\xi)$ to be determined. In the reference model, any strategy in the stock and the call that neutralises the net delta qualifies as a replicating strategy. Whence, it is less obvious whether the delta-vega hedge $\bfupsilon^\star$ or any other strategy that neutralises the agent's net delta (e.g., the standard delta hedge without trading in the call) should be the candidate strategy $\bfupsilon = (\theta,\phi)$ for the hedging problem. Thus, we leave the choice of $\phi$ open for the moment and just assume that
\begin{align}
\label{eqn:heuristic:ansatz:theta}
\theta
&=\Delta - \phi \cC_S
\end{align}
neutralises the (effective) net delta. Plugging \eqref{eqn:heuristic:ansatz:w}--\eqref{eqn:heuristic:ansatz:theta} into the HJBI equation \eqref{eqn:heuristic:HJBI} (using the explicit formula \eqref{eqn:Hamiltonian} for $H^\psi$), dropping the $\sup_{\bfupsilon}\inf_\bfzeta$ (we assume that the candidate strategy and control form a saddle point), using the expansion \eqref{eqn:heuristic:bV:perturbation}, and ordering by powers of $\psi$, we obtain
\begin{align}
\label{eqn:heuristic:HJBI:expanded}
\begin{split}
U' \times &\left(-\wt w_t - (\alpha + \frac{1}{2}\beta\Sigma^2) \wt w_A - \frac{1}{2}\Sigma^2 S^2 (\wt w_{SS} + 2 \gamma \wt w_{SA} + \gamma^2 \wt w_{AA}) \right.\\
&\quad \left.+\frac{1}{2} \wt\bfzeta^\tr \Psi^{-1} \wt\bfzeta - \bfv^\tr \wt\bfzeta - \frac{1}{2}\wt\xi\left(\phi\cC_\Sigma - \cV_\Sigma\right)^2 \frac{-U''}{U'}\right) \psi + o(\psi)
= 0.
\end{split}
\end{align}
Moreover, the constraints in the minimisation part of the HJBI equation transform to
\begin{align}
\label{eqn:heuristic:constraints:expanded}
\left(\wt\nu \cC_\Sigma + \wt\sigma \Sigma S^2 \cC_{SS} + \wt\eta \Sigma S \cC_{S\Sigma} + \wt\xi \frac{1}{2} \cC_{\Sigma\Sigma}\right)\psi + o(\psi)
&= 0
\quad\text{and}\quad
\wt\xi \geq 0.
\end{align}

Our candidates for $\wt\bfzeta$ and $\phi$ now arise as the saddle point of the min-max problem (minimising over $\wt\bfzeta$ and maximising over $\phi$) corresponding to the $O(\psi)$ term in \eqref{eqn:heuristic:HJBI:expanded} subject to the constraints \eqref{eqn:heuristic:constraints:expanded}. Clearly, the vega hedge $\phi^\star = \frac{\cV_\Sigma}{\cC_\Sigma}$ maximises the $O(\psi)$ term over $\phi \in \RR$, irrespective of the choice of $\wt\bfzeta$. With this choice, the constrained minimisation over $\wt\bfzeta$ (ignoring the $o(\psi)$ term in the equality constraint  in \eqref{eqn:heuristic:constraints:expanded}) reduces to a linearly constrained quadratic programming problem:
\begin{align}
\label{eqn:heuristic:LCQP problem}
\text{minimise } \frac{1}{2}\wt\bfzeta^\tr \Psi^{-1} \wt\bfzeta - \bfv^\tr \wt\bfzeta
\quad\text{subject to }
\wt\bfzeta \in \RR^4,
\bfc^\tr\wt\bfzeta = 0,
\wt\zeta_4 \geq 0,
\end{align}
where $\bfc = \big(\cC_\Sigma,\Sigma S^2 \cC_{SS},\Sigma S\cC_{S\Sigma},\frac{1}{2}\cC_{\Sigma\Sigma}\big)$.

\paragraph{Solving the linearly constrained quadratic program.}
The minimisation problem \eqref{eqn:heuristic:LCQP problem} is strictly convex and linearly constrained and thus has a unique minimum. The minimiser $\wt\bfzeta^*$ is characterised by the associated \emph{Karush--Kuhn--Tucker conditions}
\begin{align*}
\Psi^{-1}\wt\bfzeta^* - \bfv  + \lambda^* \bfc - \mu^* \vec\bfe_4 = 0,\quad
\bfc^\tr\wt\bfzeta^* = 0,\quad
\wt\zeta_4^* \geq 0,\quad
\mu^* \geq 0,\quad
\mu^*\wt\zeta_4^* = 0,
\end{align*}
for some scalars $\lambda^*$ and $\mu^*$. It turns out that there is an explicit solution $(\wt\bfzeta^*,\lambda^*,\mu^*)$ (cf. Lemma~\ref{lem:LCQP}~(a)), which motivates our definitions in Section~\ref{sec:assumptions}.

\section{Main results}
\label{sec:main results}

This section contains the mathematically precise statement of our main results. In Section~\ref{sec:assumptions}, we first introduce the required notation and technical assumptions; this notationally heavy part can be skipped at first reading.

\subsection{Notation and assumptions}
\label{sec:assumptions}

Our main result, Theorem~\ref{thm:main result}, provides an asymptotic expansion of the value $v(\psi)$ from \eqref{eqn:value} for small levels of uncertainty aversion $\psi$ and an asymptotic saddle point $(\bfupsilon^\star, P^\psi)_\psi$, where $\bfupsilon^\star$ is the delta-vega hedge and $(P^\psi)_\psi$ is a suitable family of models. To define the PDE that describes the first-order term of the expansion and to define the quadruple $\bfzeta^\psi$ that corresponds approximately (see Definition~\ref{def:candidate asymptotic model family}~(b) below) to $P^\psi$, we need to introduce some notation.

Recall that $\bfG = \RR_+\times\RR\times\RR_+$ is the state space of the process $(S,A,M)$ and set $\bfD^0 = (0,T)\times\bfG\times\RR_+$. A generic element of $\bfD^0$ is written as $(t,S,A,M,\Sigma)$ or $(t,\bfx)$ with $\bfx = (S,A,M,\Sigma)$. The functions $\Delta, \Gamma, \Vanna: \bfD^0 \to \RR$ defined by\footnote{Here and in the following, we assume that all relevant partial derivatives of $\cC$ and $\cV$ exist; precise conditions are given in Assumption~\ref{ass:main result} below.}
\begin{align*}
\Delta(t,\bfx)
&= \cV_{S} + \gamma \cV_{A},\quad
\Gamma(t,\bfx)
= \cV_{SS} + 2\gamma\cV_{SA} + \gamma^2\cV_{AA},\quad
\Vanna(t,\bfx)
= \cV_{S\Sigma} + \gamma\cV_{A\Sigma}
\end{align*}
are called the \emph{effective delta}, \emph{effective gamma}, and \emph{effective vanna} of the option $\sfV$, respectively; we note that these quantities correspond to the standard greeks if $\gamma \equiv 0$ like for vanilla, barrier, or lookback options, for example, and refer to Section~\ref{sec:heuristics} for a motivation of this terminology in the case $\gamma \not\equiv 0$. The functions $\bfc,\bfv:\bfD^0 \to \RR^4$ given by
\begin{align}
\label{eqn:VGVV:c}
\bfc(t,\bfx)
&= \Big(\cC_\Sigma,\Sigma S^2 \cC_{SS},\Sigma S\cC_{S\Sigma},\frac{1}{2}\cC_{\Sigma\Sigma}\Big)^\tr,\\
\label{eqn:VGVV:v}
\bfv(t,\bfx)
&= \Big(\cV_\Sigma,\Sigma(\beta \cV_A + S^2 \Gamma),\Sigma S\frac{\partial\Delta}{\partial\Sigma},\frac{1}{2}\cV_{\Sigma\Sigma}\Big)^\tr,
\end{align}
are called the \emph{vega-gamma-vanna-volga vector} of the call and the option $\sfV$, respectively. With this notation, define the functions $\lambda,\mu:\bfD^0 \to \RR$ and $\wt\bfzeta:\bfD^0 \to \RR^4$ as follows:
\begin{align}
\label{eqn:lagrange multiplier:lambda}
\lambda(t,\bfx)
&=
\begin{cases}
\frac{\bfc^\tr\Psi \bfv}{\bfc^\tr\Psi \bfc}
& \text{if } \cV_{\Sigma\Sigma} - \frac{\bfc^\tr\Psi\bfv}{\bfc^\tr\Psi\bfc} \cC_{\Sigma\Sigma} \geq 0,\\
\frac{\bfc^\tr \Psi \bfv - \frac{1}{4}\cC_{\Sigma\Sigma}\cV_{\Sigma\Sigma}\psi_\xi}{\bfc^\tr \Psi \bfc -\frac{1}{4}\cC_{\Sigma\Sigma}^2 \psi_\xi}
& \text{otherwise,}
\end{cases}\\
\label{eqn:lagrange multiplier:mu}
\mu(t,\bfx)
&= \frac{1}{2}(\cV_{\Sigma\Sigma} - \lambda \cC_{\Sigma\Sigma})^-,\\
\label{eqn:candidate control:first order}
\wt\bfzeta(t,\bfx)
&= \Psi (\bfv - \lambda \bfc + \mu \vec\bfe_4).
\end{align}
Note that the term $\mu \vec\bfe_4$ in \eqref{eqn:candidate control:first order} ensures that the fourth component of $\wt\bfzeta$ is nonnegative. Now, fix constants $0 < \ul\Sigma < \Sigma_0 < \ol\Sigma$, and define for each $\psi > 0$, the \emph{candidate feedback control} $\bfzeta^\psi =(\nu^\psi, \sigma^\psi, \eta^\psi, \xi^\psi):\bfD^0 \to \RR^4$ by
\begin{align}
\label{eqn:candidate control:short}
\bfzeta^\psi(t,\bfx)
&= \bfzeta^0(\Sigma) + \wt\bfzeta\1_{\lbrace\ul\Sigma<\Sigma<\ol\Sigma\rbrace} \psi.
\end{align}
The indicator $\1_{\lbrace\ul\Sigma<\Sigma<\ol\Sigma\rbrace}$ is a technical modification that ensures that the implied volatility stays within the interval $[\ul\Sigma,\ol\Sigma]$ by falling back to the reference feedback control $\bfzeta^0(\Sigma)$ (which corresponds to constant implied volatility) as soon as the implied volatility hits the boundary of $[\ul\Sigma,\ol\Sigma]$.
More explicitly, the candidate feedback control can be expressed as
\begin{align*}
\nu^\psi(t,\bfx)
&= (\cV_\Sigma - \lambda \cC_\Sigma) \1_{\lbrace \ul\Sigma < \Sigma < \ol\Sigma \rbrace} \psi_\nu \psi,\\
\sigma^\psi(t,\bfx)
&= \Sigma + \Sigma\left(\beta\cV_A  + S^2\Gamma- \lambda S^2\cC_{SS}\right) \1_{\lbrace \ul\Sigma < \Sigma < \ol\Sigma \rbrace} \psi_\sigma \psi,\\
\eta^\psi(t,\bfx)
&= \Sigma \left({\textstyle S\frac{\partial \Delta}{\partial \Sigma}} - \lambda S\cC_{S\Sigma}\right) \1_{\lbrace \ul\Sigma < \Sigma < \ol\Sigma \rbrace} \psi_\eta \psi,\\
\xi^\psi(t,\bfx)
&= \frac{1}{2}(\cV_{\Sigma\Sigma} - \lambda \cC_{\Sigma\Sigma})^+ \1_{\lbrace \ul\Sigma < \Sigma < \ol\Sigma \rbrace} \psi_\xi \psi.
\end{align*}

In general, there is no $P^\psi \in \fP^0$ such that $\bfzeta^\psi$ coincides with the control $\bfzeta^{P^\psi}$ corresponding to $P^\psi$ as the process $\bfzeta^\psi(t,\bfX_t)$ fulfils the drift condition \eqref{eqn:drift condition} only at the order $O(\psi)$. However, to match the drift condition exactly, one can perturb $\bfzeta^\psi$ by a suitable, asymptotically small term. This motivates part~(b) of the following definition.

\begin{definition} 
\label{def:candidate asymptotic model family}
Let $\fP \subset \fP^0$.
\begin{enumerate}
\item For each $p\geq1$, we denote by $L^p_\fP$ the vector space of Borel functions $K:\bfD^0\to\RR$ satisfying
\begin{align*}
\norm{K}_{L^p_\fP} := \sup_{P \in \fP} \EX[P]{\int_0^T \left\vert K(t,\bfX_t) \right\vert^p \dd t}^{1/p} < \infty.
\end{align*}

\item A family $(P^\psi)_{\psi\in(0,\psi_0)}\subset\fP$ for some $\psi_0 \in(0,1)$ is called a \emph{candidate asymptotic model family (in $\fP$)} if there is $K_0\in L^4_\fP$ such that for all $\psi \in (0,\psi_0)$,
\begin{align*}
\enorm{\bfzeta^{P^\psi}_t - \bfzeta^\psi(t,\bfX_t)}
&\leq K_0(t,\bfX_t) \psi^2 \quad \diff t \times P^\psi\text{-a.e.}
\end{align*}
\end{enumerate}
\end{definition}
The crucial property of a candidate asymptotic model family formalised in Definition~\ref{def:candidate asymptotic model family}~(b) is that the control $\bfzeta^{P^\psi}_t$ corresponding to $P^\psi$ is $O(\psi^2)$-close to the candidate control $\bfzeta^\psi(t,\bfX_t)$.

The leading-order coefficient of the asymptotic expansion of $v(\psi)$ is given in terms of the solution to a linear second-order parabolic PDE with a source term. Specifically, for each ${\Sigma \in [\ul\Sigma,\ol\Sigma]}$, we consider the PDE
\begin{align}
\label{eqn:cash equivalent:PDE}
\begin{split}
\wt w_t + \big(\alpha + \frac{1}{2}\beta \Sigma^2\big) \wt w_A\qquad\qquad\qquad\qquad\qquad\qquad\,&\\
+ \frac{1}{2}\Sigma^2 S^2 \left(\wt w_{SS} + 2\gamma \wt w_{SA} + \gamma^2 \wt w_{AA} \right) + \frac{1}{2}\wt g(\cdot,\Sigma)
&= 0\quad \text{on } (0,T)\times\bfG,\\
\delta \wt w_A + \wt w_M
&= 0\quad \text{on } \lbrace (t,S,A,M) : S \geq M \rbrace,\\
\wt w(T,\cdot,\Sigma)
&= 0\quad \text{on } \bfG,
\end{split}
\end{align}
where the source term $\wt g:\bfD^0 \to \RR$ is given by
\begin{align}
\label{eqn:cash equivalent:source term}
\wt g(t,\bfx)
&= \bfv(t,\bfx)^\tr\wt\bfzeta(t,\bfx).
\end{align}

We prove our main result under the following assumptions.

\begin{assumption}
\label{ass:main result}
Set $\bfD = (0,T)\times\bfG\times[\ul\Sigma,\ol\Sigma] \subset \bfD^0$.

\begin{enumerate}[label=(\alph*)]
\item
\label{ass:main result:strategy set}
\emph{Trading strategy set:} There is a constant $K_\fY > 0$ such that for each trading strategy $\bfupsilon \in \fY$ and each $P\in\fP$, $Y^{\bfupsilon,P} > -K_\fY$ $\diff t \times P$-a.e.

\item
\label{ass:main result:model set}
\emph{Model set:} $\fP \subset \fP^0$ contains a candidate asymptotic model family, a reference model, and there are constants $\ul\nu<0<\ol\nu$, $0 < \ul\sigma < \ul\Sigma$, $\ol\Sigma < \ol\sigma$, $\ul\eta<0<\ol\eta$, and $\ol\xi > 0$ such that for each $P \in \fP$,
\begin{align}
\label{eqn:ass:model set}
\nu^P \in [\ul\nu,\ol\nu],\;
\sigma^P \in [\ul\sigma,\ol\sigma],\;
\eta^P \in [\ul\eta,\ol\eta],\;
\xi^P \in [0,\ol\xi],\;
\Sigma \in [\ul\Sigma,\ol\Sigma]\quad
\diff t\times P\text{-a.e.}
\end{align}

\item
\label{ass:main result:call}
\emph{Call PDE:} $T_\sfC \geq T$ and there is $\cC \in C^{1,2,2}((0,T_\sfC)\times\RR_+\times\RR_+) \cap C([0,T_\sfC]\times\ol\RR_+\times\ol\RR_+)$ such that for each $\Sigma\in[\ul\Sigma,\ol\Sigma]$, $\cC(\cdot,\Sigma)$ is a classical solution to the PDE \eqref{eqn:C:PDE} and
\begin{align}
\label{eqn:ass:C}
\cC_\Sigma \neq 0
\quad\text{and}\quad
\left\vert\cC_{\Sigma\Sigma}\right\vert
&\leq K_\cC \left( \vert \cC_\Sigma \vert + \vert S^2 \cC_{SS} \vert + \vert S \cC_{S\Sigma} \vert \right)
\text{ on } (0,T)\times\RR_+\times[\ul\Sigma,\ol\Sigma]
\end{align}
for some $K_\cC \in L^2_\fP$.

\item
\label{ass:main result:option}
\emph{Non-traded option PDE:} There is $\cV \in C^{1,2,2,1,2}(\bfD^0) \cap C(\ol{\bfD^0})$ such that for each $\Sigma \in [\ul\Sigma,\ol\Sigma]$, $\cV(\cdot,\Sigma)$ is a classical solution to the PDE \eqref{eqn:V:PDE} with
\begin{align}
\label{eqn:ass:V}
\left\vert \cV_\Sigma \right\vert,
\left\vert \beta \cV_A + S^2 (\cV_{SS} + 2 \gamma \cV_{SA} + \gamma^2\cV_{AA})\right\vert,
\left\vert S(\cV_{S\Sigma} + \gamma \cV_{A\Sigma}) \right\vert,
\left\vert \cV_{\Sigma\Sigma} \right\vert
&\leq K_\cV
\text{ on } \bfD
\end{align}
for some constant $K_\cV > 0$.

\item
\label{ass:main result:cash equivalent}
\emph{Cash equivalent PDE:} There is $\wt w \in C^{1,2,2,1,2}(\bfD^0)\cap C(\ol{\bfD^0})$ such that for each $\Sigma\in[\ul\Sigma,\ol\Sigma]$, $\wt w(\cdot, \Sigma)$ is a classical solution to the PDE \eqref{eqn:cash equivalent:PDE},
\begin{align*}
0
&\leq \wt w
\leq K_{\wt w}
\text{ on } \bfD
\end{align*}
for some constant $K_{\wt w} > 0$, and
\begin{align*}
\wt w_\Sigma, S (\wt w_S + \gamma \wt w_A), \beta \wt w_A + S^2 (\wt w_{SS} + 2 \gamma \wt w_{SA} + \gamma^2 \wt w_{AA}), S(\wt w_{S\Sigma} + \gamma \wt w_{A\Sigma}), \wt w_{\Sigma\Sigma}
&\in L^4_\fP.
\end{align*}

\item
\label{ass:main result:utility}
\emph{Utility function:} $U:\RR \to \RR$ is $C^3$ with $U' > 0$, $U'' < 0$ everywhere and has \emph{decreasing absolute risk aversion}, i.e., $y \mapsto -\frac{U''(y)}{U'(y)}$ is nonincreasing on $\RR$.
\end{enumerate}
\end{assumption}

\begin{remark} Let us discuss the various requirements in Assumption~\ref{ass:main result}:
\label{rem:ass:main result}
\begin{enumerate}
\item This constraint on the agent's credit line is an admissibility condition for the set of trading strategies. The P\&L process $Y^{\bfupsilon,P}$ is required to be bounded from below, uniformly over all strategies in $\fY$ and all models in $\fP$. We show in Corollary~\ref{cor:Y:delta-vega} that this is satisfied for the P\&L process associated to the delta-vega hedge $\bfupsilon^\star$ (cf.~\eqref{eqn:thm:main result:hedge}). Hence, making the constant $K_\fY$ larger if necessary, the delta-vega hedge can always be added to the set of strategies $\fY$.

\item A construction of a candidate asymptotic model family compatible with \eqref{eqn:ass:model set} is outlined in Section~\ref{sec:existence}. The existence of uniform bounds on the controls as well as the implied volatility are essential for various steps in the proof of the main result. This is not as big an assumption is it might appear at first glance. Indeed, as the conclusions of our main result do not depend on the choice of these bounds,  they can be chosen arbitrarily large.

\item These regularity assumptions ensure that $\cC$ corresponds to the Black--Scholes value of the liquidly traded call. The condition $\cC_\Sigma \neq 0$ guarantees that the delta-vega hedge (cf.~\eqref{eqn:thm:main result:hedge}) is well defined. The second condition in \eqref{eqn:ass:C} ensures that the volga of the call is dominated by the sum of its vega, cash gamma, and cash vanna.

For a plain-vanilla call option with payoff $\sfC(S) = (S - K)^+$, explicit formulas for these greeks show that this requirement is met if $\log S \in L^2_\fP$. This in turn follows easily from the explicit representation of $S$ as a stochastic exponential together with the boundedness of the spot volatility from Assumption~\ref{ass:main result}~\ref{ass:main result:model set}.

Another example is the log-contract with payoff $\sfC(S) = \log(S)$, for which $\cC(t,S,\Sigma) = \log(S) - \frac{1}{2}\Sigma^2(T_{\sfC}-t)$. Computing the relevant greeks shows that \eqref{eqn:ass:C} holds in this case, too. Moreover, if $T_\sfC > T$, then even the stronger condition \eqref{eqn:lem:existence:C} of Lemma~\ref{lem:existence} is satisfied.

\item This is a regularity assumption on the option $\sfV$ similar to \ref{ass:main result:call}. However, we additionally enforce the bounds \eqref{eqn:ass:V} to ensure that the vega-gamma-vanna-volga vector $\bfv$ is bounded. This assumption is satisfied if the option payoff $\sfV$ is sufficiently regular.

For example, consider the case where the payoff function $\sfV(S,A,M)=H(S)$ only depends on the stock price $S$. The corresponding Black--Scholes value can be written as
\begin{align*}
\cV(t,S,A,M,\Sigma)
&= \int_{-\infty}^\infty H\left(S \exp\Big( \Sigma \sqrt{T-t} x -\frac{1}{2}\Sigma^2(T-t) \Big) \right) \phi(x) \dd x,
\end{align*}
where $\phi$ is the density function of the standard normal distribution. If the ``terminal cash delta'' $y H'(y)$ and the ``terminal cash gamma'' $y^2 H''(y)$ are bounded in $y\in\RR_+$,\footnote{This holds, e.g., for a ``smooth put'', whose payoff is the Black--Scholes put value with some arbitrarily short maturity.} then using dominated convergence to differentiate under the integral sign shows that $\cV$ indeed satisfies Assumption~\ref{ass:main result}~\ref{ass:main result:option}. 

For exotic options, one can argue along the same lines. For example, for a lookback option with sufficiently regular payoff $\sfV(S,A,M) = H(S,M)$ (recall that $M$ is the variable for the running maximum of the stock), one can again verify that the probabilistic representation of its Black--Scholes value solves the the PDE \eqref{eqn:V:PDE} and inherits the required regularity of Assumption~\ref{ass:main result}~\ref{ass:main result:option} from the regularity of the payoff function $H$.

\item This assumption posits that a (classical) solution $\wt w$ to the PDE \eqref{eqn:cash equivalent:PDE} exists and satisfies certain bounds. The validity of this assumption depends on the regularity of the input quantities $\cC$, $\cV$, $\alpha$, $\beta$, $\gamma$, and $\delta$, and can be checked along the lines of (d) above.\footnote{See also \cite[Remark 3.2]{HerrmannMuhleKarbeSeifried2017} for a discussion of such regularity assumptions in a similar setting.}

\item It is not essential that the utility function is defined on the whole real line. In fact, as we only consider strategies such that the P\&L process is bounded from below by $-K_\fY$ uniformly over trading strategies and models, we could also work with a (suitably displaced) utility function on $\RR_+$. Also note that power and exponential utilities both have decreasing absolute risk aversion.
\end{enumerate}
\end{remark}

\begin{remark}
\label{rem:stochastic volatility}
As long as the traded option is regular enough, there are many models which fulfil assumption \eqref{eqn:ass:model set} for the coefficients of the implied volatility dynamics. For instance, consider a stochastic volatility model of the form
\begin{align}
\label{eqn:rem:stochastic volatility}
\begin{split}
\dd S_t
&= S_t a(Y_t) \dd W^0_t,\\
\dd Y_t
&= b(Y_t) \dd t + c_0(Y_t) \dd W^0_t + c_1(Y_t) \dd W^1_t,
\end{split}
\end{align}
where the functions $a, b, c_0, c_1$ as well as their derivatives are all Lipschitz and bounded, and $a, c_0, c_1$ are in addition positive and bounded away from zero. Then the spot volatility $\sigma_t = a(Y_t)$ evolves in some bounded interval $[\ul\sigma,\ol\sigma]$. Now, let $\cC^\mathrm{sv}(t,S_t,Y_t)$ be the value of a log-contract with payoff $\log(S_{T_\sfC})$ for some $T_\sfC > T$ computed in this stochastic volatility model (under some pricing measure). As the spot volatility $\sigma_t$ is bounded from above and from below, the value $\cC^\mathrm{sv}$ of the log-contract can be bounded from above and below by its Black--Scholes values for volatility $\ul\sigma$ and $\ol\sigma$, respectively. Whence, the implied volatility $\Sigma_t$ is uniformly bounded and bounded away from zero, too. To determine its drift and diffusion coefficients $\nu_t$, $\eta_t$, and $\xi_t$, apply It\^o's formula on both sides of the equation $\cC(t,S_t,\Sigma_t) = \cC^\mathrm{sv}(t,S_t,Y_t)$ that defines $\Sigma_t$ and compare the coefficients of the $\diff W^0$- and $\diff W^1$-terms. Using also that the cash delta for the log-contract is $S\cC_S = S\cC^\mathrm{sv}_S = 1$, this leads to
\begin{align*}
\eta_t
&= c_0(Y_t)\frac{\cC^\mathrm{sv}_Y(t,S_t,Y_t)}{\cC_\Sigma(t,S_t,\Sigma_t)},
\qquad
\xi_t
= \left(c_1(Y_t)\frac{\cC^\mathrm{sv}_Y(t,S_t,Y_t)}{\cC_\Sigma(t,S_t,\Sigma_t)}\right)^2.
\end{align*}
Now, differentiating the PDE for $\cC^\mathrm{sv}$ yields a PDE for its partial derivative $\cC^\mathrm{sv}_Y$ whose probabilistic representation shows that $\cC^\mathrm{sv}_Y$ is bounded. As $\cC_\Sigma$ is uniformly bounded away from zero for the log-contract with $T_\sfC > T$, $\eta$ and $\xi$ are uniformly bounded as well. Finally, by the drift condition \eqref{eqn:drift condition} (which holds automatically here because $\cC^\mathrm{sv}(t,S_t,Y_t)$ is a local martingale by construction), it follows that also the drift coefficient $\nu_t$ of the implied volatility is uniformly bounded. In summary, the market model derived from the stochastic volatility model \eqref{eqn:rem:stochastic volatility} fulfils \eqref{eqn:ass:model set}.
\end{remark}

\subsection{Main result}
\label{sec:main result}

We are now in a position to state our main result, which provides an asymptotic expansion of the value in \eqref{eqn:value} and a corresponding asymptotically optimal policy. The existence of a suitable corresponding model set $\fP$ and a candidate asymptotic model family is considered in Section~\ref{sec:existence} below. Recall from Remark~\ref{rem:ass:main result}~(a) that the delta-vega hedge $\bfupsilon^\star$ can always be included into the set of trading strategies $\fY$ by making the constant $K_\fY$ from Assumption~\ref{ass:main result}~\ref{ass:main result:strategy set} larger if necessary.

The number $\wt w_0 := \wt w(0,\bfX_0)$ defined through the solution $\wt w$ to the PDE \eqref{eqn:cash equivalent:PDE} determines the leading-order coefficient in the expansion of the value $v(\psi)$. As it also describes the (normalised) premium that the agent demands as a compensation for exposing herself to model misspecification (cf.~the expansion \eqref{eqn:price:ask} of the indifference ask price below), we call it the \emph{cash equivalent (of small uncertainty aversion)}.

\begin{theorem}
\label{thm:main result}
Let $\fY$ be a set of trading strategies, $\fP \subset \fP^0$ a model set, and suppose that Assumption~\ref{ass:main result} is satisfied. Define the \emph{delta-vega hedging strategy} $\bfupsilon^\star = (\theta^\star_t,\phi^\star_t)_{t\in[0,T]}$ by
\begin{align}
\label{eqn:thm:main result:hedge}
\begin{split}
\theta^\star_t
&= \left(\Delta - \frac{\cV_\Sigma}{\cC_\Sigma} \cC_S \right)(t,S_t,A_t,M_t,\Sigma_t),\\
\phi^\star_t
&= \frac{\cV_\Sigma}{\cC_\Sigma}(t,S_t,A_t,M_t,\Sigma_t).
\end{split}
\end{align}
If $\bfupsilon^\star \in \fY$ and $(P^\psi)_{\psi\in(0,\psi_0)}\subset\fP$ is a candidate asymptotic model family, then as $\psi \downarrow 0$:
\begin{align}
\label{eqn:thm:main result:expansion}
\begin{split}
v(\psi)
&= \sup_{\bfupsilon\in\fY} \inf_{P \in \fP} J^\psi(\bfupsilon, P)
= \inf_{P \in \fP} \sup_{\bfupsilon\in\fY} J^\psi(\bfupsilon, P) + o(\psi)\\
&= J^\psi(\bfupsilon^\star,P^\psi) + o(\psi)
= \sup_{\bfupsilon\in\fY} J^\psi(\bfupsilon, P^\psi) + o(\psi)
= \inf_{P \in \fP} J^\psi(\bfupsilon^\star, P) + o(\psi)\\
&= U(Y_0) - U'(Y_0) \wt w_0\psi + o(\psi).
\end{split}
\end{align}
In particular, the delta-vega hedge $\bfupsilon^\star$ is an optimal strategy at the leading order $O(\psi)$ among all strategies in $\fY$, and $P^\psi$ is a leading-order optimal choice of model for the fictitious adversary among all models in $\fP$.
\end{theorem}

The lengthy proof of Theorem~\ref{thm:main result} is postponed to Section~\ref{sec:proofs:main result}. The first-order term in the expansion of $v(\psi)$ in \eqref{eqn:thm:main result:expansion} is determined by the cash equivalent $\wt w_0$. Its probabilistic representation allows to identify the main factors that determine an option's susceptibility to model misspecification:

\begin{proposition}[Feynman--Kac representation]
\label{prop:Feynman-Kac}
Suppose that Assumption~\ref{ass:main result} holds and let $P^0 \in \fP$ be a reference model. Then
\begin{align*}
\wt w_0
&= \frac{1}{2}\EX[P^0]{\int_0^T \wt g(t,S_t,A_t,M_t,\Sigma_0) \dd t}.
\end{align*}
Here, the function $\wt g$ (defined in \eqref{eqn:cash equivalent:source term}) can be written as
\begin{align}
\label{eqn:prop:Feynman-Kac:g}
\begin{split}
&\wt g(t,S,A,M,\Sigma)\\
&\quad= -\Sigma \left(\phi^\star S^2 \cC_{SS} - (\beta \cV_A + S^2 \Gamma)\right)\wt\sigma
-\Sigma \left(\phi^\star S\cC_{S\Sigma} - S\Vanna \right) \wt\eta
-\frac{1}{2}\left(\phi^\star \cC_{\Sigma\Sigma} - \cV_{\Sigma\Sigma}\right)\wt\xi,
\end{split}
\end{align}
where the functions $(\wt\nu,\wt\sigma,\wt\eta,\wt\xi) = \wt\bfzeta$ are defined in \eqref{eqn:candidate control:first order} and $\phi^\star = \frac{\cV_\Sigma}{\cC_\Sigma}$ is the vega hedge from Theorem~\ref{thm:main result}.
\end{proposition}

\begin{proof}
The Feynman--Kac representation is proved in Proposition~\ref{prop:Feynman-Kac raw} (also note that $\Sigma_t = \Sigma_0$ ${\diff t\times P^0}$-a.e.~because $P^0$ is a reference model). The representation of $\wt g$ is the content of Corollary~\ref{cor:source term representation}.
\end{proof}

For an interpretation of this representation in the case of $\beta\equiv\gamma \equiv 0$, we refer to the discussion after equation \eqref{eqn:intro:source term} in the introduction. If $\gamma \not\equiv 0$ (e.g., for a forward-start call as in Example~\ref{ex:forward start call}), then the effective gamma and effective vanna are in general different from the gamma and vanna of the option. If the option $\sfV$ depends on the realised variance of the stock (e.g., a call on the realised variance), then a term $\beta \cV_A$ is added to the effective gamma in \eqref{eqn:prop:Feynman-Kac:g}.

The next proposition implies that whenever the vega-gamma-vanna-volga vectors of the call and the non-traded option $\sfV$ are collinear, the local impact of uncertainty aversion vanishes at the leading order.

\begin{proposition}
\label{prop:VGVV collinear}
Fix $(t,\bfx) \in \bfD^0$. If the vega-gamma-vanna-volga vectors $\bfc(t,\bfx)$ and $\bfv(t,\bfx)$ are collinear, then $\wt g(t,\bfx) = 0$.
\end{proposition}

\begin{proof}
Fix $(t,\bfx) \in \bfD^0$ and let $k \in \RR$ such that $\bfv(t,\bfx) = k \bfc(t,\bfx)$. Then by construction (cf.~\eqref{eqn:lagrange multiplier:lambda}--\eqref{eqn:candidate control:first order}), ${\lambda(t,\bfx) = k}$, $\mu(t,\bfx) = 0$, and $\wt\bfzeta(t,\bfx) = 0$. Thus, $\wt g(t,\bfx) = 0$.
\end{proof}

For example, consider the case where the non-traded option is a put with the same strike and maturity as the liquidly traded call. Then the put-call parity implies that the vegas, gammas, vannas, and volgas of both options coincide everywhere. Thus, $\wt g \equiv 0$ and hence the cash equivalent $\wt w_0$ vanishes. This is expected as put-call parity also provides a model-free hedge for this situation.

\paragraph{Indifference prices.}
The \emph{indifference ask price} (for the non-traded option $\sfV$) is the price at which the agent is indifferent between keeping a flat position and changing her position by selling the non-traded option for that price.

Recall that $V_0$ is the initial reference value of the non-traded option $\sfV$ and that $\wt w_0$ is its cash equivalent. Let $v(y;\psi)$ denote the value of our hedging problem corresponding to initial P\&L $y$. If the agent decides to sell the non-traded option for a price $p_a(\psi)$, then her initial P\&L for the hedging problem is $Y_0 + p_a(\psi) - V_0$. Therefore, the equation determining the indifference ask price $p_a(\psi)$ reads as follows:
\begin{align*}
U(Y_0)
&= v(Y_0 + p_a(\psi) - V_0;\psi).
\end{align*}
Using the expansion of $v$ from Theorem~\ref{thm:main result}, straightforward computations yield
\begin{align}
\label{eqn:price:ask}
p_a(\psi)
&= V_0 + \wt w_0 \psi + o(\psi).
\end{align}
Therefore, $\wt w_0 \psi$ is the leading-order premium demanded by the agent as a compensation for exposing herself to model uncertainty.

\begin{remark}
\label{rem:asymmetry}
Buying an option is the same as selling the negative of that option. However, the cash equivalents corresponding to $\sfV$ and $-\sfV$ are in general different. This asymmetry is caused by the constraint that the uncorrelated squared volatility must be \emph{nonnegative} and the fact that the reference model has \emph{zero} uncorrelated squared volatility. In other words, the uncorrelated squared volatility can only depart from its reference value in one direction. In contrast, the other control variables can deviate from their reference value in both directions.
\end{remark}

\subsection{On the existence of a candidate asymptotic model family}
\label{sec:existence}

Our main result, Theorem~\ref{thm:main result}, assumes that the set of models $\fP$ contains a candidate asymptotic model family. In this section, we prescribe a set of models $\fP$ and sketch the construction of a candidate asymptotic model family in $\fP$. Fix constants $0 < \ul\Sigma < \Sigma_0 < \ol\Sigma$, $\ul\nu<0<\ol\nu$, $0 < \ul\sigma < \ul\Sigma$, $\ol\Sigma < \ol\sigma$, $\ul\eta<0<\ol\eta$, and $\ol\xi > 0$, and let $\fP$ denote the subset of models $P$ in $\fP^0$ such that the bounds \eqref{eqn:ass:model set} are satisfied. Under some further regularity assumptions on the greeks of the liquid option, $\fP$ then contains a candidate asymptotic model family.

The construction of the candidate asymptotic model family comprises two steps. The first is to prove that the candidate feedback control $\bfzeta^\psi$ can be modified by a term of order $O(\psi^2)$ such that the resulting \emph{modified feedback control} $\check\bfzeta^\psi$ satisfies the drift condition \eqref{eqn:drift condition}: 

\begin{lemma}
\label{lem:existence}
Let $\fP \subset \fP^0$ be such that \eqref{eqn:ass:model set} holds for every $P\in\fP$. Suppose in addition that Assumption~\ref{ass:main result}~\ref{ass:main result:call}--\ref{ass:main result:option} holds with \eqref{eqn:ass:C} in Assumption~\ref{ass:main result}~\ref{ass:main result:call} replaced by the stronger condition that
\begin{align}
\label{eqn:lem:existence:C}
\vert \cC_\Sigma \vert
&\geq 1/K_\cC
\quad\text{and}\quad
\vert S^2 \cC_{SS} \vert, \vert S \cC_{S\Sigma} \vert, \vert \cC_{\Sigma\Sigma} \vert
\leq K_\cC \text{ on } (0,T)\times \RR_+ \times [\ul\Sigma,\ol\Sigma],
\end{align}
for some constant $K_\cC > 0$. Then there are $\psi_0 > 0$ and functions
\begin{align*}
\check\bfzeta^\psi: \bfD^0 \to [\ul\nu,\ol\nu]\times[\ul\sigma,\ol\sigma]\times[\ul\eta,\ol\eta]\times[0,\ol\xi], \quad \psi \in (0,\psi_0),
\end{align*}
such that for each $\psi \in (0,\psi_0)$, the restriction $\left.\check\bfzeta^\psi\right.\vert_{(0,T)\times\bfG\times(\ul\Sigma,\ol\Sigma)}$ is continuous and can be extended to a continuous function on $\bfD^0 = (0,T)\times\bfG\times\RR_+$. Moreover, there is $K_0 > 0$ such that for each $(t,\bfx) = (t,S,A,M,\Sigma) \in \bfD$ and $\psi \in (0,\psi_0)$,
\begin{enumerate}
\item $\check\bfzeta^\psi(t,\bfx) = \bfzeta^0(\Sigma)$ if $\Sigma \not\in(\ul\Sigma,\ol\Sigma)$, i.e., the modified feedback control falls back to the reference feedback control if the bounds on the implied volatility are reached;

\item writing $(\check\nu^\psi,\check\sigma^\psi,\check\eta^\psi,\check\xi^\psi) = \check\bfzeta^\psi(t,\bfx)$, we have
\begin{align*}
\check\nu^\psi \cC_\Sigma + \frac{1}{2}S^2\cC_{SS}((\check\sigma^\psi)^2 - \Sigma^2) + \check\sigma^\psi \check\eta^\psi S \cC_{S\Sigma} + \frac{1}{2}((\check\eta^\psi)^2 + \check\xi^\psi) \cC_{\Sigma\Sigma}
&= 0,
\end{align*}
i.e., the drift condition \eqref{eqn:drift condition} is satisfied for the modified feedback control $\check\bfzeta^\psi$;
\item
\begin{align}
\label{eqn:lem:existence:estimate}
\enorm{\check\bfzeta^\psi(t,\bfx) - \bfzeta^\psi(t,\bfx)}
&\leq K_0\psi^2,
\end{align}
i.e., the modified feedback control $\check\bfzeta^\psi$ is $O(\psi^2)$-close to the candidate $\bfzeta^\psi$.
\end{enumerate}
\end{lemma}

\begin{proof}
See Section~\ref{sec:proofs:existence}.
\end{proof}

Let $\psi_0$, $K_0$, and $\check\bfzeta^\psi$ be as in Lemma~\ref{lem:existence}. The second step now is to show that the stochastic differential equations (SDEs) corresponding to the modified feedback control $\check\bfzeta^\psi = (\check\nu^\psi, \check\sigma^\psi, \check\eta^\psi, \check\xi^\psi)^\tr$ have a weak solution. Fix $\psi \in (0,\psi_0)$. Writing $\nu,\sigma,\eta,\xi$ instead of $\check\nu^\psi, \check\sigma^\psi, \check\eta^\psi, \check\xi^\psi$ to ease the notation, the relevant SDEs read as
\begin{align}
\label{eqn:existence:SDE}
\begin{split}
\diff S'_t
&= S'_t \sigma \dd W^0_t,\\
\diff \Sigma'_t
&= \nu \dd t + \eta \dd W^0_t + \sqrt{\xi} \dd W^1_t,\\
\diff A'_t
&= \left(\alpha + \frac{1}{2}\beta \sigma^2\right) \dd t
+ \gamma S'_t \sigma \dd W^0_t + \delta \dd M'_t,
\end{split}
\end{align}
where $\alpha,\beta,\gamma$, and $\delta$ are evaluated at $(t,S'_t,A'_t,M'_t := \sup_{u\in[0,t]} S'_u)$, $\nu,\sigma,\eta$, and $\xi$ are evaluated at $(t,S'_t,A'_t,M'_t,\Sigma'_t)$, and $(W^0, W^1)$ is a bivariate standard Brownian motion.

Suppose there exists a weak solution to \eqref{eqn:existence:SDE} (starting in $S_0,\Sigma_0,A_0$) with the property that $\Sigma'$ evolves in $[\ul\Sigma,\ol\Sigma]$ almost surely and denote by $P^\psi$ its image measure (under $(S',\Sigma',A')$) on the canonical space $(\Omega,\cF)$. Then by construction (cf.~Definitions~\ref{def:fP00} and \ref{def:fP0}), $P^\psi \in \fP^0$ and $\bfzeta^{P^\psi}_t = \check\bfzeta^\psi(t,\bfX_t)$ $\diff t \times P^\psi$-a.e. Moreover, by Lemma~\ref{lem:existence} and the fact that under $P^\psi$, $\Sigma$ evolves in $[\ul\Sigma,\ol\Sigma]$ almost surely, \eqref{eqn:ass:model set} holds for every $P \in (P^\psi)_{\psi\in(0,\psi_0)}$ and
\begin{align*}
\enorm{\bfzeta^{P^\psi}_t - \bfzeta^\psi(t,\bfX_t)}
&= \enorm{\check\bfzeta^\psi(t,\bfX_t)- \bfzeta^\psi(t,\bfX_t)}
\leq K_0\psi^2 \quad \diff t\times P^\psi\text{-a.e.}
\end{align*}
So $(P^\psi)_{\psi\in(0,\psi_0)}$ is a candidate asymptotic model family in $\fP$.

It remains to argue the existence of a weak solution to \eqref{eqn:existence:SDE} with the property that $\Sigma'$ evolves in $[\ul\Sigma,\ol\Sigma]$. Note that we cannot directly apply standard existence results for weak solutions as the control $\check\bfzeta^\psi$ is not continuous in $\Sigma \in \RR_+$. However, one can apply a standard existence result to the SDEs corresponding to the continuous extension of $\left.\check\bfzeta^\psi\right.\vert_{(0,T)\times\bfG\times(\ul\Sigma,\ol\Sigma)}$ to $\bfD^0$. Then the obvious idea is to stop the resulting weak solution as soon as $\Sigma'$ hits the boundary of $[\ul\Sigma,\ol\Sigma]$ and restart the SDEs with new dynamics from there. After the restart, we keep $\Sigma' \in \lbrace \ul\Sigma,\ol\Sigma \rbrace$ constant, let $S'$ evolve like a standard Black--Scholes model with constant volatility $\Sigma'$, and (assuming suitable Lipschitz and linear growth conditions on the coefficients of the SDE for $A$; cf.~\cite[Appendix B]{HerrmannMuhleKarbeSeifried2017}) find a solution $A'$ according to the dynamics in \eqref{eqn:existence:SDE}, but with the new dynamics of $S'$. Then one can check that the constructed process satisfies the SDEs \eqref{eqn:existence:SDE} with the original feedback control $\check\bfzeta^\psi$; see \cite[Theorem 3.7]{HerrmannMuhleKarbeSeifried2017} for more details in a similar setup.

\section{Proofs}
\label{sec:proofs}

This section contains the proofs of our main results. We first establish the value expansion and almost-optimality of the delta-vega hedge asserted in Theorem~\ref{thm:main result}. Afterwards, we turn to the construction of the modified feedback control from Lemma~\ref{lem:existence}.

\subsection{Value expansion and almost optimality of the delta-vega hedge}
\label{sec:proofs:main result}

In this section, we prove Theorem~\ref{thm:main result}. Throughout, we assume that Assumption~\ref{ass:main result} is in force, that $\bfupsilon^\star \in \fY$, and that  $(P^\psi)_{\psi\in(0,\psi_0)} \subset \fP$ is a candidate asymptotic model family.\footnote{Recall from Remark~\ref{rem:ass:main result}~(a) that the delta-vega hedge $\bfupsilon^\star$ can always be included into the set of trading strategies $\fY$ by making the constant $K_\fY$ from Assumption~\ref{ass:main result}~\ref{ass:main result:strategy set} larger if necessary. The existence of a candidate asymptotic model family is discussed in Section~\ref{sec:existence}.}
In particular (recall Definition~\ref{def:candidate asymptotic model family}~(b)), we fix $1 \leq K_0 \in L^4_\fP$ such that for every $\psi \in (0,\psi_0)$,
\begin{align}
\label{eqn:pf:estimate:candidate asymptotic model family}
\enorm{\bfzeta^{P^\psi}_t - \bfzeta^\psi(t,\bfX_t)}
&\leq K_0(t,\bfX_t) \psi^2 \quad \diff t \times P^\psi\text{-a.e.}
\end{align}

For each $\psi>0$, define the \emph{candidate value function} $w^\psi:\ol{\bfD^0}\times\RR\to\RR$ by
\begin{align}
\label{eqn:value function}
w^\psi(t,\bfx,y)
&= U(y) - U'(y)\wt w(t,\bfx)\psi
\end{align}
and set $w^\psi_0 := w^\psi(0,S_0,A_0,M_0,\Sigma_0,Y_0)$. Suppose for the moment that we have already proved the following two inequalities (cf.~Lemmas~\ref{lem:lower bound} and \ref{lem:upper bound}):
\begin{align}
\inf_{P \in \fP} J^\psi(\bfupsilon^\star,P)
&\geq w^\psi_0 + o(\psi), \qquad \text{as } \psi \downarrow 0,\label{eqn:SDG inequalities1}\\
\sup_{\bfupsilon\in\fY} J^\psi(\bfupsilon,P^\psi)
&\leq w^\psi_0 + o(\psi),
\qquad \text{as } \psi \downarrow 0.\label{eqn:SDG inequalities2}
\end{align}
Denoting by $\lesssim$ ``less or equal up to a term of order $o(\psi)$'', we obtain from \eqref{eqn:SDG inequalities1}--\eqref{eqn:SDG inequalities2} that
\begin{align*}
w^\psi_0
&\lesssim \inf_{P \in \fP} J^\psi(\bfupsilon^\star,P)
\lesssim \sup_{\bfupsilon \in \fY} \inf_{P \in \fP} J^\psi(\bfupsilon,P)
\lesssim \inf_{P \in \fP} \sup_{\bfupsilon \in \fY} J^\psi(\bfupsilon,P)
\lesssim \sup_{\bfupsilon \in \fY} J^\psi(\bfupsilon,P^\psi)
\lesssim w^\psi_0
\end{align*}
and
\begin{align*}
w^\psi_0
&\lesssim \inf_{P \in \fP} J^\psi(\bfupsilon^\star,P)
\lesssim J^\psi(\bfupsilon^\star,P^\psi)
\lesssim \sup_{\bfupsilon \in \fY} J^\psi(\bfupsilon,P^\psi)
\lesssim w^\psi_0.
\end{align*}
Hence, we have equality up to a term of order $o(\psi)$ everywhere. In particular, assertion~\eqref{eqn:thm:main result:expansion} of Theorem~\ref{thm:main result} holds. This completes the proof of Theorem~\ref{thm:main result} modulo the proof of \eqref{eqn:SDG inequalities1}--\eqref{eqn:SDG inequalities2}. The proof of these two inequalities is based on careful estimates of the HJBI equation associated to the SDG~\eqref{eqn:value}. Section~\ref{sec:proofs:notation} introduces the notation used in the rest of the proof as well as some preliminary results. Sections~\ref{sec:proofs:Hamiltonian}--\ref{sec:proofs:candidate value function} are purely analytic and provide the required estimates of the HJBI equation. Finally, Sections~\ref{sec:proofs:lower bound} and \ref{sec:proofs:upper bound} contain the proofs of the inequalities \eqref{eqn:SDG inequalities1} and \eqref{eqn:SDG inequalities2}.

\subsubsection{Notation and preliminaries}
\label{sec:proofs:notation}

Set $\psi_{\min} = \min(\psi_\nu,\psi_\sigma,\psi_\eta,\psi_\xi)$, $\psi_{\max} = \max(\psi_\nu,\psi_\sigma,\psi_\eta,\psi_\xi)$ (recall \eqref{eqn:Psi}), and denote by $\norm{Q}_F$ the Frobenius norm of a matrix $Q$.
Recalling that the \emph{squared} uncorrelated volatility $\xi^P$ has to be nonnegative, let $\bfZ^0 := \RR^3 \times [0,\infty)$ be the natural range for the controls $\bfzeta^P$. A generic element of $\bfZ^0$ is always denoted by $\bfzeta = (\nu,\sigma,\eta,\xi)^\tr$. Next, define the function $b^\cC:\bfD^0 \times \bfZ^0 \to \RR$ by
\begin{align}
b^\cC(t,\bfx;\bfzeta)
&= \nu \cC_\Sigma
+\frac{1}{2}S^2\cC_{SS}(\sigma^2 - \Sigma^2)
+\sigma\eta S\cC_{S\Sigma}
+\frac{1}{2}(\eta^2 + \xi)\cC_{\Sigma\Sigma}\notag\\
\label{eqn:bC:matrix form}
&= \bfc(t,\bfx)^\tr (\bfzeta-\bfzeta^0(\Sigma))
+\frac{1}{2}
\begin{pmatrix}
\sigma - \Sigma\\
\eta
\end{pmatrix}^\tr
\begin{pmatrix}
S^2 \cC_{SS}    & S \cC_{S\Sigma}\\
S \cC_{S\Sigma} & \cC_{\Sigma\Sigma}
\end{pmatrix}
\begin{pmatrix}
\sigma - \Sigma\\
\eta
\end{pmatrix};
\end{align}
cf.~\eqref{eqn:VGVV:c} for the definition of the vega-gamma-vanna-volga vector $\bfc(t,\bfx)$ of the call. This definition is motivated by the drift condition \eqref{eqn:drift condition}, which states that $b^\cC(t,\bfX_t;\bfzeta^P_t) = 0$ $\diff t\times P$-a.e.~for every $P\in\fP^0$. For each $(t,\bfx) \in \bfD$, write
\begin{align*}
\bfZ^0(t,\bfx)
&= \lbrace \bfzeta \in \bfZ^0 : b^\cC(t,\bfx;\bfzeta) = 0 \rbrace
\end{align*}
for the set of controls $\bfzeta$ that fulfil the drift condition at $(t,\bfx)$, and define
\begin{align}
\label{eqn:Z0lin}
\bfZ^0_\lin(t,\bfx)
&= \lbrace \bfzeta \in \bfZ^0 : \bfc(t,\bfx)^\tr(\bfzeta -\bfzeta^0(\Sigma))= 0\rbrace,
\end{align}
the set of controls $\bfzeta$ that satisfy the ``linearised drift condition'' at $(t,\bfx)$.
Next, set
\begin{align*}
\bfZ
&= [\ul\nu,\ol\nu]\times[\ul\sigma,\ol\sigma]\times[\ul\eta,\ol\eta]\times[0,\ol\xi]
\end{align*}
for the range of the controls in $\fP$ (cf.~Assumption~\ref{ass:main result}~\ref{ass:main result:model set}) and denote by $\bfZ(t,\bfx) = \bfZ^0(t,\bfx) \cap \bfZ$ and $\bfZ_\lin(t,\bfx) = \bfZ^0_\lin(t,\bfx) \cap \bfZ$ the intersections of $\bfZ^0(t,\bfx)$ and $\bfZ^0_\lin(t,\bfx)$ with $\bfZ$, respectively. Also recall from Definition~\ref{def:reference model} that the reference feedback control is $\bfzeta^0(\Sigma) = (0,\Sigma,0,0)^\tr$.

We start with the probabilistic representation of the solution to the PDE \eqref{eqn:cash equivalent:PDE} for the cash equivalent $\wt w_0 = \wt w(0,\bfX_0)$.

\begin{proposition}[Feynman--Kac representation]
\label{prop:Feynman-Kac raw}
Let $P^0 \in \fP$ be a reference model. Then
\begin{align}
\label{eqn:prop:Feynman-Kac raw}
\wt w_0
= \wt w(0,\bfX_0)
&= \frac{1}{2}\EX[P^0]{\int_0^T \wt g(t,\bfX_t) \dd t}.
\end{align}
\end{proposition}

\begin{proof}
We only sketch the standard proof. Applying It\^o's formula to $\wt w(t,\bfX_t)$ under $P^0$ and using the PDE \eqref{eqn:cash equivalent:PDE} for $\wt w$ shows that
\begin{align*}
\wt w(0,\bfX_0)
&= \frac{1}{2}\int_0^T \wt g(t,\bfX_t) \dd t + \text{(local martingale)}.
\end{align*}
Using Assumption~\ref{ass:main result}~\ref{ass:main result:cash equivalent}, the local martingale term is easily shown to be a martingale. Hence, taking expectations yields the Feynman--Kac representation \eqref{eqn:prop:Feynman-Kac raw}.
\end{proof}

The next lemma provides the dynamics of the P\&L processes:

\begin{lemma}
\label{lem:Y}
Let $\bfupsilon = (\theta,\phi) \in \fY$ and $P\in\fP$. Then under $P$,
\begin{align}
\label{eqn:lem:Y:dynamics}
\begin{split}
\diff Y^{\bfupsilon,P}_t
&= \big(\theta_t - (\Delta(t,\bfX_t) - \phi_t \cC_S(t,S_t,\Sigma_t)) \big) \dd S_t + \big(\phi_t \cC_\Sigma(t,S_t,\Sigma_t) -\cV_\Sigma(t,\bfX_t) \big) \dd \Sigma^{\contlocmartpart,P}_t\\
&\qquad -b^\cV(t,\bfX_t;\bfzeta^P_t) \dd t.
\end{split}
\end{align}
Here,
\begin{align*}
\Sigma^{\contlocmartpart,P}
&= \Sigma - \int_0^\cdot \nu^P_u \dd u
\end{align*}
is the (continuous) local martingale part of $\Sigma$ under $P$ and $b^\cV:\bfD^0\times\bfZ^0 \to \RR$ is given by
\begin{align}
\label{eqn:lem:Y:bV}
\begin{split}
b^\cV(t,\bfx;\bfzeta)
&= \nu\cV_{\Sigma}
+ \frac{1}{2}(\beta \cV_A + S^2 \Gamma)(\sigma^2-\Sigma^2)
+ \sigma\eta  S \Vanna + \frac{1}{2}(\eta^2 + \xi) \cV_{\Sigma\Sigma}\\
&= \bfv(t,\bfx)^\tr (\bfzeta-\bfzeta^0(\Sigma))
+\frac{1}{2}
\begin{pmatrix}
\sigma - \Sigma\\
\eta
\end{pmatrix}^\tr
\begin{pmatrix}
\beta \cV_A + S^2 \Gamma    & S \frac{\partial\Delta}{\partial\Sigma}\\
S \frac{\partial\Delta}{\partial\Sigma} & \cV_{\Sigma\Sigma}
\end{pmatrix}
\begin{pmatrix}
\sigma - \Sigma\\
\eta
\end{pmatrix},
\end{split}
\end{align}
where $\bfv$ is the vega-gamma-vanna-volga vector of the non-traded option (cf.~\eqref{eqn:VGVV:v}).
\end{lemma}

\begin{proof}
Fix $\bfupsilon = (\theta,\phi) \in \fY$ and $P \in \fP$ and recall from \eqref{eqn:Y:definition} and \eqref{eqn:V:definition} that
\begin{align}
\label{eqn:lem:Y:pf:Y dynamics}
\diff Y^{\bfupsilon,P}_t
&= \theta_t \dd S_t + \phi_t \dd C_t - \diff V_t,
\end{align}
where $V_t = \cV(t,S_t,A_t,M_t,\Sigma_t)$. Thus, it remains to compute the dynamics of $C$ and $V$ under $P$.

First, by \eqref{eqn:C:definition}, It\^o's formula (under $P$), and the drift condition \eqref{eqn:drift condition}, we have
\begin{align}
\label{eqn:lem:Y:pf:C dynamics}
\diff C_t
&= \cC_S \dd S_t + \cC_\Sigma \dd \Sigma^{\contlocmartpart,P}_t.
\end{align}
Second, applying It\^o's formula to $V_t = \cV(t,S_t,A_t,M_t,\Sigma_t)$ and using the PDE \eqref{eqn:V:PDE} to substitute $\cV_t = \cV_t(t,\bfX_t)$ and to eliminate the $\diff M_t$-term, we arrive at
\begin{align}
\diff V_t
\label{eqn:lem:Y:pf:V dynamics:2}
&= \Delta \dd S_t + \cV_\Sigma \dd \Sigma^{\contlocmartpart,P}_t + b^\cV(\bfzeta^P_t) \dd t.
\end{align}
Finally, inserting \eqref{eqn:lem:Y:pf:C dynamics} and \eqref{eqn:lem:Y:pf:V dynamics:2} into \eqref{eqn:lem:Y:pf:Y dynamics} yields \eqref{eqn:lem:Y:dynamics}. The last equality in the definition \eqref{eqn:lem:Y:bV} of $b^\cV$ is the Taylor expansion of $b^\cV(\bfzeta)$ around $\bfzeta^0(\Sigma)$ and can be verified by computing the gradient and the Hessian of $b^\cV(\bfzeta)$ at $\bfzeta^0(\Sigma)$.
\end{proof}

We next analyse the dynamics of the P\&L process $Y^{\bfupsilon^\star,P}$ corresponding to the delta-vega hedge~$\bfupsilon^\star$. To this end we define, for each $(t,\bfx) \in \bfD^0$:
\begin{align}
\label{eqn:delta-vega hedge:feedback}
\bfupsilon^\star(t,\bfx)
&= \left( \Delta - \frac{\cV_\Sigma}{\cC_\Sigma} \cC_S, \frac{\cV_\Sigma}{\cC_\Sigma} \right).
\end{align}
Note that, with a slight abuse of notation, we use the symbol $\bfupsilon^\star$ both for the function defined in \eqref{eqn:delta-vega hedge:feedback} and the delta-vega hedge defined in Theorem~\ref{thm:main result}. This is, of course, motivated by the relationship $\bfupsilon^\star_t = \bfupsilon^\star(t,\bfX_t)$.\footnote{With a slight abuse of notation, $\bfupsilon^\star_t$ always denotes the time-$t$ value of the \emph{process} $\bfupsilon^\star$ and not the partial derivative of the \emph{function} $\bfupsilon^\star$ with respect to the first variable.} The following corollary to Lemma~\ref{lem:Y} shows that the P\&L process $Y^{\bfupsilon^\star,P}$ corresponding to the  delta-vega hedge $\bfupsilon^\star$ has no local martingale part and is bounded, uniformly in $P\in\fP$.

\begin{corollary}
\label{cor:Y:delta-vega}
There are constants $\ul Y, \ol Y \in \RR$ such that for each $P \in \fP$,
\begin{align*}
Y^{\bfupsilon^\star,P} \in [\ul Y, \ol Y]
\quad \diff t \times P\text{-a.e.}
\end{align*}
Moreover, under each $P \in \fP$,
\begin{align*}
\diff Y^{\bfupsilon^\star,P}_t
&= -b^\cV(t,\bfX_t;\bfzeta^P_t) \dd t,
\end{align*}
where $b^{\cV}$ is defined in \eqref{eqn:lem:Y:bV}.
\end{corollary}

\begin{proof}
By construction of $\bfupsilon^\star_t = \bfupsilon^\star(t,\bfX_t)$, the local martingale part in the dynamics \eqref{eqn:lem:Y:dynamics} of $Y^{\bfupsilon^\star,P}$ is zero for each $P \in \fP$. Thus, it suffices to find a uniform bound (independent of $P \in \fP$) for the drift coefficient $b^\cV(t,\bfX_t;\bfzeta^P_t)$. But this is immediate from Assumption~\ref{ass:main result}~\ref{ass:main result:model set} and \ref{ass:main result:option}.
\end{proof}

Lemma~\ref{lem:Y} together with the covariations of $S$ and $\Sigma$ in \eqref{eqn:covariation:S and IV} and the semimartingale decomposition \eqref{eqn:A:dynamics} of $A$ specifies the joint dynamics of the process $(S,A,M,\Sigma,Y^{\bfupsilon,P})$ under $P\in\fP$. This allows to write down the Hamilton--Jacobi--Bellman--Isaacs equation corresponding to the SDG \eqref{eqn:value}: for each $\psi > 0$, the HJBI equation reads as
\begin{align}
\label{eqn:HJBI}
w^\psi_t(t,\bfx,y) + \sup_{\bfupsilon \in \RR^2}\inf_{\bfzeta \in \bfZ(t,\bfx)} H^\psi(t,\bfx,y;\bfupsilon,\bfzeta)
&= 0,
\end{align}
where the \emph{Hamiltonian}\footnote{Note that the definition of $H^\psi$ already contains the candidate first-order expansion $w^\psi$ of the value function and thus does not feature a general solution function and its derivatives as arguments.} $H^\psi:\bfD^0\times\RR\times\RR^2\times\bfZ^0 \to \RR$ is given by
\begin{align}
H^\psi(t,\bfx,y;\bfupsilon,\bfzeta)
&= \frac{1}{\psi}U'(y)f(\Sigma,\bfzeta)
+ \nu w^\psi_\Sigma + (\alpha + \frac{1}{2}\beta \sigma^2) w^\psi_A - b^\cV(\bfzeta) w^\psi_Y\notag\\
&\qquad+ \frac{1}{2} \sigma^2 S^2(w^\psi_{SS} + 2\gamma w^\psi_{SA} + \gamma^2 w^\psi_{AA})\notag\\
&\qquad+ \sigma S \eta (w^\psi_{S\Sigma} + \gamma w^\psi_{A\Sigma})\notag\\
&\qquad+ \frac{1}{2}(\eta^2 + \xi) w^\psi_{\Sigma\Sigma}\notag\\
&\qquad+ \sigma^2 S^2 [\theta - (\Delta - \phi \cC_S)] (w^\psi_{SY} + \gamma w^\psi_{AY})
+ \sigma S \eta [\phi \cC_\Sigma - \cV_\Sigma](w^\psi_{SY} + \gamma w^\psi_{AY})\notag\\
&\qquad+ \sigma S \eta [\theta - (\Delta - \phi \cC_S)] w^\psi_{\Sigma Y}
+ [\phi \cC_\Sigma - \cV_\Sigma](\eta^2 + \xi) w^\psi_{\Sigma Y}\notag\\
&\qquad+ \frac{1}{2}\sigma^2 S^2 [\theta - (\Delta - \phi \cC_S)]^2 w^\psi_{YY}
+ \frac{1}{2} (\eta^2 + \xi) [\phi \cC_\Sigma - \cV_\Sigma]^2 w^\psi_{YY}\notag\\
\label{eqn:Hamiltonian}
&\qquad\qquad+ \sigma S \eta [\theta - (\Delta - \phi \cC_S)] [\phi \cC_\Sigma - \cV_\Sigma] w^\psi_{YY}.
\end{align}
We emphasise that our candidate value function $w^\psi$ defined in \eqref{eqn:value function} does not solve the HJBI equation \eqref{eqn:HJBI} \emph{exactly}. However, a key step in the proof of the two inequalities \eqref{eqn:SDG inequalities1}--\eqref{eqn:SDG inequalities2} is to show that $w^\psi$ is \emph{asymptotically} (in a suitable sense) a solution to \eqref{eqn:HJBI}; cf.~Lemmas~\ref{lem:HJBI:lower bound} and \ref{lem:HJBI:upper bound} below.

We close this preliminary section by providing an auxiliary lemma that allows to estimate quantities like $b^\cV(t,\bfx;\bfzeta)$ or $b^\cC(t,\bfx;\bfzeta)$ in terms of $\enorm{\bfzeta - \bfzeta^0(\Sigma)}$.

\begin{lemma}
\label{lem:general drift function}
Define the function $q:\RR_+\times\RR^4\times\RR^4 \to \RR$ by
\begin{align}
\label{eqn:lem:general drift function:q}
q(\Sigma,\bfa,\bfzeta)
&= \nu a_1 + \frac{1}{2} a_2 (\sigma^2 - \Sigma^2) + \sigma\eta a_3 + \frac{1}{2}(\eta^2 + \xi)a_4,
\end{align}
where $\bfa = (a_1,a_2,a_3,a_4)^\tr$ and $\bfzeta = (\nu, \sigma, \eta, \xi)^\tr$. Then: 
\begin{align*}
\vert q(\Sigma,\bfa,\bfzeta) \vert 
&\leq \max(1,\Sigma) \enorm{\bfa} \enorm{\bfzeta - \bfzeta^0(\Sigma)} + \enorm{\bfa} \enorm{\bfzeta - \bfzeta^0(\Sigma)}^2.
\end{align*}
\end{lemma}

\begin{proof}
Fix $\Sigma \in \RR_+$, $\bfa \in \RR^4$, and $\bfzeta\in\RR^4$. As $q$ is quadratic in $\bfzeta$, we can recast it in matrix form:
\begin{align}
\label{eqn:lem:general drift function:pf:matrix form}
q(\Sigma,\bfa,\bfzeta)
&=
\begin{pmatrix}
a_1\\ \Sigma a_2 \\ \Sigma a_3 \\ \frac{1}{2} a_4
\end{pmatrix}^\tr
\left(\bfzeta - \bfzeta^0(\Sigma)\right)
+ \frac{1}{2}
\begin{pmatrix}
\sigma - \Sigma\\\eta
\end{pmatrix}^\tr
\begin{pmatrix}
a_2 & a_3 \\ a_3 & a_4
\end{pmatrix}
\begin{pmatrix}
\sigma - \Sigma\\\eta
\end{pmatrix}.
\end{align}
Using the Cauchy--Schwarz inequality, the absolute value of the first summand on the right-hand side of \eqref{eqn:lem:general drift function:pf:matrix form} is easily estimated from above by $\max(1,\Sigma) \enorm{\bfa} \enorm{\bfzeta - \bfzeta^0(\Sigma)}$. Likewise, using also the compatibility of the Frobenius norm with the Euclidean norm, the absolute value of the second summand is dominated by
\begin{align*}
&\frac{1}{2}
\norm{
\begin{pmatrix}
a_2 & a_3 \\ a_3 & a_4
\end{pmatrix}
}_F
\enorm{
\begin{pmatrix}
\sigma - \Sigma\\\eta
\end{pmatrix}
}^2
\leq \enorm{\bfa}
\enorm{\bfzeta - \bfzeta^0(\Sigma)}^2.\qedhere
\end{align*}
\end{proof}

\subsubsection{Estimates for the Hamiltonian}
\label{sec:proofs:Hamiltonian}

In order to prove that the candidate value function is -- asymptotically -- a solution to the HJBI equation \eqref{eqn:HJBI}, we need several estimates for the Hamiltonian $H^\psi$ defined in \eqref{eqn:Hamiltonian}. To this end, we decompose it into four parts: 
\begin{align}
\label{eqn:Hamiltonian:decomposition}
H^\psi(t,\bfx,y;\bfupsilon,\bfzeta)
&= U'(y) H^\psi_1(t,\bfx;\bfzeta)
+ U'(y) H^\psi_2(t,\bfx,y;\bfzeta)
- U'(y) H_3(t,\bfx;\bfzeta)\psi 
- H^\psi_4(t,\bfx;\bfupsilon,\bfzeta),
\end{align}
where
\begin{align}
\label{eqn:H1}
H^\psi_1(t,\bfx;\bfzeta)
&:=  \frac{1}{2\psi}(\bfzeta-\bfzeta^0(\Sigma))^\tr \Psi^{-1} (\bfzeta-\bfzeta^0(\Sigma))
- \bfv(t,\bfx)^\tr (\bfzeta - \bfzeta^0(\Sigma)),\\
\notag
H^\psi_2(t,\bfx,y;\bfzeta)
&:= -\frac{1}{2}
\begin{pmatrix}
\sigma - \Sigma\\
\eta
\end{pmatrix}^\tr
\begin{pmatrix}
\beta \cV_A + S^2 \Gamma                & S \frac{\partial\Delta}{\partial\Sigma}\\
S \frac{\partial\Delta}{\partial\Sigma} & \cV_{\Sigma\Sigma}
\end{pmatrix}
\begin{pmatrix}
\sigma - \Sigma\\
\eta
\end{pmatrix}
+ \frac{U''(y)}{U'(y)} b^\cV(\bfzeta)\wt w \psi,\\
\notag
\begin{split}
H_3(t,\bfx;\bfzeta)
&:= \nu \wt w_\Sigma
+ (\alpha +\frac{1}{2}\beta\sigma^2) \wt w_A
+ \frac{1}{2}\sigma^2 S^2 (\wt w_{SS} + 2\gamma \wt w_{SA} + \gamma^2 \wt w_{AA})\\
&\qquad+\sigma \eta S (\wt w_{S\Sigma}
+ \gamma \wt w_{A\Sigma})
+ \frac{1}{2}(\eta^2 + \xi) \wt w_{\Sigma\Sigma},
\end{split}\\
\notag
\begin{split}
H^\psi_4(t,\bfx,y;\bfupsilon,\bfzeta)
&:= -\frac{w^\psi_{YY}}{2}
\begin{pmatrix}
\sigma S (\theta - (\Delta - \phi \cC_S))\\
\phi \cC_\Sigma - \cV_\Sigma
\end{pmatrix}^\tr
\begin{pmatrix}
1 & \eta\\
\eta & \eta^2 + \xi
\end{pmatrix}
\begin{pmatrix}
\sigma S (\theta - (\Delta - \phi \cC_S))\\
\phi \cC_\Sigma - \cV_\Sigma
\end{pmatrix}\\
&\qquad+ \psi U''(y)
\begin{pmatrix}
\sigma S (\wt w_S + \gamma \wt w_A)\\
\wt w_\Sigma
\end{pmatrix}^\tr
\begin{pmatrix}
1 & \eta\\
\eta & \eta^2 + \xi
\end{pmatrix}
\begin{pmatrix}
\sigma S (\theta - (\Delta - \phi \cC_S))\\
\phi \cC_\Sigma - \cV_\Sigma
\end{pmatrix}.
\end{split}
\end{align}
$H^\psi_1$ includes the penalty term (cf.~the definition of $f$ in \eqref{eqn:penalty function}) and the linear $O(1)$ part of $b^\cV(\bfzeta)w^\psi_Y$; $H^\psi_2$ contains the quadratic $O(1)$ part and the $O(\psi)$ part of $b^\cV(\bfzeta)w^\psi_Y$; $H^\psi_4$ collects all second-order partial derivatives of $w^\psi$ that involve at least one partial derivative with respect to $Y$; and $H_3$ takes care of all remaining partial derivatives of $w^\psi$.

For later reference, we note that by the definition of $H_3$ and $\bfzeta^0$, the PDE \eqref{eqn:cash equivalent:PDE} for $\wt w$ can be written as
\begin{align}
\label{eqn:cash equivalent:PDE:H3}
\wt w_t(t,\bfx) + H_3(t,\bfx;\bfzeta^0(\Sigma)) + \frac{1}{2}\wt g(t,\bfx)
&=0\quad \text{for } (t,\bfx)\in\bfD.
\end{align}
Moreover, for every $(t,\bfx,y)\in\bfD \times \RR$ and $\bfzeta \in \bfZ$,
\begin{align}
\label{eqn:H1 H2 H4:zero}
H^\psi_1(t,\bfx;\bfzeta^0(\Sigma))
&= 0,\quad
H^\psi_2(t,\bfx;\bfzeta^0(\Sigma))
= 0,\quad
H^\psi_4(t,\bfx,y;\bfupsilon^\star(t,\bfx),\bfzeta)
=0,
\end{align}
by construction of the reference feedback control $\bfzeta^0$ (cf.~Definition~\ref{def:reference model}) and the delta-vega hedge $\bfupsilon^\star$ (cf.~\eqref{eqn:delta-vega hedge:feedback}).

\begin{remark}
\label{rem:proof}
Recall that the HJBI equation \eqref{eqn:HJBI} involves a minimisation over $\bfzeta\in \bfZ(t,\bfx)$ and a maximisation over $\bfupsilon \in \RR^2$. The strategy variable $\bfupsilon$ only shows up in the $H^\psi_4$ term. Moreover, using an ansatz of the form $\bfzeta = \bfzeta^0(\Sigma) + \wt\bfzeta \psi$, one can check that $\wt\bfzeta$ only affects the $O(\psi)$ term of $H^\psi$ through $H^\psi_1$ (provided that $\bfupsilon = \bfupsilon^\star$ so that the $H^\psi_4$ term vanishes; cf.~\eqref{eqn:H1 H2 H4:zero}). The impact of $\wt\bfzeta$ through $H^\psi_2$ and $H_3$ only appears at higher orders. This distinction is reflected in the proofs of this section as follows.

On the one hand, the estimates for the terms $H^\psi_2$ and $H_3$ in Propositions~\ref{prop:HJBI:H2}--\ref{prop:HJBI:H3} and Corollary~\ref{cor:HJBI:H2 and H3} are rather direct and provide simultaneously asymptotic upper and lower bounds. On the other hand, the proofs of the estimates for $H^\psi_4$ and, in particular, $H^\psi_1$ are more difficult as the corresponding bounds arise from optimisation problems over the strategy variables and the controls, respectively. The asymptotic bound for $H^\psi_4$ is the easier one because $H^\psi_4$ is quadratic in $\bfupsilon$ and the optimisation is unconstrained. In contrast, the asymptotic bound for $H^\psi_1$ in Proposition~\ref{prop:HJBI:H1} arises from the linearly constrained quadratic programming problem \eqref{eqn:heuristic:LCQP problem}. An additional difficulty stems from the fact that we need this bound to hold for controls $\bfzeta \in \bfZ(t,\bfx)$ that satisfy the nonlinear contraint $b^\cC(t,\bfx;\bfzeta) = 0$ instead of the linear one (cf.~Proposition~\ref{prop:HJBI:H1}~(a)).
\end{remark}

We first provide the solution to a linearly constrained quadratic programming problem  involving the vega-gamma-vanna-volga vectors $\bfc(t,\bfx)$ and $\bfv(t,\bfx)$ that lies at the core of the minimisation part of the HJBI equation. In particular, the candidate feedback control $\bfzeta^\psi$ (cf.~\eqref{eqn:candidate control:short}) is a suitably modified version of the minimiser $\bfzeta^{\psi*}$ (cf.~\eqref{eqn:lem:HJBI:H1LCQP:minimiser} below) of this quadratic programming problem; both controls differ only by the indicator $\1_{\lbrace \ul\Sigma < \Sigma < \ol\Sigma \rbrace}$ that ensures that $\bfzeta^{\psi}$ falls back to the reference feedback control $\bfzeta^0(\Sigma)$ once the implied volatility hits the boundary of $[\ul\Sigma,\ol\Sigma]$. Recall the definitions of $\Psi$, $\bfc$, $\bfv$, $\lambda$, $\mu$, $\wt\bfzeta$, and $\bfzeta^\psi$ in \eqref{eqn:Psi} and \eqref{eqn:VGVV:c}--\eqref{eqn:candidate control:short}.

\begin{lemma}
\label{lem:HJBI:H1LCQP}
For each $(t,\bfx) \in \bfD$ and $\psi > 0$, consider the linearly constrained minimisation problem
\begin{align}
\label{eqn:lem:HJBI:H1LCQP:problem}
\text{minimise } H^\psi_1(t,\bfx;\bfzeta)
\quad\text{subject to } \bfzeta \in \bfZ^0_\lin(t,\bfx).
\end{align}
\begin{enumerate}
\item For each $(t,\bfx) \in \bfD$ and $\psi > 0$, we have
\begin{align}
\label{eqn:lem:HJBI:H1LCQP:minimum}
\min_{\bfzeta \in \bfZ^0_\lin(t,\bfx)} H^\psi_1(t,\bfx;\bfzeta)
&= -\frac{1}{2}\wt g(t,\bfx) \psi
\end{align}
and the minimum is attained at
\begin{align}
\label{eqn:lem:HJBI:H1LCQP:minimiser}
\bfzeta^{\psi*}(t,\bfx)
&= \bfzeta^0(\Sigma) + \psi \wt\bfzeta(t,\bfx).
\end{align}
In particular, as $\bfzeta^{\psi*} \in \bfZ^0_\lin(t,\bfx)$, $\bfc(t,\bfx)^\tr \wt\bfzeta(t,\bfx) = 0$ and $\vec\bfe_4^\tr\wt\bfzeta(t,\bfx)\geq 0$.

\item For each $(t,\bfx)\in\bfD$ and $\psi > 0$, $(\lambda(t,\bfx),\mu(t,\bfx))$ is a Lagrange multiplier for \eqref{eqn:lem:HJBI:H1LCQP:problem} (independent of $\psi$), i.e.,
\begin{align*}
-\frac{1}{2}\wt g(t,\bfx) \psi 
&= \inf_{\bfzeta\in\RR^4} L^\psi_1(t,\bfx;\bfzeta,\lambda(t,\bfx),\mu(t,\bfx)),
\end{align*}
where
\begin{align}
\label{eqn:lem:HJBI:H1LCQP:Lagrangian}
L^\psi_1(t,\bfx;\bfzeta,\lambda',\mu')
&= H^\psi_1(t,\bfx;\bfzeta) + \lambda' \bfc(t,\bfx)^\tr (\bfzeta - \bfzeta^0(\Sigma)) - \mu' \vec\bfe_4^\tr (\bfzeta-\bfzeta^0(\Sigma))
\end{align}
is the Lagrangian corresponding to the constrained minimisation problem \eqref{eqn:lem:HJBI:H1LCQP:problem}.

\item There is $K_{\wt g} > 0$ such that $0 \leq \wt g \leq K_{\wt g}$ on $\bfD$.

\item There is $K_\bfzeta \geq 1$ such that for every $(t,\bfx) \in \bfD$ and $\psi > 0$,
\begin{align*}
\enorm{\bfzeta^\psi(t,\bfx) - \bfzeta^0(\Sigma)}
&\leq K_\bfzeta \psi.
\end{align*}

\item For every $(t,\bfx)\in\bfD$ and $\psi > 0$, 
\begin{align*}
H^\psi_1(t,\bfx;\bfzeta^\psi(t,\bfx))
&= -\frac{1}{2}\wt g(t,\bfx) \1_{\lbrace \Sigma \in (\ul\Sigma,\ol\Sigma)\rbrace} \psi.
\end{align*}

\item There is $K_\lambda \in L^2_\fP$ such that for every $(t,\bfx)\in\bfD$,
\begin{align}
\label{eqn:lem:HJBI:H1LCQP:K lambda}
\vert \lambda(t,\bfx) \vert
\norm{
\begin{pmatrix}
S^2 \cC_{SS}    & S \cC_{S\Sigma}\\
S \cC_{S\Sigma} & \cC_{\Sigma\Sigma}
\end{pmatrix}
}_F
&\leq K_\lambda(t,\bfx).
\end{align}
\end{enumerate}
\end{lemma}

\begin{proof}
Recalling the definitions of $H^\psi_1(t,\bfx;\bfzeta)$ (cf.~\eqref{eqn:H1}) and $\bfZ^0_\lin(t,\bfx)$ (cf.~\eqref{eqn:Z0lin}) and using the substitution $\bfz = \bfzeta - \bfzeta^0(\Sigma)$, it is easy to see that for each $(t,\bfx)\in\bfD$ and $\psi > 0$, the minimisation problem \eqref{eqn:lem:HJBI:H1LCQP:problem} can be recast as
\begin{align}
\label{eqn:lem:HJBI:H1LCQP:pf:recast}
\text{minimise } \frac{1}{2\psi}\bfz^\tr \Psi^{-1} \bfz - \bfv(t,\bfx)^\tr \bfz
\quad\text{subject to } \bfz \in \RR^4, \bfz^\tr \bfc(t,\bfx) = 0, z_4 \geq 0.
\end{align}
Note that \eqref{eqn:lem:HJBI:H1LCQP:pf:recast} is a linearly constrained minimisation problem of the form \eqref{eqn:lem:LCQP:primal problem} with $n=4$, $D=\Psi^{-1}/\psi$, $\bfv = \bfv(t,\bfx)$, and $\bfc = \bfc(t,\bfx)$. Also note that with this choice of $D$, we have (denoting by $d_{\max}$ and $d_{\min}$ the maximal and minimal element on the diagonal of $D$, respectively) $d_{\max} = \psi_{\min}^{-1}/\psi$ and $d_{\min} = \psi_{\max}^{-1}/\psi$, so that, in particular, $\frac{d_{\max}}{d_{\min}} = \frac{\psi_{\max}}{\psi_{\min}}$. 

(a): Fix $(t,\bfx) \in \bfD$ and $\psi > 0$. By Lemma~\ref{lem:LCQP}~(a), the minimiser of \eqref{eqn:lem:HJBI:H1LCQP:pf:recast} is $\bfz^* = \psi \wt\bfzeta(t,\bfx)$. After resubstitution, this yields the minimiser \eqref{eqn:lem:HJBI:H1LCQP:minimiser} of the original minimisation problem \eqref{eqn:lem:HJBI:H1LCQP:problem}. Moreover, by Lemma~\ref{lem:LCQP}~(c), the minimum of \eqref{eqn:lem:HJBI:H1LCQP:pf:recast} (which clearly coincides with the minimum of \eqref{eqn:lem:HJBI:H1LCQP:problem}) is
\begin{align*}
-\frac{1}{2}\bfv(t,\bfx)^\tr \bfz^*
&= -\frac{1}{2}\bfv(t,\bfx)^\tr \wt\bfzeta(t,\bfx) \psi
= -\frac{1}{2}\wt g(t,\bfx) \psi;
\end{align*}
recall the definition of $\wt g$ in \eqref{eqn:cash equivalent:source term}. For further reference, we also note that the bound on $\enorm{\bfz^*}$ from Lemma~\ref{lem:LCQP}~(a) translates to
\begin{align}
\label{eqn:lem:HJBI:H1LCQP:pf:zeta tilde:norm bound}
\enorm{\wt\bfzeta(t,\bfx)}
&\leq \psi_{\max} \enorm{\bfv(t,\bfx)}.
\end{align}

(b): This follows immediately from the second assertion of Lemma~\ref{lem:LCQP}~(d).

(c): By Assumption~\ref{ass:main result}~\ref{ass:main result:option} and the definition of $\bfv(t,\bfx)$, there is a constant $K_\bfv > 0$ such that $\enorm{\bfv(t,\bfx)} \leq K_\bfv$ for all $(t,\bfx)\in\bfD$. Set $K_{\wt g} = \psi_{\max} K_\bfv^2$ and fix $(t,\bfx) \in \bfD$ and $\psi > 0$. As $\bfzeta^0(\Sigma) \in \bfZ^0_\lin(t,\bfx)$ and $H^\psi_1(t,\bfx;\bfzeta^0(\Sigma)) = 0$, we have $\wt g(t,\bfx) \geq 0$ by \eqref{eqn:lem:HJBI:H1LCQP:minimum}. On the other hand, using the Cauchy--Schwarz inequality and \eqref{eqn:lem:HJBI:H1LCQP:pf:zeta tilde:norm bound}, we have for every $(t,\bfx) \in \bfD$:
\begin{align*}
g(t,\bfx)
&= \bfv(t,\bfx)^\tr \wt\bfzeta(t,\bfx)
\leq \enorm{\bfv(t,\bfx)} \enorm{\wt\bfzeta(t,\bfx)}
\leq \psi_{\max}\enorm{\bfv(t,\bfx)}^2 
\leq K_{\wt g}.
\end{align*}

(d): Set $K_\bfzeta = \max(\psi_{\max}K_\bfv,1)$ where $K_\bfv$ is as in the proof of part~(c), and fix $(t,\bfx)\in\bfD$ as well as $\psi > 0$. If $\Sigma \in \lbrace \ul\Sigma,\ol\Sigma\rbrace$, then $\bfzeta^\psi(t,\bfx) = \bfzeta^0(\Sigma)$ by construction and the assertion is trivial. Otherwise, if $\Sigma \in (\ul\Sigma,\ol\Sigma)$, then $\bfzeta^\psi(t,\bfx) = \bfzeta^{\psi*}(t,\bfx)$ and \eqref{eqn:lem:HJBI:H1LCQP:pf:zeta tilde:norm bound} implies that
\begin{align}
\label{eqn:lem:HJBI:H1LCQP:pf:norm estimate}
\enorm{\bfzeta^\psi(t,\bfx) - \bfzeta^0(\Sigma)}
&= \enorm{\bfzeta^{\psi*}(t,\bfx) - \bfzeta^0(\Sigma)}
= \enorm{\wt\bfzeta(t,\bfx)} \psi
\leq \psi_{\max}\enorm{\bfv(t,\bfx)}\psi 
\leq K_\bfzeta \psi.
\end{align}

(e): Fix $(t,\bfx)\in\bfD$ and $\psi > 0$. First, suppose that $\Sigma \in \lbrace \ul\Sigma,\ol\Sigma \rbrace$. Then $\bfzeta^\psi(t,\bfx) = \bfzeta^0(\Sigma)$ by construction of $\bfzeta^\psi$ and the assertion follows from the fact that $H^\psi_1(t,\bfx;\bfzeta^0(\Sigma)) = 0$ (cf.~\eqref{eqn:H1 H2 H4:zero}). Second, suppose that $\Sigma \in (\ul\Sigma,\ol\Sigma)$. Then $\bfzeta^\psi(t,\bfx) = \bfzeta^{\psi*}(t,\bfx)$ and the assertion follows from part~(a).

(f): Let $K_\bfv > 0$ be as in the proof of part~(c). Then the bound \eqref{eqn:lem:LCQP:Lagrange multiplier:norm bound} from Lemma~\ref{lem:LCQP}~(b) implies that, for every $(t,\bfx)\in\bfD$,
\begin{align}
\label{eqn:lem:HJBI:H1LCQP:pf:Lagrange multiplier:norm bound}
\enorm{\lambda(t,\bfx) \bfc(t,\bfx) - \mu(t,\bfx) \vec\bfe_4}
&\leq \left(1 + \frac{\psi_{\max}}{\psi_{\min}}\right) K_{\bfv}.
\end{align}
Recall from \eqref{eqn:VGVV:c} that $\bfc(t,\bfx) = (\cC_\Sigma, \Sigma S^2 \cC_{SS}, \Sigma S \cC_{S\Sigma}, \frac{1}{2} \cC_{\Sigma\Sigma})^\tr$. Clearly, each of the first three components of $\lambda(t,\bfx)\bfc(t,\bfx)$ is bounded in absolute value by the length of the vector $\lambda(t,\bfx) \bfc(t,\bfx) - \mu(t,\bfx) \vec\bfe_4$. Using also that $\ul\Sigma > 0$, we can find a constant $K > 0$ such that for every $(t,\bfx) = (t,S,A,M,\Sigma) \in \bfD$,
\begin{align}
\label{eqn:lem:HJBI:H1LCQP:pf:lambda c bounds}
\vert \lambda(t,\bfx) \cC_\Sigma(t,S,\Sigma) \vert
&\leq K,\quad
\vert \lambda(t,\bfx) S^2 \cC_{SS}(t,S,\Sigma) \vert
\leq K,\quad
\vert \lambda(t,\bfx) S \cC_{S\Sigma}(t,S,\Sigma) \vert
\leq K.
\end{align}
(This argument does not work for the fourth component of $\lambda(t,\bfx)\bfc(t,\bfx)$ due to the presence of the term $\mu(t,\bfx)\vec\bfe_4$ in \eqref{eqn:lem:HJBI:H1LCQP:pf:Lagrange multiplier:norm bound}.) Set $K_\lambda(t,\bfx) = 3K(2 + K_\cC(t,\bfx))$ where $K_\cC \in L^2_\fP$ is as in Assumption~\ref{ass:main result}~\ref{ass:main result:call}. Clearly, $K_\lambda \in L^2_\fP$.

Now, fix $(t,\bfx) \in \bfD$. Using that the Euclidean norm of a vector is dominated by the sum of the absolute values of each of its entries as well as Assumption~\ref{ass:main result}~\ref{ass:main result:call} to bound $\vert \cC_{\Sigma\Sigma} \vert$,
\begin{align*}
\vert \lambda(t,\bfx) \vert
\norm{
\begin{pmatrix}
S^2 \cC_{SS}    & S \cC_{S\Sigma}\\
S \cC_{S\Sigma} & \cC_{\Sigma\Sigma}
\end{pmatrix}
}_F
&= \vert \lambda(t,\bfx) \vert \enorm{(S^2 C_{SS}, S\cC_{S\Sigma}, S\cC_{S\Sigma}, \cC_{\Sigma\Sigma})}\\
&\leq \vert \lambda(t,\bfx) \vert \left( \vert S^2 \cC_{SS} \vert + 2 \vert S \cC_{S\Sigma} \vert +  \vert \cC_{\Sigma\Sigma} \vert \right)\\
&\leq (2 + K_\cC(t,\bfx)) \vert\lambda(t,\bfx)\vert \left( \vert \cC_\Sigma \vert + \vert S^2 \cC_{SS} \vert + \vert S \cC_{S\Sigma} \vert \right).
\end{align*}
Combining this with \eqref{eqn:lem:HJBI:H1LCQP:pf:lambda c bounds} and the choice of $K_\lambda$ completes the proof.
\end{proof}

\begin{corollary}
\label{cor:source term representation}
For each $(t,\bfx)\in\bfD$,
\begin{align*}
\wt g(t,\bfx)
&= -\Sigma \left(\phi^\star S^2 \cC_{SS} - (\beta \cV_A + S^2 \Gamma)\right)\wt\sigma
-\Sigma  \left(\phi^\star S\cC_{S\Sigma} - S\Vanna \right) \wt\eta
-\frac{1}{2}\left(\phi^\star \cC_{\Sigma\Sigma} - \cV_{\Sigma\Sigma}\right)\wt\xi,
\end{align*}
where the functions $(\wt\nu,\wt\sigma,\wt\eta,\wt\xi) = \wt\bfzeta$ are defined in \eqref{eqn:candidate control:first order} and $\phi^\star = \frac{\cV_\Sigma}{\cC_\Sigma}$.
\end{corollary}

\begin{proof}
We fix $(t,\bfx) \in \bfD$ and drop all arguments in the following to ease the notation. Recall that $\wt g = \bfv^\tr \wt\bfzeta$ by definition (cf.~\eqref{eqn:cash equivalent:source term}). Moreover, $\bfc^\tr \wt\bfzeta = 0$ by Lemma~\ref{lem:HJBI:H1LCQP}~(a). Hence,
\begin{align*}
\wt g
&= \left(\bfv - \phi^\star \bfc\right)^\tr\wt\bfzeta.
\end{align*}
Note that the first component of $\bfv - \phi^\star \bfc$ is zero by the choice of $\phi^\star$ (the vega hedge neutralises the portfolio vega). Now the assertion follows from the definitions of $\bfc$ and $\bfv$ in \eqref{eqn:VGVV:c} and \eqref{eqn:VGVV:v}.
\end{proof}

The remainder of this subsection provides estimates for the four terms $H^\psi_1, H^\psi_2, H_3$, and $H^\psi_4$. Roughly speaking, part~(a) of the first of the following propositions shows that $-\frac{1}{2}\wt g(t,\bfx) \psi$ is not only a lower bound for $H^\psi_1(t,\bfx;\bfzeta)$ over $\bfzeta\in\bfZ^0_\lin(t,\bfx)$ (as is shown by Lemma~\ref{lem:HJBI:H1LCQP}), but also, up to a term of order $O(\psi^2)$, for $\bfzeta \in \bfZ(t,\bfx)$ that are close to $\bfzeta^0(\Sigma)$. Moreover, part~(b) shows that this lower bound is approximately attained by controls $\bfzeta$ that are close to the candidate feedback control $\bfzeta^\psi$.

\begin{proposition}[$H^\psi_1$ estimate]~
\label{prop:HJBI:H1}
\begin{enumerate}
\item Let $0 \leq K \in L^4_\fP$. There is a nonnegative $K_1 \in L^1_\fP$ (depending on $K$) such that for every $(t,\bfx)\in\bfD$, $\bfzeta \in \bfZ(t,\bfx)$, and $\psi \in (0,1)$ satisfying
\begin{align}
\label{eqn:prop:HJBI:H1:lower bound:condition}
\enorm{\bfzeta - \bfzeta^0(\Sigma)}
&\leq K(t,\bfx)\psi,
\end{align}
we have
\begin{align*}
H^\psi_1(t,\bfx;\bfzeta)
&\geq - \frac{1}{2}\wt g(t,\bfx)\psi - K_1(t,\bfx) \psi^2.
\end{align*}

\item Let $0 \leq \bar K \in L^4_\fP$. There is a nonnegative $K_1 \in L^2_\fP$ (depending on $\bar K$) such that for every $(t,\bfx)\in\bfD$, $\bfzeta \in \bfZ$, and $\psi \in (0,1)$ satisfying
\begin{align}
\label{eqn:prop:HJBI:H1:upper bound:condition}
\enorm{\bfzeta - \bfzeta^\psi(t,\bfx)}
&\leq \bar K(t,\bfx)\psi^2,
\end{align}
we have
\begin{align*}
H^\psi_1(t,\bfx;\bfzeta)
&\leq -\frac{1}{2}\wt g(t,\bfx)\1_{\lbrace \Sigma \in (\ul\Sigma,\ol\Sigma)\rbrace}\psi + K_1(t,\bfx) \psi^2.
\end{align*}
\end{enumerate}
\end{proposition}

\begin{proof}
(a): Choose $0\leq K_\lambda \in L^2_\fP$ as in Lemma~\ref{lem:HJBI:H1LCQP}~(f) and set $K_1(t,\bfx) = \frac{1}{2} K_\lambda(t,\bfx)K(t,\bfx)^2$. As $K \in L^4_\fP$ and $K_\lambda \in L^2_\fP$, it follows that $K_1 \in L^1_\fP$. Now, fix $(t,\bfx) \in \bfD$, $\bfzeta \in \bfZ(t,\bfx)$, and $\psi \in (0,1)$ satisfying \eqref{eqn:prop:HJBI:H1:lower bound:condition}. As $\bfzeta = (\nu,\sigma,\eta,\xi) \in \bfZ(t,\bfx)$, we have $b^\cC(t,\bfx;\bfzeta) = 0$ and $\xi \geq 0$. Hence, using also that $\mu(t,\bfx) \geq 0$ by definition (cf.~\eqref{eqn:lagrange multiplier:mu}),
\begin{align*}
H^\psi_1(t,\bfx;\bfzeta)
&\geq H^\psi_1(t,\bfx;\bfzeta) + \lambda(t,\bfx) b^\cC(t,\bfx;\bfzeta) - \mu(t,\bfx) \xi.
\end{align*}
Substituting the expression \eqref{eqn:bC:matrix form} for $b^\cC$ and using the definition  \eqref{eqn:lem:HJBI:H1LCQP:Lagrangian} of the Lagrangian $L^\psi_1$, we obtain
\begin{align*}
H^\psi_1(t,\bfx;\bfzeta)
&\geq L^\psi_1(t,\bfx;\bfzeta,\lambda(t,\bfx),\mu(t,\bfx))
+\frac{\lambda(t,\bfx)}{2}
\begin{pmatrix}
\sigma - \Sigma\\
\eta
\end{pmatrix}^\tr
\begin{pmatrix}
S^2 \cC_{SS}    & S \cC_{S\Sigma}\\
S \cC_{S\Sigma} & \cC_{\Sigma\Sigma}
\end{pmatrix}
\begin{pmatrix}
\sigma - \Sigma\\
\eta
\end{pmatrix}.
\end{align*}
The first term on the right-hand side is bounded from below by $-\frac{1}{2}\wt g(t,\bfx) \psi$ by Lemma~\ref{lem:HJBI:H1LCQP}~(b). To estimate the second term, we use the Cauchy--Schwarz inequality, the compatibility of the Frobenius norm with the Euclidean norm, and the fact that $(\sigma - \Sigma, \eta)$ is just the second and third component of $\bfzeta - \bfzeta^0(\Sigma)$. As a result, we find that
\begin{align*}
H^\psi_1(t,\bfx;\bfzeta)
&\geq -\frac{1}{2} \wt g(t,\bfx) \psi 
-\frac{1}{2} \vert \lambda(t,\bfx) \vert
\norm{
\begin{pmatrix}
S^2 \cC_{SS}    & S \cC_{S\Sigma}\\
S \cC_{S\Sigma} & \cC_{\Sigma\Sigma}
\end{pmatrix}
}_F
\enorm{\bfzeta-\bfzeta^0(\Sigma)}^2.
\end{align*}
Finally, condition \eqref{eqn:prop:HJBI:H1:lower bound:condition} and the bound \eqref{eqn:lem:HJBI:H1LCQP:K lambda} from Lemma~\ref{lem:HJBI:H1LCQP}~(f) give
\begin{align*}
H^\psi_1(t,\bfx;\bfzeta)
&\geq -\frac{1}{2} \wt g(t,\bfx) \psi 
-\frac{1}{2}K_\lambda(t,\bfx) K(t,\bfx)^2 \psi^2.
\end{align*}
This proves assertion~(a) by the choice of $K_1$.

(b): Set $K'(t,\bfx) = K_\bfzeta + \bar K(t,\bfx)$ where $K_\bfzeta$ is chosen as in Lemma~\ref{lem:HJBI:H1LCQP}~(d). Clearly, ${K'\in L^4_\fP}$. By Assumption~\ref{ass:main result}~\ref{ass:main result:option} and the definition of $\bfv(t,\bfx)$, there is a constant $K_\bfv > 0$ such that ${\enorm{\bfv(t,\bfx)}\leq K_\bfv}$ for all $(t,\bfx)\in\bfD$. Next, set $K_1(t,\bfx) = \left(\psi_{\min}^{-1} K'(t,\bfx) + K_{\bfv}\right)\bar K(t,\bfx)$. Since ${K',\bar K\in L^4_\fP}$, it follows that $K_1\in L^2_\fP$.

Now, fix $(t,\bfx)\in\bfD$, $\bfzeta\in \bfZ$, and $\psi\in(0,1)$ satisfying \eqref{eqn:prop:HJBI:H1:upper bound:condition}. For brevity, we write $\bfzeta^\psi = \bfzeta^\psi(t,\bfx)$ and $\bfzeta^0 = \bfzeta^0(\Sigma)$. Now, by the multivariate mean-value theorem, there is $\ell \in [0,1]$ such that
\begin{align}
\label{eqn:prop:HJBI:H1:pf:mean value theorem 2}
H^\psi_1(t,\bfx;\bfzeta)
&= H^\psi_1(t,\bfx;\bfzeta^\psi) + \D_\bfzeta H^\psi_1(t,\bfx;\bfzeta_\ell)^\tr (\bfzeta - \bfzeta^\psi)
\end{align}
where $\bfzeta_\ell = (1-\ell)\bfzeta^\psi + \ell \bfzeta$. By the definition of $H^\psi_1$, we have
\begin{align*}
\D_\bfzeta H^\psi_1(t,\bfx;\bfzeta_\ell)
&= \frac{1}{\psi} \Psi^{-1} (\bfzeta_\ell - \bfzeta^0) - \bfv(t,\bfx).
\end{align*}
By Lemma~\ref{lem:HJBI:H1LCQP}~(d) and \eqref{eqn:prop:HJBI:H1:upper bound:condition},
\begin{align*}
\enorm{\bfzeta_\ell - \bfzeta^0}
&= \enorm{\bfzeta^\psi - \bfzeta^0 + \ell(\bfzeta-\bfzeta^\psi)}
\leq \enorm{\bfzeta^\psi - \bfzeta^0} + \ell \enorm{\bfzeta-\bfzeta^\psi}\\
&\leq K_\bfzeta\psi + \bar K(t,\bfx)\psi^2
\leq K'(t,\bfx)\psi,
\end{align*}
so that 
\begin{align}
\label{eqn:prop:HJBI:H1:pf:gradient estimate}
\enorm{\D_\bfzeta H^\psi_1(t,\bfx;\bfzeta_\ell)}
&\leq\psi_{\min}^{-1} K'(t,\bfx) + K_{\bfv}.
\end{align}
Moreover, by Lemma~\ref{lem:HJBI:H1LCQP}~(e), 
\begin{align*}
H^\psi_1(t,\bfx;\bfzeta^\psi)
&= -\frac{1}{2}\wt g(t,\bfx)\1_{\lbrace \Sigma \in (\ul\Sigma,\ol\Sigma)\rbrace}\psi.
\end{align*}
Combining this with \eqref{eqn:prop:HJBI:H1:pf:gradient estimate} and \eqref{eqn:prop:HJBI:H1:upper bound:condition} in \eqref{eqn:prop:HJBI:H1:pf:mean value theorem 2}, we obtain
\begin{align*}
H^\psi_1(t,\bfx;\bfzeta)
&\leq -\frac{1}{2}\wt g(t,\bfx)\1_{\lbrace \Sigma \in (\ul\Sigma,\ol\Sigma)\rbrace}\psi + \left(\psi_{\min}^{-1} K'(t,\bfx) + K_{\bfv}\right)\bar K(t,\bfx)\psi^2.
\end{align*}
This proves assertion~(b) by the choice of $K_1$.
\end{proof}

The next two propositions provide estimates for $H^\psi_2$ and $H_3$ in terms of the Euclidean distance between the reference feedback control $\bfzeta^0(\Sigma)$ and alternatives $\bfzeta$:

\begin{proposition}[$H^\psi_2$ estimate]
\label{prop:HJBI:H2}
There is $K_2 > 0$ such that for every $(t,\bfx,y) \in \bfD\times\RR$, $\bfzeta \in \bfZ$, and $\psi > 0$,
\begin{align*}
\left\vert H^\psi_2(t,\bfx,y;\bfzeta)\right\vert
&\leq K_2 \enorm{\bfzeta - \bfzeta^0(\Sigma)}  \left( \enorm{\bfzeta - \bfzeta^0(\Sigma)}
+ \frac{-U''(y)}{U'(y)} \max\left(1,\enorm{\bfzeta-\bfzeta^0(\Sigma)} \right)\psi \right).
\end{align*}
\end{proposition}

\begin{proof}
Set $K_2 = \max(K_\cV, 4\max(1,\ol\Sigma) K_\cV K_{\wt w})$, where $K_\cV$ and $K_{\wt w}$ are as in Assumption~\ref{ass:main result}~\ref{ass:main result:option}--\ref{ass:main result:cash equivalent}. Also fix $(t,\bfx,y)\in\bfD\times\RR$, $\bfzeta \in \bfZ$, and $\psi > 0$. Now, first note that $b^\cV(t,\bfx;\bfzeta)$ is of the form \eqref{eqn:lem:general drift function:q} with
\begin{align*}
\bfa
&= \left(\cV_\Sigma, \beta \cV_A + S^2 \Gamma, S \frac{\partial \Delta}{\partial \Sigma}, \cV_{\Sigma\Sigma}\right)^\tr.
\end{align*}
Hence, by Lemma~\ref{lem:general drift function}, the fact that $\enorm{\bfa} \leq 2 K_\cV$ by Assumption~\ref{ass:main result}~\ref{ass:main result:option}, and the choice of $K_2$, 
\begin{align}
\label{eqn:prop:HJBI:H2:pf:estimate 1}
\begin{split}
\vert b^\cV(t,\bfx;\bfzeta) \vert
&\leq \max(1,\Sigma) \enorm{\bfa} \enorm{\bfzeta - \bfzeta^0(\Sigma)} + \enorm{\bfa} \enorm{\bfzeta - \bfzeta^0(\Sigma)}^2\\
&\leq \frac{K_2}{K_{\wt w}} \enorm{\bfzeta - \bfzeta^0(\Sigma)} \max \left(1,\enorm{\bfzeta - \bfzeta^0(\Sigma)} \right).
\end{split}
\end{align}
Next, using first norm estimates as in the proof of Proposition~\ref{prop:HJBI:H2}~(a) and then Assumption~\ref{ass:main result}~\ref{ass:main result:option} to estimate the resulting Frobenius norm by $2 K_\cV$, we find
\begin{align}
\label{eqn:prop:HJBI:H2:pf:estimate 2}
\frac{1}{2}\left\vert
\begin{pmatrix}
\sigma - \Sigma\\
\eta
\end{pmatrix}^\tr
\begin{pmatrix}
\beta \cV_A + S^2 \Gamma                & S \frac{\partial\Delta}{\partial\Sigma}\\
S \frac{\partial\Delta}{\partial\Sigma} & \cV_{\Sigma\Sigma}
\end{pmatrix}
\begin{pmatrix}
\sigma - \Sigma\\
\eta
\end{pmatrix}
\right\vert
&\leq
K_\cV \enorm{\bfzeta-\bfzeta^0(\Sigma)}^2.
\end{align}
Finally, using \eqref{eqn:prop:HJBI:H2:pf:estimate 1}--\eqref{eqn:prop:HJBI:H2:pf:estimate 2} and the fact that $\vert\wt w\vert\leq K_{\wt w}$ on $\bfD$ by Assumption~\ref{ass:main result}~\ref{ass:main result:cash equivalent}, we obtain
\begin{align*}
\left\vert H^\psi_2(t,\bfx,y;\bfzeta)\right\vert
&\leq K_\cV \enorm{\bfzeta-\bfzeta^0(\Sigma)}^2 + K_2 \frac{-U''(y)}{U'(y)} \enorm{\bfzeta - \bfzeta^0(\Sigma)} \max \left(1,\enorm{\bfzeta - \bfzeta^0(\Sigma)} \right)\psi\\
&\leq K_2 \enorm{\bfzeta - \bfzeta^0(\Sigma)}  \left( \enorm{\bfzeta - \bfzeta^0(\Sigma)}
+ \frac{-U''(y)}{U'(y)} \max\left(1,\enorm{\bfzeta-\bfzeta^0(\Sigma)} \right)\psi \right).\qedhere
\end{align*}
\end{proof}

\begin{proposition}[$H_3$ estimate]
\label{prop:HJBI:H3}
There is $K_3 \in L^4_\fP$ such that for every $(t,\bfx)\in\bfD$ and $\bfzeta \in \bfZ$,
\begin{align*}
\left\vert H_3(t,\bfx;\bfzeta) - H_3(t,\bfx;\bfzeta^0(\Sigma))\right\vert
&\leq K_3(t,\bfx) \enorm{\bfzeta - \bfzeta^0(\Sigma)} \max\left(1,\enorm{\bfzeta-\bfzeta^0(\Sigma)} \right).
\end{align*}
\end{proposition}

\begin{proof}
Set $K_3(t,\bfx) =2 \max(1,\ol\Sigma) \enorm{\bfa(t,\bfx)}$ where 
\begin{align*}
\bfa(t,\bfx)
&= \left(\wt w_\Sigma,\beta \wt w_A + S^2 (\wt w_{SS} + 2\gamma \wt w_{SA} + \gamma^2 \wt w_{AA}), S(\wt w_{S\Sigma} + \gamma \wt w_{A\Sigma}),\wt w_{\Sigma\Sigma}\right)^\tr.
\end{align*}
By Assumption~\ref{ass:main result}~\ref{ass:main result:cash equivalent}, every component of $\bfa$ is in $L^4_\fP$ and thus also $K_3 \in L^4_\fP$. Now, fix $(t,\bfx)\in\bfD$ and $\bfzeta \in \bfZ$. It is easy to see that the difference 
\begin{align*}
d&:= H_3(t,\bfx;\bfzeta) - H_3(t,\bfx;\bfzeta^0(\Sigma))
\end{align*}
is of the form \eqref{eqn:lem:general drift function:q} for $\bfa = \bfa(t,\bfx)$. Hence, by Lemma~\ref{lem:general drift function} and the choice of $K_3$,
\begin{align*}
\vert d \vert
&\leq \max(1,\Sigma) \enorm{\bfa} \enorm{\bfzeta-\bfzeta^0(\Sigma)} + \enorm{\bfa} \enorm{\bfzeta-\bfzeta^0(\Sigma)}^2\\
&\leq K_3(t,\bfx) \enorm{\bfzeta - \bfzeta^0(\Sigma)} \max\left(1,\enorm{\bfzeta-\bfzeta^0(\Sigma)} \right).\qedhere
\end{align*}
\end{proof}

By combining Propositions \ref{prop:HJBI:H2} and \ref{prop:HJBI:H3}, the following corollary guarantees that if $\bfzeta$ is close to $\bfzeta^0(\Sigma)$, then $H^\psi_2(t,\bfx;\bfzeta)$ is of order $O(\psi^2)$ and $H_3(t,\bfx;\bfzeta)$ can be replaced by $H_3(t,\bfx;\bfzeta^0(\Sigma))$ and a term of order $O(\psi)$.

\begin{corollary}
\label{cor:HJBI:H2 and H3}
Let $0 \leq K \in L^4_\fP$. There is $K_{2,3} \in L^2_\fP$ (depending on $K$) such that for every $(t,\bfx,y)\in\bfD\times\RR$, $\bfzeta \in \bfZ$, and $\psi > 0$ satisfying
\begin{align}
\label{eqn:cor:HJBI:H2 and H3:condition}
\enorm{\bfzeta - \bfzeta^0(\Sigma)}
&\leq K(t,\bfx) \psi,
\end{align}
we have
\begin{align*}
\left\vert H^\psi_2 (t,\bfx,y;\bfzeta)\right\vert
&\leq K_{2,3}(t,\bfx)\left(1+\frac{-U''(y)}{U'(y)}\right) \psi^2,\\
\left\vert H_3(t,\bfx;\bfzeta) - H_3(t,\bfx;\bfzeta^0(\Sigma)) \right\vert
&\leq K_{2,3}(t,\bfx) \psi.
\end{align*}
\end{corollary}

\begin{proof}
Choose $K_2 > 0$ and $K_3 \in L^4_\fP$ as in Propositions~\ref{prop:HJBI:H2} and \ref{prop:HJBI:H3}. Since $\bfZ$ and $[\ul\Sigma,\ol\Sigma]$ are bounded, there is a constant $K' \geq 1$ such that for every $\bfzeta \in \bfZ$ and $\Sigma \in [\ul\Sigma,\ol\Sigma]$,
\begin{align*}
\enorm{\bfzeta - \bfzeta^0(\Sigma)} \leq K'.
\end{align*}

Set $K_{2,3} = K(t,\bfx) \left(K_2 \max(K',K(t,\bfx)) + K_3(t,\bfx)K'\right)$. It is easy to see that $K_{2,3} \in L^2_\fP$. Now, fix $(t,\bfx,y) \in \bfD\times\RR$, $\bfzeta \in \bfZ$, and $\psi > 0$ satisfying \eqref{eqn:cor:HJBI:H2 and H3:condition}. Then by Proposition~\ref{prop:HJBI:H2} and \eqref{eqn:cor:HJBI:H2 and H3:condition},
\begin{align*}
\left\vert H^\psi_2(t,\bfx,y;\bfzeta)\right\vert
&\leq K_2 K(t,\bfx)\psi  \left( K(t,\bfx)\psi + \frac{-U''(y)}{U'(y)} K'\psi \right)
\leq K_{2,3}(t,\bfx) \left(1 + \frac{-U''(y)}{U'(y)}\right) \psi^2.
\end{align*}
Similarly, by Proposition~\ref{prop:HJBI:H3} and \eqref{eqn:cor:HJBI:H2 and H3:condition},
\begin{align*}
\left\vert H_3(t,\bfx;\bfzeta) - H_3(t,\bfx;\bfzeta^0(\Sigma))\right\vert
&\leq K_3(t,\bfx) K(t,\bfx) K' \psi
\leq K_{2,3}(t,\bfx) \psi.\qedhere
\end{align*}
\end{proof}

Finally, Proposition~\ref{prop:HJBI:H4} below shows that $H^\psi_4$ is bounded from below by $0$ up to a term of order $O(\psi^2)$. Recall from \eqref{eqn:H1 H2 H4:zero} that this asymptotic lower bound is attained by the delta-vega hedge, i.e., $H^\psi_4(t,\bfx,y;\bfupsilon^\star(t,\bfx),\bfzeta) = 0$.

\begin{proposition}[$H^\psi_4$ estimate]
\label{prop:HJBI:H4}
There is a nonnegative $K_4 \in L^2_\fP$ such that for every $(t,\bfx,y) \in \bfD\times\RR$, $\bfupsilon\in\RR^2$, $\bfzeta \in \bfZ$, and $\psi>0$,
\begin{align*}
H^\psi_4(t,\bfx,y;\bfupsilon,\bfzeta)
\geq U''(y) K_4(t,\bfx) \psi^2.
\end{align*}
\end{proposition}

\begin{proof}
We first argue that $w^\psi_{YY} \leq U''$ on $\bfD \times \RR$. As $U$ has decreasing absolute risk aversion (cf.~Assumption~\ref{ass:main result}~\ref{ass:main result:utility}), we have for each $y \in \RR$,
\begin{align*}
0 
&\geq\frac{\diff}{\diff y}\left(-\frac{U''(y)}{U'(y)}\right)
= -\frac{U'(y)U'''(y) - U''(y)^2}{U'(y)^2}.
\end{align*}
In particular, since $U' > 0$, we have $U''' > 0$. Together with $\wt w \geq 0$ (cf.~Assumption~\ref{ass:main result}~\ref{ass:main result:cash equivalent}), this yields $w^\psi_{YY} = U'' - U''' \wt w \psi \leq U'' < 0$ on $\bfD\times\RR$.

Now, set $K_4(t,\bfx) = \frac{1}{2} K \left( \ol\sigma^2 S^2 (\wt w_{S} + \gamma \wt w_A)^2 + \wt w_\Sigma^2 \right)$ for some constant
\begin{align*}
K \geq \max_{\bfzeta \in \bfZ} \;(1 + 2\eta^2 + (\eta^2 + \xi)^2)^{1/2}.
\end{align*}
As $S(\wt w_S + \gamma \wt w_A), \wt w_\Sigma \in L^4_\fP$ by Assumption~\ref{ass:main result}~\ref{ass:main result:cash equivalent}, we have $K_4 \in L^2_\fP$. Next, fix $(t,\bfx,y) \in \bfD \times \RR$, $\bfzeta \in \bfZ$, and $\psi > 0$. Write
\begin{align*}
Q
&=
\begin{pmatrix}
1 & \eta\\
\eta & \eta^2 + \xi
\end{pmatrix}
\quad\text{and}\quad
\wt \bfw
=
\begin{pmatrix}
\sigma S (\wt w_S + \gamma \wt w_A)\\
\wt w_\Sigma
\end{pmatrix},
\end{align*}
and consider the function $q:\RR^2 \to \RR$ given by
\begin{align*}
q(\bfz) 
&= -\frac{w^\psi_{YY}}{2} \bfz^\tr Q \bfz + \psi U''(y)
\wt\bfw^\tr
Q \bfz.
\end{align*}
Clearly, minimising $q$ over $\bfz \in \RR^2$ is equivalent to minimising $H^\psi_4(t,\bfx,y;\bfupsilon,\bfzeta)$ over $\bfupsilon \in \RR^2$ (recall that $\cC_\Sigma \neq 0$ by Assumption~\ref{ass:main result}~\ref{ass:main result:call}). Moreover, 
\begin{align*}
\det Q
&= (\eta^2 + \xi) - \eta^2
= \xi \geq 0
\quad\text{and}\quad
\Trace Q
= 1 + \eta^2 + \xi > 0,
\end{align*}
so that the symmetric matrix $Q$ is positive semi-definite. It follows that $q$ is convex and any solution to the first-order condition
\begin{align*}
Q \left(- w^\psi_{YY} \bfz + \psi U''(y)
\wt\bfw\right)
&= 0
\end{align*}
is a (global) minimiser. As $\bfz^* = \frac{U''(y)}{w^\psi_{YY}}\wt\bfw\psi$ solves the first-order condition, we obtain after some algebra that the minimum of $q$ is
\begin{align*}
\frac{1}{2} \frac{U''(y)^2}{w^\psi_{YY}}
\wt\bfw^\tr
Q
\wt\bfw
\psi^2.
\end{align*}
Using also that $w^\psi_{YY} \leq U'' < 0$ on $\bfD \times \RR$, we conclude that for all $\bfupsilon\in\RR^2$,
\begin{align*}
H^\psi_4(t,\bfx,y;\bfupsilon,\bfzeta)
\geq \frac{1}{2}U''(y) \wt\bfw^\tr Q \wt\bfw \psi^2.
\end{align*}
Finally,
\begin{align*}
\frac{1}{2}\left\vert  \wt\bfw^\tr Q \wt\bfw \right\vert
&\leq \frac{1}{2}\enorm{\wt\bfw}^2 \norm{Q}_F
\leq K_4(t,\bfx)
\end{align*}
by the choice of $K_4$. Combining the preceding two estimates completes the proof.
\end{proof}

\subsubsection{Approximate solution to the HJBI equation}
\label{sec:proofs:candidate value function}

The following lemma shows that the candidate value function $w^\psi$ defined in \eqref{eqn:value function} is, up to a term of order $O(\psi^2)$, a ``supersolution'' to the HJBI equation \eqref{eqn:HJBI}. This analytic result is the main ingredient for the proof of the inequality~\eqref{eqn:SDG inequalities1} in  Section~\ref{sec:proofs:lower bound}.

\begin{lemma}[Lower bound]
\label{lem:HJBI:lower bound}
Fix constants $\ul Y \leq \ol Y$. There is a nonnegative $K_\mathrm{lo}\in L^1_\fP$ (depending on $\ul Y,\ol Y$) such that for every $(t,\bfx,y) \in \bfD\times[\ul Y,\ol Y]$ and $\psi \in (0,1)$,
\begin{align}
\label{eqn:lem:HJBI:lower bound}
w^\psi_t(t,\bfx,y) + \inf_{\bfzeta \in \bfZ(t,\bfx)} H^\psi(t,\bfx,y;\bfupsilon^\star(t,\bfx),\bfzeta)
&\geq -K_\mathrm{lo}(t,\bfx)\psi^2.
\end{align}

\end{lemma}

\begin{proof}
As an auxiliary result, we first prove that there is $K \in L^4_\fP$ such that for every $(t,\bfx,y) \in \bfD\times[\ul Y,\ol Y]$, $\bfzeta \in \bfZ$, and $\psi \in (0,1)$ satisfying
\begin{align}
\label{eqn:lem:HJBI:lower bound:pf:auxiliary estimate:condition}
H^\psi(t,\bfx,y;\bfupsilon^\star(t,\bfx),\bfzeta)
&\leq H^\psi(t,\bfx,y;\bfupsilon^\star(t,\bfx),\bfzeta^0(\Sigma)),
\end{align}
we have
\begin{align}
\label{eqn:lem:HJBI:lower bound:pf:auxiliary estimate}
\enorm{\bfzeta - \bfzeta^0(\Sigma)}
&\leq K(t,\bfx) \psi.
\end{align}
Using Proposition~\ref{prop:HJBI:H2} and the fact that $\bfZ$ and $[\ul Y,\ol Y]$ are compact, there is a constant $K_2' > 0$ such that for every $(t,\bfx,y)\in\bfD\times[\ul Y,\ol Y]$, $\bfzeta \in \bfZ$, and $\psi \in (0,1)$,
\begin{align*}
\left\vert H^\psi_2(t,\bfx,y;\bfzeta)\right\vert
&\leq K_2' \enorm{\bfzeta-\bfzeta^0(\Sigma)}.
\end{align*}
Similarly, using Proposition~\ref{prop:HJBI:H3}, there is $K_3' \in L^4_\fP$ such that for every $(t,\bfx)\in\bfD$ and $\bfzeta \in \bfZ$,
\begin{align*}
\left\vert H_3(t,\bfx;\bfzeta) - H_3(t,\bfx;\bfzeta^0(\Sigma))\right\vert
&\leq K_3'(t,\bfx) \enorm{\bfzeta-\bfzeta^0(\Sigma)}.
\end{align*}
Now, set $K(t,\bfx) = 2\psi_{\max}\left(\enorm{\bfv(t,\bfx)} + K_2' + K_3'(t,\bfx) \right)$. Using that $\bfv(t,\bfx)$ is uniformly bounded by Assumption~\ref{ass:main result}~\ref{ass:main result:option}, it follows that $K\in L^4_\fP$.
Fix $(t,\bfx,y) \in \bfD \times [\ul Y,\ol Y]$, $\bfzeta \in \bfZ$, and $\psi \in (0,1)$ satisfying \eqref{eqn:lem:HJBI:lower bound:pf:auxiliary estimate:condition}. Rearranging \eqref{eqn:lem:HJBI:lower bound:pf:auxiliary estimate:condition} and using the decomposition \eqref{eqn:Hamiltonian:decomposition} of $H^\psi$, the fact that the $H^\psi_4$ term vanishes by \eqref{eqn:H1 H2 H4:zero}, the above estimates for $H^\psi_2$ and $H_3$ as well as a direct estimate for the $H^\psi_1$ term, we find
\begin{align*}
0
&\geq \left(H^\psi(t,\bfx,y;\bfupsilon^\star(t,\bfx),\bfzeta) - H^\psi(t,\bfx,y;\bfupsilon^\star(t,\bfx),\bfzeta^0(\Sigma)) \right)/U'(y)\\
&= H^\psi_1(t,\bfx;\bfzeta) + H^\psi_2(t,\bfx,y;\bfzeta) - \left(H_3(t,\bfx;\bfzeta) - H_3(t,\bfx;\bfzeta^0(\Sigma))\right) \psi\\
&\geq \frac{1}{2\psi \psi_{\max}} \enorm{\bfzeta-\bfzeta^0(\Sigma)}^2 - \bfv(t,\bfx)^\tr(\bfzeta-\bfzeta^0(\Sigma)) - (K_2' + K_3'(t,\bfx))\enorm{\bfzeta-\bfzeta^0(\Sigma)}.
\end{align*}
By rearranging terms and applying the Cauchy--Schwarz inequality, we obtain
\begin{align*}
\enorm{\bfzeta-\bfzeta^0(\Sigma)}^2
&\leq 2\psi_{\max}\big(\enorm{\bfv(t,\bfx)} + K_2' + K_3'(t,\bfx) \big) \enorm{\bfzeta-\bfzeta^0(\Sigma)} \psi\\
&\leq K(t,\bfx)\enorm{\bfzeta-\bfzeta^0(\Sigma)} \psi
\end{align*}
and \eqref{eqn:lem:HJBI:lower bound:pf:auxiliary estimate} follows.

We now turn to the proof of \eqref{eqn:lem:HJBI:lower bound}. Choose $K_{2,3} \in L^2_\fP$ as in Corollary~\ref{cor:HJBI:H2 and H3} (with $K$ as in the auxiliary result), $K_1 \in L^1_\fP$ as in Proposition~\ref{prop:HJBI:H1}~(a), and set
\begin{align*}
K_\mathrm{lo}(t,\bfx) = U'(\ul Y)\left(K_1(t,\bfx) + K_{2,3}(t,\bfx)\left(2+\frac{-U''(\ul Y)}{U'(\ul Y)}\right)\right).
\end{align*}
Clearly, $K_\mathrm{lo} \in L^1_\fP$. Fix $(t,\bfx,y)\in\bfD\times[\ul Y,\ol Y]$, $\bfzeta \in \bfZ(t,\bfx)$, and $\psi \in (0,1)$. First, we note that $\bfzeta^0(\Sigma) \in \bfZ(t,\bfx)$ and that by \eqref{eqn:cash equivalent:PDE:H3}--\eqref{eqn:H1 H2 H4:zero} and Lemma~\ref{lem:HJBI:H1LCQP}~(c),
\begin{align*}
w^\psi_t(t,\bfx,y) + H^\psi(t,\bfx,y;\bfupsilon^\star(t,\bfx),\bfzeta^0(\Sigma))
&= -U'(y)\left( \wt w_t(t,\bfx) + H_3(t,\bfx;\bfzeta^0(\Sigma)) \right)\psi\\
&= \frac{1}{2}U'(y) \wt g(t,\bfx) \psi
\geq 0.
\end{align*}
In view of assertion \eqref{eqn:lem:HJBI:lower bound}, we may thus assume that \eqref{eqn:lem:HJBI:lower bound:pf:auxiliary estimate:condition} is satisfied. In turn, \eqref{eqn:lem:HJBI:lower bound:pf:auxiliary estimate} holds by the auxiliary result. In particular, we may use the estimates of Proposition~\ref{prop:HJBI:H1}~(a) (for $H^\psi_1$) and Corollary~\ref{cor:HJBI:H2 and H3} (for $H^\psi_2$ and $H_3$) in the following. These together with the fact that the $H^\psi_4$ term vanishes by \eqref{eqn:H1 H2 H4:zero} yield
\begin{align*}
&w^\psi_t(t,\bfx,y) + H^\psi(t,\bfx,y;\bfupsilon^\star(t,\bfx),\bfzeta)\\
&\;= U'(y) \left(-\wt w_t(t,\bfx)\psi + H^\psi_1(t,\bfx;\bfzeta) + H^\psi_2(t,\bfx,y;\bfzeta) - H_3(t,\bfx;\bfzeta)\psi\right)\\
&\;\geq -U'(y)\left(\wt w_t(t,\bfx) +  H_3(t,\bfx;\bfzeta^0(\Sigma)) + \frac{1}{2}\wt g(t,\bfx)\right)\psi\\
&\qquad- U'(y)\left(K_1(t,\bfx) + K_{2,3}(t,\bfx)\left(1+\frac{-U''(y)}{U'(y)}\right) + K_{2,3}(t,\bfx)\right)\psi^2\\
&\;\geq -K_\mathrm{lo}(t,\bfx) \psi^2,
\end{align*}
where in the last inequality, we also use \eqref{eqn:cash equivalent:PDE:H3} to eliminate the $O(\psi)$ term and the fact that $U$ has decreasing absolute risk aversion (cf.~Assumption~\ref{ass:main result}~\ref{ass:main result:utility}) to estimate the $O(\psi^2)$ term. As $\bfzeta \in \bfZ(t,\bfx)$ was arbitrary, \eqref{eqn:lem:HJBI:lower bound} follows.
\end{proof}

Conversely, the next lemma shows that the candidate value function $w^\psi$ defined in \eqref{eqn:value function} is asymptotically a ``subsolution'' to the HJBI equation \eqref{eqn:HJBI}. Here, the asymptotic estimate is of order $O(\psi^2)$ if $\Sigma$ is in the interior of $[\ul\Sigma,\ol\Sigma]$ and of order $O(\psi)$ otherwise. This analytic result is the main ingredient for the proof of the inequality~\eqref{eqn:SDG inequalities2} in Section~\ref{sec:proofs:upper bound}.

\begin{lemma}[Upper bound]
\label{lem:HJBI:upper bound}
Let $0 \leq \bar K \in L^4_\fP$. There is a nonnegative $K_\mathrm{up}\in L^2_\fP$ (depending on $\bar K$) such that for every $(t,\bfx,y) \in \bfD\times\RR$, $\bfzeta \in \bfZ$, and $\psi\in(0,1)$ satisfying
\begin{align}
\label{eqn:lem:HJBI:upper bound:condition}
\enorm{\bfzeta - \bfzeta^\psi(t,\bfx)}
&\leq \bar K(t,\bfx) \psi^2,
\end{align}
we have
\begin{align}
\label{eqn:lem:HJBI:upper bound}
w^\psi_t(t,\bfx,y) + \sup_{\bfupsilon \in \RR^2} H^\psi(t,\bfx,y;\bfupsilon,\bfzeta)
&\leq K_\mathrm{up}(t,\bfx) U'(y)\left(1+\frac{-U''(y)}{U'(y)}\right)\psi^{1+\1_{\lbrace \Sigma \in (\ul\Sigma,\ol\Sigma)\rbrace}}.
\end{align}
\end{lemma}

\begin{proof}
Define $K_\bfzeta \geq 1$ as in Lemma~\ref{lem:HJBI:H1LCQP}~(d), and set $K(t,\bfx) = \bar K(t,\bfx) + K_\bfzeta$. Clearly, $K \in L^4_\fP$. With this choice of $K$, let $K_{2,3} \in L^2_\fP$ be defined as in Corollary~\ref{cor:HJBI:H2 and H3}. Moreover, define $K_1 \in L^2_\fP$ as in Proposition~\ref{prop:HJBI:H1}~(b) and $K_4 \in L^2_\fP$ as in Proposition~\ref{prop:HJBI:H4}. In addition, note that there is $K_{\wt g} > 0$ such that $0\leq \wt g \leq K_{\wt g}$ on $\bfD$ by Lemma~\ref{lem:HJBI:H1LCQP}~(c).

Now, set $K_\mathrm{up}(t,\bfx) = 4\max\left(K_1(t,\bfx) + 2 K_{2,3}(t,\bfx) + K_4(t,\bfx),\frac{1}{2}K_{\wt g}\right)$. Clearly, $K_\mathrm{up} \in L^2_\fP$. Fix $(t,\bfx,y)\in\bfD\times\RR$, $\bfupsilon \in \RR^2$, $\bfzeta \in \bfZ$, and $\psi \in (0,1)$ satisfying \eqref{eqn:lem:HJBI:upper bound:condition}. In particular, condition \eqref{eqn:prop:HJBI:H1:upper bound:condition} of Proposition~\ref{prop:HJBI:H1}~(b) holds. By Lemma~\ref{lem:HJBI:H1LCQP}~(d) and \eqref{eqn:lem:HJBI:upper bound:condition},
\begin{align*}
\enorm{\bfzeta - \bfzeta^0(\Sigma)}
&\leq \enorm{\bfzeta - \bfzeta^\psi(t,\bfx)} + \enorm{\bfzeta^\psi(t,\bfx) - \bfzeta^0(\Sigma)}
\leq \bar K(t,\bfx)\psi^2 + K_\bfzeta \psi
\leq K(t,\bfx)\psi,
\end{align*}
so that condition~\eqref{eqn:cor:HJBI:H2 and H3:condition} of Corollary~\ref{cor:HJBI:H2 and H3} is satisfied as well.

Using Propositions~\ref{prop:HJBI:H1}~(b) (for $H^\psi_1$) and \ref{prop:HJBI:H4} (for $H^\psi_4$) as well as Corollary~\ref{cor:HJBI:H2 and H3} (for $H^\psi_2$ and $H_3$) to estimate the four summands in the decomposition~\eqref{eqn:Hamiltonian:decomposition} of $H^\psi$, and also \eqref{eqn:cash equivalent:PDE:H3} in the penultimate step, we obtain
\begin{align*}
&w^\psi_t(t,\bfx,y) + H^\psi(t,\bfx,y;\bfupsilon,\bfzeta)\\
&\leq -U'(y)\wt w_t(t,\bfx)\psi
+ U'(y)\left(-\frac{1}{2}\wt g(t,\bfx) \1_{\lbrace \Sigma\in(\ul\Sigma,\ol\Sigma)\rbrace} \psi + K_1(t,\bfx)\psi^2 \right)\\
&\qquad+ U'(y)K_{2,3}(t,\bfx) \left(1 + \frac{-U''(y)}{U'(y)}\right)\psi^2\\
&\qquad- U'(y)\left(H_3(t,\bfx;\bfzeta^0(\Sigma)) - K_{2,3}(t,\bfx)\psi\right)\psi
- U''(y) K_4(t,\bfx)\psi^2\\
&= -U'(y)\left(\wt w_t(t,\bfx) + H_3(t,\bfx;\bfzeta^0(\Sigma)) + \frac{1}{2}\wt g(t,\bfx) \1_{\lbrace \Sigma\in(\ul\Sigma,\ol\Sigma)\rbrace} \right)\psi\\
&\qquad +U'(y)\left(K_1(t,\bfx) + 2 K_{2,3}(t,\bfx) + \big(K_{2,3}(t,\bfx) + K_4(t,\bfx)\big)\frac{-U''(y)}{U'(y)}\right)\psi^2\\
&\leq U'(y) \frac{1}{2}\wt g(t,\bfx) \1_{\lbrace \Sigma \in \lbrace \ul\Sigma,\ol\Sigma \rbrace\rbrace} \psi
+ \frac{1}{4}K_\mathrm{up}(t,\bfx) U'(y) \left(1 + \frac{-U''(y)}{U'(y)}\right)\psi^2\\
&\leq \frac{1}{2} K_\mathrm{up}(t,\bfx)U'(y) \left(1+ \frac{-U''(y)}{U'(y)}\right)\left(\1_{\lbrace \Sigma \in \lbrace \ul\Sigma,\ol\Sigma \rbrace\rbrace} \psi + \psi^2 \right).
\end{align*}
As $\bfupsilon \in \RR^2$ was arbitrary, the assertion follows easily by distinguishing the cases $\Sigma \in (\ul\Sigma,\ol\Sigma)$ and $\Sigma \in \lbrace\ul\Sigma,\ol\Sigma\rbrace$ (using that $\psi \in (0,1)$ in the second case).
\end{proof}

\subsubsection{The asymptotic lower bound for the stochastic differential game}
\label{sec:proofs:lower bound}

We are now in a position to establish an asymptotic lower bound for the SDG \eqref{eqn:value}, as required for the proof of Theorem~\ref{thm:main result} at the beginning of Section~\ref{sec:proofs:main result}.

\begin{lemma}
\label{lem:lower bound}
As $\psi \downarrow 0$,
\begin{align*}
\inf_{P \in \fP} J^\psi(\bfupsilon^\star,P)
&\geq w^\psi_0 + o(\psi).
\end{align*}
\end{lemma}

\begin{proof}
Choose $\ul Y, \ol Y$ as in Corollary~\ref{cor:Y:delta-vega} and, with this choice, let $K_\mathrm{lo} \in L^1_\fP$ be as in Lemma~\ref{lem:HJBI:lower bound}. Now, fix $\varepsilon > 0$, $\psi_0' \in (0,\psi_0)$ such that $\norm{K_\mathrm{lo}}_{L^1_\fP} \psi_0' \leq \frac{1}{2}\varepsilon$, and let $\psi \in (0,\psi_0')$. We need to show that 
\begin{align}
\label{eqn:lem:lower bound:pf:to show}
\inf_{P'\in\fP} J^\psi(\bfupsilon^\star,P') - w^\psi_0
&\geq -\varepsilon \psi.
\end{align}
Choose $P \in \fP$ such that $J^\psi(\bfupsilon^\star,P) - \frac{1}{2}\varepsilon\psi \leq \inf_{P'\in\fP} J^\psi(\bfupsilon^\star,P')$. Then
\begin{align}
\label{eqn:lem:lower bound:pf:10}
\inf_{P'\in\fP} J^\psi(\bfupsilon^\star,P') - w^\psi_0
&\geq J^\psi(\bfupsilon^\star,P) - w^\psi_0 - \frac{1}{2}\varepsilon\psi.
\end{align}

Applying It\^o's formula (under $P$) to the process $w^\psi(u,\bfX_u,Y^{\bfupsilon^\star,P}_u)$ (recall the dynamics of $S,A,M,\Sigma$, and $Y^{\bfupsilon^\star,P}$ given in \eqref{eqn:covariation:S and IV}--\eqref{eqn:A:dynamics} and Corollary~\ref{cor:Y:delta-vega}) and using the third line in \eqref{eqn:cash equivalent:PDE} (so that the $\diff M$-integral vanishes) yields for each $u\in[0,T]$,
\begin{align}
\label{eqn:lem:lower bound:pf:Ito}
\begin{split}
I^\psi_u(\bfupsilon^\star,P)
&:=w^\psi(u,\bfX_u,Y^{\bfupsilon^\star,P}_u) + \frac{1}{\psi}\int_0^u U'(Y^{\bfupsilon^\star,P}_t) f(\Sigma_t,\bfzeta^P_t) \dd t - w^\psi_0 \\
&= N_u
+ \int_0^u \left( w^\psi_t(t,\bfX_t,Y^{\bfupsilon^\star,P}_t) + H^\psi(t,\bfX_t,Y^{\bfupsilon^\star,P}_t;\bfupsilon^\star(t,\bfX_t),\bfzeta^P_t)\right) \dd t,
\end{split}
\end{align}
where 
\begin{align*}
\begin{split}
N
&:= \int_0^\cdot \left(w^\psi_S(t,\bfX_t,Y^{\bfupsilon^\star,P}_t) + \gamma(t,S_t,A_t,M_t) w^\psi_A(t,\bfX_t,Y^{\bfupsilon^\star,P}_t) \right) \dd S_t\\
&\qquad+ \int_0^\cdot w^\psi_\Sigma(t,\bfX_t,Y^{\bfupsilon^\star,P}_t) \dd \Sigma^{\contlocmartpart,P}_t.
\end{split}
\end{align*}
Note that $\bfzeta^P_t \in \bfZ(t,\bfX_t)$ $\diff t\times P$-a.e.~by \eqref{eqn:drift condition}. Hence, by Lemma~\ref{lem:HJBI:lower bound}, for each $u \in [0,T]$,
\begin{align}
\label{eqn:lem:lower bound:pf:before expecation}
I^\psi_u(\bfupsilon^\star,P)
&\geq N_u - \int_0^u K_\mathrm{lo}(t,\bfX_t) \dd t \; \psi^2.
\end{align}

By construction, $N$ is a local $P$-martingale starting in $0$. Suppose for the moment that $N$ is also a submartingale. Then by taking expectations under $P$ on both sides of \eqref{eqn:lem:lower bound:pf:before expecation} (for $u = T$), we obtain
\begin{align*}
J^\psi(\bfupsilon^\star,P) - w^\psi_0
&\geq -\norm{K_\mathrm{lo}}_{L^1_\fP} \psi^2
\geq -\frac{1}{2}\varepsilon\psi.
\end{align*}
Combining this with \eqref{eqn:lem:lower bound:pf:10} yields \eqref{eqn:lem:lower bound:pf:to show}.

It remains to show that $N$ is a submartingale under $P$. As it is a local martingale, it suffices to show that it is bounded from above by a $P$-integrable random variable. To this end, first note from the definition of $w^\psi$ in \eqref{eqn:value function}, the fact that $\wt w \geq 0$ on $\bfD$ by Assumption~\ref{ass:main result}~\ref{ass:main result:cash equivalent}, and Assumption~\ref{ass:main result}~\ref{ass:main result:utility} that $w^\psi \leq U(\ol Y)$ on $\bfD\times[\ul Y,\ol Y]$. Clearly, $U'(y)f(\Sigma, \bfzeta)$ is also uniformly bounded over $y \geq \ul Y$, $\Sigma
 \in [\ul\Sigma,\ol\Sigma]$, and $\bfzeta \in \bfZ$. In view of the definition of $I^\psi(\bfupsilon^\star,P)$ in \eqref{eqn:lem:lower bound:pf:Ito}, the fact that  $Y^{\bfupsilon^\star,P} \in [\ul Y,\ol Y]$ $\diff t\times P$-a.e.~by Corollary~\ref{cor:Y:delta-vega}, and Assumption~\ref{ass:main result}~\ref{ass:main result:model set}, we conclude that $I^\psi(\bfupsilon^\star,P) \leq K_I$ $\diff t\times P$-a.e.~for some constant $K_I > 0$. Using this and \eqref{eqn:lem:lower bound:pf:before expecation}, we obtain for each $u\in[0,T]$,
\begin{align*}
N_u
&\leq K_I + \int_0^T K_\mathrm{lo}(t,\bfX_t) \dd t.
\end{align*}
As $K_\mathrm{lo} \in L^1_\fP$, $N$ is bounded from above by a $P$-integrable random variable and therefore is a submartingale. This completes the proof.
\end{proof}

\subsubsection{The asymptotic upper bound for the stochastic differential game}
\label{sec:proofs:upper bound}

To establish an asymptotic upper bound for the stochastic differential game \eqref{eqn:value}, we first prove that the probability under $P^\psi$ that $\Sigma$ leaves $(\ul\Sigma,\ol\Sigma)$ before time $T$ is of order $O(\psi)$.

\begin{proposition}
\label{prop:exit time}
Let $\tau := \inf \lbrace t \in [0,T] : \Sigma_t \not\in (\ul\Sigma,\ol\Sigma)\rbrace \wedge T$ be the first time that $\Sigma$ leaves $(\ul\Sigma,\ol\Sigma)$. Then $\tau$ is a stopping time and there is $K_\tau > 0$ such that for every $\psi \in (0,\psi_0)$,
\begin{align}
\label{eqn:prop:exit time}
P^\psi[\tau < T]
&\leq K_\tau \psi.
\end{align}
\end{proposition}

\begin{proof}
It is an easy exercise to show that $\tau$ is a stopping time for the (non-augmented, non-right-continuous) filtration $\FF$. This uses the fact that all paths of $\Sigma$ are continuous and $(\ul\Sigma,\ol\Sigma)$ is open; cf.~\cite[Problem 2.7 in Chapter 1]{KaratzasShreve1998}.

Turning to the proof of \eqref{eqn:prop:exit time}, by standard estimates for It\^o processes (cf., e.g., \cite[Lemma V.11.5]{RogersWilliams2000}), there is a constant $K > 0$ (depending only on $T$) such that for every $\psi \in (0,\psi_0)$,
\begin{align}
\label{eqn:prop:exit time:pf:10}
\EX[P^\psi]{\sup_{0\leq t \leq T} \vert \Sigma_t - \Sigma_0 \vert^2}
&\leq K \EX[P^\psi]{\int_0^T \left( (\nu_t^{P^\psi})^2 + (\eta_t^{P^\psi})^2 + \xi_t^{P^\psi} \right) \dd t}.
\end{align}
Define $K_\bfzeta \geq 1$ as in Lemma~\ref{lem:HJBI:H1LCQP}~(d), and let $K_0'(t,\bfx) = K_0(t,\bfx) + K_\bfzeta \geq 1$ as well as $K_\tau = 2\ell^{-2}K \norm{K_0'}_{L^2_\fP}^2$ with $\ell := \min(\ol\Sigma - \Sigma_0,\Sigma_0-\ul\Sigma) > 0$. Clearly, $K_0' \in L^4_\fP\subset L^2_\fP$, so that $0 \leq K_\tau < \infty$. Fix $\psi \in (0,\psi_0)$. By \eqref{eqn:pf:estimate:candidate asymptotic model family} and Lemma~\ref{lem:HJBI:H1LCQP}~(d),
\begin{align*}
\enorm{\bfzeta^{P^\psi}_t - \bfzeta^0(\Sigma_t)}
&\leq \enorm{\bfzeta^{P^\psi}_t - \bfzeta^\psi(t,\bfX_t)} + \enorm{\bfzeta^\psi(t,\bfX_t) - \bfzeta^0(\Sigma_t)}
\leq K_0(t,\bfX_t)\psi^2 + K_\bfzeta \psi\\
&\leq K_0'(t,\bfX_t)\psi
\quad \diff t\times P^\psi\text{-a.e.}
\end{align*}
Recalling that $\bfzeta^0(\Sigma) = (0,\Sigma,0,0)^\tr$, this estimate yields
\begin{align}
(\nu_t^{P^\psi})^2 + (\eta_t^{P^\psi})^2 + \xi_t^{P^\psi}
&\leq \enorm{\bfzeta^{P^\psi}_t - \bfzeta^0(\Sigma_t)}^2
+ \enorm{\bfzeta^{P^\psi}_t - \bfzeta^0(\Sigma_t)}
\leq K_0'(t,\bfX_t)^2 \psi^2 + K_0'(t,\bfX_t) \psi\notag\\
\label{eqn:prop:exit time:pf:20}
&\leq 2K_0'(t,\bfX_t)^2 \psi \quad \diff t \times P^\psi\text{-a.e.}
\end{align}
Moreover, by the definition of $\ell$ and Markov's inequality,
\begin{align}
\label{eqn:prop:exit time:pf:30}
P^\psi[\tau < T]
&\leq P^\psi\left[\sup_{0\leq t \leq T} \left\vert \Sigma_t -\Sigma_0 \right\vert^2 \geq \ell^2\right]
\leq \ell^{-2} \EX[P^\psi]{\sup_{0\leq t \leq T} \vert \Sigma_t - \Sigma_0 \vert^2}.
\end{align}
Combining \eqref{eqn:prop:exit time:pf:10}--\eqref{eqn:prop:exit time:pf:30} proves \eqref{eqn:prop:exit time}.
\end{proof}

We are now in a position to establish an asymptotic upper bound for the stochastic differential game \eqref{eqn:value}, which completes the proof of Theorem~\ref{thm:main result} at the beginning of Section~\ref{sec:proofs:main result}.

\begin{lemma}
\label{lem:upper bound}
As $\psi \downarrow 0$,
\begin{align*}
\sup_{\bfupsilon\in\fY} J^\psi(\bfupsilon,P^\psi)
&\leq w^\psi_0 + o(\psi).
\end{align*}
\end{lemma}

\begin{proof}
Set $\bar K = K_0 \in L^4_\fP$. With this choice of $\bar K$, define $K_\mathrm{up} \in L^2_\fP$ as in Lemma~\ref{lem:HJBI:upper bound}. As $U'$ is decreasing and $U$ has decreasing absolute risk aversion (cf.~Assumption~\ref{ass:main result}~\ref{ass:main result:utility}), there is $K_U > 0$ such that (cf.~Assumption~\ref{ass:main result}~\ref{ass:main result:strategy set} for the choice of $K_\fY$)
\begin{align}
\label{eqn:lem:upper bound:pf:utility estimate}
U'(y) \left(1 - U''(y)\big/U'(y) \right)
&\leq K_U, \quad\text{for all }  y \geq -K_\fY.
\end{align}

Now, fix $\varepsilon > 0$, choose $\psi_0' \in (0,\psi_0)$ such that
\begin{align*}
K_U \norm{K_\mathrm{up}}_{L^1_\fP} \psi_0' + K_U \sqrt{T} \norm{K_\mathrm{up}}_{L^2_\fP} K_\tau^{1/2} (\psi_0')^{1/2}
&\leq \frac{1}{2}\varepsilon
\end{align*}
where $K_\tau>0$ is as in Proposition~\ref{prop:exit time}, and let $\psi \in (0,\psi_0')$. We need to show that
\begin{align}
\label{eqn:lem:upper bound:pf:to show}
\sup_{\bfupsilon'\in\fY} J^\psi(\bfupsilon',P^\psi) - w^\psi_0
&\leq \varepsilon \psi.
\end{align}
Choose $\bfupsilon\in\fY$ such that $J^\psi(\bfupsilon,P^\psi) + \frac{1}{2}\varepsilon\psi \geq \sup_{\bfupsilon'\in\fY} J^\psi(\bfupsilon',P^\psi)$. Then
\begin{align}
\label{eqn:lem:upper bound:pf:10}
\sup_{\bfupsilon'\in\fY} J^\psi(\bfupsilon',P^\psi) - w^\psi_0
&\leq J^\psi(\bfupsilon,P^\psi) - w^\psi_0 + \frac{1}{2}\varepsilon\psi.
\end{align}

Applying It\^o's formula (under $P^\psi$) to the process $w^\psi(u,\bfX_u,Y^{\bfupsilon,P^\psi}_u)$ (recall the dynamics of $S,A,M,\Sigma$, and $Y^{\bfupsilon,P^\psi}$ given in \eqref{eqn:covariation:S and IV}--\eqref{eqn:A:dynamics} and Lemma~\ref{lem:Y}) and using the third line in \eqref{eqn:cash equivalent:PDE} (so that the $\diff M$-integral vanishes) yields for each $u\in[0,T]$,
\begin{align}
\label{eqn:lem:upper bound:pf:Ito}
\begin{split}
I^\psi_u(\bfupsilon,P^\psi)
&:=w^\psi(u,\bfX_u,Y^{\bfupsilon,P^\psi}_u) + \frac{1}{\psi}\int_0^u U'(Y^{\bfupsilon,P^\psi}_t) f(\Sigma_t,\bfzeta^{P^\psi}_t) \dd t - w^\psi_0 \\
&= N_u
+ \int_0^u \left( w^\psi_t(t,\bfX_t,Y^{\bfupsilon,P^\psi}_t) + H^\psi(t,\bfX_t,Y^{\bfupsilon,P^\psi}_t;\bfupsilon_t,\bfzeta^{P^\psi}_t)\right) \dd t.
\end{split}
\end{align}
Here,
\begin{align*}
\begin{split}
N
&:= \int_0^\cdot \left(w^\psi_S(t,\bfX_t,Y^{\bfupsilon,P^\psi}_t) + \gamma(t,S_t,A_t,M_t) w^\psi_A(t,\bfX_t,Y^{\bfupsilon,P^\psi}_t) \right) \dd S_t\\
&\qquad+ \int_0^\cdot w^\psi_\Sigma(t,\bfX_t,Y^{\bfupsilon,P^\psi}_t) \dd \Sigma^{\contlocmartpart,P^\psi}_t
+\int_0^\cdot w^\psi_Y(t,\bfX_t,Y^{\bfupsilon,P^\psi}_t)\dd Y^{\contlocmartpart,\bfupsilon,P^\psi}_t
\end{split}
\end{align*}
and $Y^{\contlocmartpart,\bfupsilon,P^\psi} = Y^{\bfupsilon,P^\psi} + \int_0^\cdot b^\cV(t,\bfX_t;\bfzeta^{P^\psi}_t)\dd t$ is the local martingale part of $Y^{\bfupsilon,P^\psi}$ under $P^\psi$. 

We want to use Lemma~\ref{lem:HJBI:upper bound} to estimate the drift term in the last line of \eqref{eqn:lem:upper bound:pf:Ito}. Note that condition \eqref{eqn:lem:HJBI:upper bound:condition} of Lemma~\ref{lem:HJBI:upper bound} with $\bfzeta$ and $\bfx$ replaced by $\bfzeta^{P^\psi}_t$ and $\bfX_t$, respectively, is fulfilled $\diff t \times P^\psi$-a.e.~by \eqref{eqn:pf:estimate:candidate asymptotic model family} and our choice of $\bar K = K_0$. Moreover, denoting by $\tau$ the first time that $\Sigma$ leaves $(\ul\Sigma,\ol\Sigma)$ (cf.~Proposition~\ref{prop:exit time}), we have for each $t\in[0,T)$ that $\Sigma_t \in (\ul\Sigma,\ol\Sigma)$ $P^\psi$-a.s.~on $\lbrace \tau \geq T \rbrace$. Therefore, by \eqref{eqn:lem:HJBI:upper bound}, \eqref{eqn:lem:upper bound:pf:utility estimate}, and Assumption~\ref{ass:main result}~\ref{ass:main result:strategy set}, for each $u \in [0,T]$,
\begin{align}
\label{eqn:lem:upper bound:pf:before expecation}
I^\psi_u(\bfupsilon,P^\psi)
&\leq N_u + \int_0^u K_\mathrm{up}(t,\bfX_t) U'(Y^{\bfupsilon,P^\psi}_t)\left(1+\frac{-U''(Y^{\bfupsilon,P^\psi}_t)}{U'(Y^{\bfupsilon,P^\psi}_t)}\right) \dd t \;
\left(\psi^2 \1_{\lbrace \tau \geq T \rbrace} + \psi \1_{\lbrace \tau < T \rbrace}\right)\notag\\
&\leq N_u + K_U \int_0^u K_\mathrm{up}(t,\bfX_t) \dd t \;
\left(\psi^2 \1_{\lbrace \tau \geq T \rbrace} + \psi \1_{\lbrace \tau < T \rbrace}\right).
\end{align}

By construction, $N$ is a local $P^\psi$-martingale starting in $0$. Suppose for the moment that $N$ is also a supermartingale. Then by taking expectations under $P^\psi$ on both sides of \eqref{eqn:lem:upper bound:pf:before expecation} (for $u=T$) and using the Cauchy--Schwarz and Jensen inequalities as well as Proposition~\ref{prop:exit time}, we obtain
\begin{align*}
J^\psi(\bfupsilon,P^\psi)-w^\psi_0
&\leq K_U \norm{K_\mathrm{up}}_{L^1_\fP} \psi^2 + K_U \EX[P^\psi]{\int_0^T K_\mathrm{up}(t,\bfX_t) \dd t \; \1_{\lbrace \tau < T\rbrace}}\psi\\
&\leq K_U \norm{K_\mathrm{up}}_{L^1_\fP} \psi^2 + K_U \EX[P^\psi]{\left(\int_0^T K_\mathrm{up}(t,\bfX_t) \dd t\right)^2}^{1/2} P^\psi[\tau < T]^{1/2} \psi\\
&\leq \left(K_U \norm{K_\mathrm{up}}_{L^1_\fP} \psi + K_U \sqrt{T} \norm{K_\mathrm{up}}_{L^2_\fP} K_\tau^{1/2} \psi^{1/2}\right)\psi
\leq \frac{1}{2}\varepsilon \psi.
\end{align*}
Combining this with \eqref{eqn:lem:upper bound:pf:10} yields \eqref{eqn:lem:upper bound:pf:to show}.

It remains to show that $N$ is a supermartingale under $P^\psi$. As it is a local martingale, it suffices to show that it is bounded from below by a $P^\psi$-integrable random variable. To this end, first note from the definition of $w^\psi$ in \eqref{eqn:value function} and Assumption~\ref{ass:main result}~\ref{ass:main result:cash equivalent} and \ref{ass:main result:utility} that for every $(t,\bfx)\in\bfD$ and $y \geq -K_\fY$,
\begin{align*}
w^\psi(t,\bfx,y)
&= U(y) - U'(y)\wt w(t,\bfx)\psi
\geq U(-K_\fY) - U'(-K_\fY) K_{\wt w}.
\end{align*}
By Assumption~\ref{ass:main result}~\ref{ass:main result:strategy set} and the fact that $f \geq 0$, we obtain for each $u\in[0,T]$,
\begin{align*}
I^\psi_u(\bfupsilon,P^\psi)
&\geq U(-K_\fY) - U'(-K_\fY) K_{\wt w} - w^\psi_0
=: K_I.
\end{align*}
Using this and \eqref{eqn:lem:upper bound:pf:before expecation} yields for each $u\in[0,T]$,
\begin{align*}
N_u
&\geq K_I - K_U \int_0^T K_\mathrm{up}(t,\bfX_t) \dd t.
\end{align*}
As $K_\mathrm{up}\in L^2_\fP$, $N$ is bounded from below by a $P^\psi$-integrable random variable and therefore is a supermartingale. This completes the proof.
\end{proof}

\subsection{Construction of a modified feedback control}
\label{sec:proofs:existence}

\begin{proof}[Proof of Lemma~\ref{lem:existence}]
For each $\psi > 0$, define the functions $\wh\nu,\check\bfzeta^{\psi*},\check\bfzeta^\psi:\bfD^0\to\RR$ by
\begin{align}
\label{eqn:lem:existence:pf:nu hat}
\wh\nu(t,\bfx)
&= -\frac{1}{2 \cC_\Sigma}
\begin{pmatrix}
\wt\sigma\\
\wt\eta
\end{pmatrix}^\tr
\begin{pmatrix}
S^2 \cC_{SS}    & S \cC_{S\Sigma}\\
S \cC_{S\Sigma} & \cC_{\Sigma\Sigma}
\end{pmatrix}
\begin{pmatrix}
\wt\sigma\\
\wt\eta
\end{pmatrix},\\
\check\bfzeta^{\psi*}(t,\bfx)
&= \bfzeta^0(\Sigma) + \wt\bfzeta \psi + \wh\nu \vec\bfe_1 \psi^2,\notag\\
\label{eqn:lem:existence:pf:modified candidate}
\check\bfzeta^\psi(t,\bfx)
&= \bfzeta^0(\Sigma) + (\wt\bfzeta \psi + \wh\nu \vec\bfe_1\psi^2)\1_{\lbrace \ul\Sigma<\Sigma<\ol\Sigma \rbrace}
= \bfzeta^\psi + \wh\nu \vec\bfe_1\psi^2 \1_{\lbrace \ul\Sigma<\Sigma<\ol\Sigma \rbrace},
\end{align}
where the functions $\wt\sigma$ and $\wt\eta$ are the second and third component of $\wt\bfzeta$ defined in \eqref{eqn:candidate control:first order}. That is, $\check\bfzeta^\psi$ arises from $\bfzeta^\psi$ (cf.~\eqref{eqn:candidate control:short}) by a perturbation of the first component (the drift of the implied volatility) by a term of order $O(\psi^2)$.

First, we show the asserted continuity of $\check\bfzeta^\psi$ and the extension property. It is easy to see from Assumption~\ref{ass:main result}~\ref{ass:main result:call} and \ref{ass:main result:option} that the vega-gamma-vanna-volga vectors $\bfc$ and $\bfv$ are continuous on $\bfD^0$. Then also $\lambda$ (it is not hard to show that the two expressions on the right-hand side of \eqref{eqn:lagrange multiplier:lambda} coincide whenever $\cV_{\Sigma\Sigma} - \frac{\bfc^\tr\Psi\bfv}{\bfc^\tr\Psi\bfc} \cC_{\Sigma\Sigma} = 0$), $\mu$, and hence $\wt\bfzeta$ are continuous on $\bfD^0$. Therefore, also $\wh\nu$ and $\check\bfzeta^{\psi*}$ are continuous on $\bfD^0$, and it follows that $\check\bfzeta^\psi$ is continuous on $(0,T)\times\bfG\times(\ul\Sigma,\ol\Sigma)$. By construction, $\check\bfzeta^{\psi*}$ is a continuous extension of $\left.\check\bfzeta^\psi\right\vert_{(0,T)\times\bfG\times(\ul\Sigma,\ol\Sigma)}$ to $\bfD^0$.

Second, we show that the range of $\check\bfzeta^\psi$ is contained in $\bfZ = [\ul\nu,\ol\nu]\times[\ul\sigma,\ol\sigma]\times[\ul\eta,\ol\eta]\times[0,\ol\xi]$ for sufficiently small $\psi \in (0,\psi_0)$. To this end, it suffices to show that $\wt\bfzeta$ and $\wh\nu$ are bounded on $\bfD = (0,T)\times\bfG\times[\ul\Sigma,\ol\Sigma]$. First, by \eqref{eqn:lem:HJBI:H1LCQP:pf:zeta tilde:norm bound} from the proof of Lemma~\ref{lem:HJBI:H1LCQP}~(a), we have
\begin{align*}
\enorm{\wt\bfzeta(t,\bfx)}
&\leq \psi_{\max} \enorm{\bfv(t,\bfx)},\quad (t,\bfx) \in \bfD.
\end{align*}
By Assumption~\ref{ass:main result}~\ref{ass:main result:option}, $\enorm{\bfv}$ is bounded on $\bfD$, so that $\wt\bfzeta$ is bounded on $\bfD$ as well. Second, using the boundedness of $\wt\bfzeta$ as well as \eqref{eqn:lem:existence:C}, it follows  from \eqref{eqn:lem:existence:pf:nu hat} that $\wh\nu$ is bounded on $\bfD$ as well. We conclude that there is $\psi_0 \in (0,1)$ such that $\check\bfzeta^\psi(t,\bfx) \in \bfZ$ for each $(t,\bfx)\in\bfD$ and $\psi \in (0,\psi_0)$. 

Third, we prove part~(c). By \eqref{eqn:lem:existence:pf:modified candidate}, for each $(t,\bfx)\in\bfD$ and $\psi > 0$, we have
\begin{align*}
\enorm{\check\bfzeta^\psi(t,\bfx) - \bfzeta^\psi(t,\bfx)}
&\leq \vert \wh\nu(t,\bfx) \vert \psi^2.
\end{align*}
As we have argued above that $\wh\nu$ is bounded on $\bfD$, there is $K_0 > 0$ such that \eqref{eqn:lem:existence:estimate} holds.

Fourth, we show part~(b), i.e., that $b^\cC(t,\bfx;\check\bfzeta^\psi(t,\bfx)) = 0$. Fix $(t,\bfx) \in \bfD$. If $\Sigma \in \lbrace \ul\Sigma,\ol\Sigma \rbrace$, then $\check\bfzeta^\psi(t,\bfx) = \bfzeta^0(\Sigma)$ by construction and the assertion is trivial. So suppose that $\Sigma \in (\ul\Sigma,\ol\Sigma)$. By the representation \eqref{eqn:bC:matrix form} of $b^\cC$, noting that the second and third components of $\check\bfzeta^\psi$ coincide with those of $\bfzeta^\psi$, we have
\begin{align*}
b^\cC(\check\bfzeta^\psi)
&= \bfc^\tr \left(\check\bfzeta^\psi - \bfzeta^0(\Sigma) \right) + \frac{1}{2}
\begin{pmatrix}
\sigma^\psi - \Sigma\\
\eta^\psi
\end{pmatrix}^\tr
\begin{pmatrix}
S^2 \cC_{SS}    & S \cC_{S\Sigma}\\
S \cC_{S\Sigma} & \cC_{\Sigma\Sigma}
\end{pmatrix}
\begin{pmatrix}
\sigma^\psi - \Sigma\\
\eta^\psi
\end{pmatrix}.
\end{align*}
We note that $\bfc^\tr \left(\bfzeta^\psi - \bfzeta^0(\Sigma)\right) = \bfc^\tr\wt\bfzeta \psi = 0$ by Lemma~\ref{lem:HJBI:H1LCQP}~(a), so that $\bfc^\tr\left( \check\bfzeta^\psi - \bfzeta^0(\Sigma) \right) = \bfc^\tr\left( \check\bfzeta^\psi - \bfzeta^\psi\right)$. Using this together with the fact that $\sigma^\psi = \Sigma + \wt\sigma \psi$ and $\eta^\psi = \wt\eta \psi$ (because $\Sigma \in (\ul\Sigma,\ol\Sigma)$) as well as the definition of $\wh\nu$ yields
\begin{align*}
b^\cC(\check\bfzeta^\psi)
&= \bfc^\tr \left(\check\bfzeta^\psi - \bfzeta^\psi \right) - \cC_\Sigma \wh \nu \psi^2.
\end{align*}
Now by \eqref{eqn:lem:existence:pf:modified candidate} and the definition of $\bfc$ (cf.~\eqref{eqn:VGVV:c}), $\bfc^\tr (\check\bfzeta^\psi - \bfzeta^\psi) = \cC_\Sigma \wh\nu \psi^2$, so that $b^\cC(\check\bfzeta^\psi) = 0$.

Finally, part~(a) follows immediately from the construction of $\check\bfzeta^\psi$.
\end{proof}

\appendix

\section{Linearly constrained quadratic programming}
\label{sec:quadratic programming}

\paragraph{Lagrangian duality.}

We recall some basic Lagrange duality results from \cite[Section 5.1.5]{Bertsekas1999}. Fix $n \in \NN$ and functions $f,g,h:\RR^n \to \RR$. We refer to the problem
\begin{align}
\label{eqn:duality:primal problem}
\text{minimise } f(\bfz)
\quad\text{subject to }
\bfz\in\RR^n,
h(\bfz) = 0,
g(\bfz) \leq 0,
\end{align}
as the \emph{primal problem} and denote by
\begin{align*}
f^*
&= \inf \lbrace f(\bfz) : \bfz \in \RR^n, h(\bfz) = 0, g(\bfz) \leq 0 \rbrace
\end{align*}
its optimal value. The corresponding \emph{Lagrangian} is
\begin{align*}
L(\bfz, \lambda, \mu)
&= f(\bfz) + \lambda h(\bfz) + \mu g(\bfz), \quad \bfz \in \RR^n, \lambda,\mu \in \RR,
\end{align*}
and a pair $(\mu^*, \lambda^*)$ is called a \emph{Lagrange multiplier} if
\begin{align*}
f^*
&= \inf_{\bfz \in \RR^n} L(\bfz,\lambda^*,\mu^*)
\quad\text{and}\quad
\mu^* \geq 0.
\end{align*}
The \emph{dual problem} for \eqref{eqn:duality:primal problem} is
\begin{align*}
\text{maximise } q(\lambda,\mu)
\quad\text{subject to }
\lambda \in \RR,
\mu \geq 0,
\end{align*}
where the \emph{dual function} $q$ is 
\begin{align*}
q(\lambda,\mu)
&= \inf_{\bfz \in \RR^n} L(\bfz,\lambda,\mu), \quad \lambda,\mu \in \RR.
\end{align*}
Finally, $q^* = \sup\lbrace q(\lambda,\mu):\lambda\in\RR, \mu \geq0 \rbrace$ denotes the optimal value of the dual problem.

\paragraph{A quadratic programming problem with linear equality and inequality constraints.}
The following lemma provides the solution to a primal problem with a strictly convex quadratic cost function and specific linear equality and inequality constraints.

\begin{lemma}
\label{lem:LCQP}
Fix $n \in \NN$, a diagonal matrix $D = \diag(d_1,\ldots,d_n) \in\RR^{n\times n}$ with positive diagonal entries, and vectors $\bfv = (v_1,\ldots,v_n)^\tr$ and $\bfc = (c_1,\ldots, c_n)^\tr$ in $\RR^n$ such that $c_i \neq 0$ for some $i \in \lbrace 1,\ldots,n-1\rbrace$. Moreover, set $\1_A = 1$ if $v_n - \frac{\bfc^\tr D^{-1}\bfv}{\bfc^\tr D^{-1}\bfc} c_n < 0$ and $\1_A = 0$ otherwise, and define
\begin{align}
\label{eqn:lem:LCQP:lambda}
\lambda^*
&= \frac{\bfc^\tr D^{-1}\bfv - c_n v_n d_n^{-1} \1_A}{\bfc^\tr D^{-1}\bfc - c_n^2 d_n^{-1}\1_A},\\
\label{eqn:lem:LCQP:mu}
\mu^*
&= (v_n - \lambda^* c_n)^-,\\
\label{eqn:lem:LCQP:primal minimiser}
\bfz^*
&= D^{-1}(\bfv - \lambda^* \bfc + \mu^* \vec\bfe_n).
\end{align}

\begin{enumerate}
\item $\bfz^*$ is the unique optimiser of the primal problem
\begin{align}
\label{eqn:lem:LCQP:primal problem}
\text{minimise }\frac{1}{2}\bfz^\tr D \bfz - \bfv^\tr \bfz
\quad\text{ subject to }
\bfz\in\RR^n,
\bfc^\tr \bfz = 0,
z_n \geq 0,
\end{align}
and satisfies the bound $\enorm{\bfz^*} \leq d_{\min}^{-1} \enorm{\bfv}$, where $d_{\min} = \min(d_1,\ldots,d_n)$. 

\item $(\lambda^*,\mu^*)$ is the unique optimiser of the dual problem for \eqref{eqn:lem:LCQP:primal problem}, which can be written as
\begin{align}
\label{eqn:lem:LCQP:dual problem}
\text{maximise } -\frac{1}{2}(\bfv - \lambda \bfc + \mu \vec\bfe_n)^\tr D^{-1}(\bfv - \lambda \bfc + \mu \vec\bfe_n)
\quad\text{subject to }
\lambda \in \RR, \mu \geq 0.
\end{align}
The optimiser satisfies the bound
\begin{align}
\label{eqn:lem:LCQP:Lagrange multiplier:norm bound}
\enorm{\lambda^* \bfc - \mu^* \vec\bfe_n}
&\leq \left(1 + \frac{d_{\max}}{d_{\min}}\right)\enorm{\bfv},
\end{align}
where $d_{\max} = \max(d_1,\ldots,d_n)$.

\item The optimal values of the primal and dual problems coincide (i.e., there is \emph{no duality gap}) and equal
\begin{align*}
-\frac{1}{2}\bfv^\tr \bfz^*.
\end{align*}

\item The triplet $(\bfz^*,\mu^*,\lambda^*)$ satisfies the optimality conditions
\begin{align}
\label{eqn:lem:LCQP:conditions}
\bfz^* 
&= \arg\min_{\bfz\in\RR^n} L(\bfz,\lambda^*,\mu^*),\quad
\bfc^\tr \bfz^* = 0,\quad
z^*_n \geq 0,\quad
\mu^* \geq 0,\quad
\mu^*z_n^* = 0,
\end{align}
where $L$ is the Lagrangian corresponding to the primal problem. Moreover, $(\lambda^*,\mu^*)$ is a Lagrange multiplier for the primal problem.
\end{enumerate}
\end{lemma}

\begin{proof}
First of all, note that the Lagrangian
\begin{align*}
L(\bfz,\lambda,\mu)
&= \frac{1}{2}\bfz^\tr D \bfz - \bfv^\tr \bfz + \lambda \bfc^\tr \bfz - \mu \vec\bfe_n^\tr \bfz
\end{align*}
corresponding to the primal problem is strictly convex over $\bfz \in \RR^n$. Hence, the dual function $q(\lambda,\mu)= \inf_{\bfz\in\RR^n} L(\bfz,\lambda,\mu)$ can be computed explicitly by substituting the solution $\bfz'$ to the first-order condition $\bfNull = \D_\bfz L(\bfz',\lambda,\mu) = D \bfz' - \bfv + \lambda \bfc - \mu \vec\bfe_n$ back into $L(\bfz',\mu,\lambda)$. This yields
\begin{align*}
q(\lambda,\mu)
&= -\frac{1}{2}(\bfv - \lambda \bfc + \mu \vec\bfe_n)^\tr D^{-1}(\bfv - \lambda \bfc + \mu \vec\bfe_n)
\end{align*}
and thus the dual problem takes the form \eqref{eqn:lem:LCQP:dual problem}.

The crucial part of the proof is to show that the triplet $(\bfz^*,\lambda^*,\mu^*)$ satisfies the optimality conditions \eqref{eqn:lem:LCQP:conditions}. As $L$ is strictly convex over $\bfz \in \RR^n$, the optimality conditions are equivalent to
\begin{align}
\label{eqn:lem:LCQP:pf:conditions}
D\bfz^* - \bfv + \lambda^* \bfc - \mu^* \vec \bfe_n = 0,\quad
\bfc^\tr \bfz^* = 0,\quad
z^*_n \geq 0,\quad
\mu^* \geq 0,\quad
\mu^* z^*_n = 0.
\end{align}
Recall the definitions of $\lambda^*$, $\mu^*$, and $\bfz^*$ in \eqref{eqn:lem:LCQP:lambda}--\eqref{eqn:lem:LCQP:primal minimiser} and note that the assumption that $c_i \neq 0$ for some $i\in\lbrace1,\ldots,n-1\rbrace$ together with the positive definiteness of $D^{-1}$ ensures that $\lambda^*$ is well defined. The stationarity condition $D\bfz^* - \bfv + \lambda^* \bfc - \mu^* \vec \bfe_n = 0$ holds by definition of $\bfz^*$. For the other conditions, we distinguish two cases.
First, suppose that $\1_A = 0$, i.e., $v_n -  \frac{\bfc^\tr D^{-1}\bfv}{\bfc^\tr D^{-1} \bfc} c_n \geq 0$. Then $\lambda^* = \frac{\bfc^\tr D^{-1}\bfv}{\bfc^\tr D^{-1} \bfc}$, $\mu^* = 0$, and $z^*_n = d_n^{-1}(v_n - \lambda^* c_n) \geq 0$. Moreover, $\bfc^\tr \bfz^* = \bfc^\tr D^{-1} \bfv - \lambda^* \bfc^\tr D^{-1} \bfc = 0$.
Second, suppose that $\1_A = 1$, i.e., $v_n -  \frac{\bfc^\tr D^{-1}\bfv}{\bfc^\tr D^{-1} \bfc} c_n < 0$, or, equivalently (multiply by $\bfc^\tr D^{-1}\bfc$, add and subtract $c_n^2 v_n d_n^{-1}$, and then divide by $\bfc^\tr D^{-1} \bfc - c_n^2 d_n^{-1} > 0$), $v_n -\lambda^* c_n < 0$. Then $\mu^* > 0$ and $z^*_n = 0$ by definition of $\mu^*$ and $\bfz^*$. Finally, setting $c = \bfc^\tr D^{-1} \bfc$ and $v = \bfc^\tr D^{-1} \bfv$ for brevity,
\begin{align*}
\bfc^\tr\bfz^*
&= \bfc^\tr D^{-1}(\bfv - \lambda^*\bfc+ \mu^*\vec\bfe_n)
= v -\lambda^* c +\mu^* c_n d_n^{-1}
= v -\lambda^* c - (v_n - \lambda^* c_n) c_n d_n^{-1}\\
&= v - \lambda^*(c - c_n^2 d_n^{-1}) - c_n v_n d_n^{-1}
= v - (v - c_n v_n d_n^{-1}) - c_n v_n d_n^{-1}
= 0.
\end{align*}
So, \eqref{eqn:lem:LCQP:pf:conditions} holds in both cases. By the characterisation of primal optimal solutions \cite[Proposition~5.1.5]{Bertsekas1999}, this implies that $\bfz^*$ is an optimiser for the primal problem, that $(\lambda^*,\mu^*)$ is a Lagrange multiplier, and that there is no duality gap. Moreover, $(\lambda^*,\mu^*)$ is an optimiser for the dual problem by a corollary \cite[Proposition~5.1.4~(a)]{Bertsekas1999} of the weak duality theorem \cite[Proposition~5.1.3]{Bertsekas1999}. As the primal and dual problems are strictly convex and strictly concave, respectively, the optimisers are unique.

Plugging the optimiser $(\lambda^*,\mu^*)$ of the dual problem into the cost function of the dual problem \eqref{eqn:lem:LCQP:dual problem} and using the definition of $\bfz^*$, the optimal value $q^*$ (of both the primal and the dual problem) reads
\begin{align*}
q^*&:=-\frac{1}{2}(\bfv - \lambda^* \bfc + \mu^* \vec\bfe_n)^\tr D^{-1}(\bfv - \lambda^* \bfc + \mu \vec\bfe_n)
= -\frac{1}{2} (\bfv - \lambda^* \bfc + \mu^* \vec\bfe_n)^\tr \bfz^*.
\end{align*}
Now, note that $\bfc^\tr\bfz^* = 0$ and $\mu^* \vec\bfe_n^\tr \bfz^* = \mu^* z_n^* = 0$ by \eqref{eqn:lem:LCQP:pf:conditions}. Hence, $q^* = -\frac{1}{2}\bfv^\tr \bfz^*$.

Next, using that $\bfz^*$ achieves the optimal value $-\frac{1}{2}\bfv^\tr\bfz^*$ for the primal problem and applying the Cauchy--Schwarz inequality, we obtain
\begin{align*}
\frac{1}{2}d_{\min} \enorm{\bfz^*}^2
&\leq \frac{1}{2}(\bfz^*)^\tr D \bfz^*
= \frac{1}{2}\bfv^\tr\bfz^*
\leq \frac{1}{2} \enorm{\bfv} \enorm{\bfz^*}.
\end{align*}
This yields the last claim of part~(a). Finally, using \eqref{eqn:lem:LCQP:primal minimiser}, the triangle inequality, and the bound $\enorm{\bfz^*} \leq d_{\min}^{-1}\enorm{\bfv}$ which we just proved, we obtain
\begin{align*}
\enorm{\lambda^* \bfc - \mu^* \vec\bfe_n}
&= \enorm{D\bfz^* - \bfv}
\leq \enorm{D\bfz^*} + \enorm{\bfv}
\leq d_{\max} \enorm{\bfz^*} + \enorm{\bfv}
\leq \left(1+\frac{d_{\max}}{d_{\min}}\right) \enorm{\bfv}.
\end{align*}
This proves the last claim of part~(b) and thereby concludes the proof.
\end{proof}

\small
\providecommand{\bysame}{\leavevmode\hbox to3em{\hrulefill}\thinspace}
\providecommand{\MR}{\relax\ifhmode\unskip\space\fi MR }
\providecommand{\MRhref}[2]{%
  \href{http://www.ams.org/mathscinet-getitem?mr=#1}{#2}
}
\providecommand{\href}[2]{#2}

\end{document}